\documentclass[12pt,openright]{crofts}
\title{Efficient Method for Detection of Periodic Orbits in Chaotic Maps and Flows}
\author{Jonathan J.~Crofts}

\departement{Mathematics}
\university{University of Leicester}

\begin{document}
\maketitle

\begin{romanpages}

\chapter*{Acknowledgements}

I would like to thank Ruslan Davidchack, my supervisor, for his many
suggestions and constant support and understanding during this
research. I am also thankful to Michael Tretyakov for his support
and advice. Further, I would like to acknowledge my gratitude to the
people from the Department of Mathematics at the University of
Leicester for their help and support.

Finally, I would like to thank my family and friends for their
patience and support throughout the past few years. In particular, I
thank my wife Lisa and my daughter Ellen, without whom I would have
completed this research far quicker, but somehow, it just would not
have been the same. At this point I would also like to reassure Lisa
that I will get a real job soon.\\

\noindent
Leicester, Leicestershire, UK \hspace{6cm} Jonathan J.~Crofts\\
31 March 2007

\setlength{\baselineskip}{0.85\baselineskip} \tableofcontents[]

\pagestyle{empty}

\chapter*{Abstract} \addcontentsline{toc}{chapter}{Abstract}

An algorithm for detecting unstable periodic orbits in chaotic
systems~[Phys. Rev. E, 60 (1999), pp.~6172--6175] which combines the
set of stabilising transformations proposed by Schmelcher and
Diakonos~[Phys. Rev. Lett., 78 (1997), pp.~4733--4736] with a
modified semi-implicit Euler iterative scheme and seeding with
periodic orbits of neighbouring periods, has been shown to be highly
efficient when applied to low-dimensional system. The difficulty in
applying the algorithm to higher dimensional systems is mainly due
to the fact that the number of stabilising transformations grows
extremely fast with increasing system dimension. In this thesis, we
construct stabilising transformations based on the knowledge of the 
stability matrices of already detected periodic orbits (used as seeds).  
The advantage of our approach is in a substantial reduction of the 
number of transformations, which increases the efficiency of the 
detection algorithm, especially in the case of high-dimensional 
systems. The dependence of the number of transformations on the 
dimensionality of the unstable manifold rather than on system size 
enables us to apply, for the first time, the method of stabilising 
transformations to high-dimensional systems. Another important aspect 
of our treatment of high-dimensional flows is that we do not restrict 
to a Poincar\'{e} surface of section. This is a particularly nice 
feature, since the correct placement of such a section in a 
high-dimensional phase space is a challenging problem in itself. The 
performance of the new approach is illustrated by its application to 
the four-dimensional kicked double rotor map, a six-dimensional system 
of three coupled H\'{e}non maps and to the Kuramoto-Sivashinsky 
system in the weakly turbulent regime.

\pagestyle{fancy}
\addcontentsline{toc}{chapter}{List of figures}\listoffigures

\pagestyle{fancy}
\addcontentsline{toc}{chapter}{List of tables} \listoftables

\end{romanpages}

\fancypagestyle{plain}{
\fancyhead{}\renewcommand{\headrulewidth}{0pt}}

\chapter{Introduction}
\label{ch:Intro}
\begin{quote}
The successes of the differential equation paradigm were impressive
and extensive. Many problems, including basic and important ones,
led to equations that could be solved. A process of self-selection
set in, whereby equations that could not be solved were
automatically of less interest than those that could.\\
\emph{I.~Stewart}
\end{quote}

In this chapter we start in \S\ref{sec:His} by giving a brief primer
into the theory of dynamical systems. Here our intention is not to
give an exhaustive review (concise reviews on the subject are given
in~\cite{Grebogi87,Helleman80,Ott81}). Rather, it is to illustrate
the role played by periodic orbits in the development of the theory.
In \S\ref{sec:pot} a brief introduction to the {\em periodic orbit
theory} is provided, followed by a discussion concerning the
efficient detection of unstable periodic orbits (UPOs). Section
\ref{sec:extend} looks at the application to high-dimensional
systems, in particular, large nonequilibrium systems that are
extensively chaotic. It is well known that numerical methods can
both introduce spurious chaos, as well as suppress
it~\cite{Corless91,Yamaguti81}. Thus in \S\ref{sec:numerics} we
discuss some of the numerical issues which can arise when detecting
UPOs\newabb{UPO} for a chaotic dynamical system. We give an overview
of the objectives of this thesis in \S\ref{sec:overview}. The final
section, \ref{sec:results}, details the contribution to the
literature of this thesis.

\section{History, theory and applications}\label{sec:His}
Although the subject of \emph{modern} dynamical systems has seen an
explosion of interest in the past thirty years -- mainly due to the
advent of the digital computer -- its roots firmly belong at the
foot of the twentieth century. Partly motivated by his work on the
famous three body problem, the French mathematician and philosopher
Henri Jules Poincar\'{e} was to revolutionise the study of nonlinear
differential equations.

Since the birth of the calculus, differential equations have been
studied both in their own right and for modeling phenomena in the
natural sciences. Indeed, Newton considered them in his work on
differential calculus~\cite{Newton}\footnote{Newton's De Methodis
Serierum et Fluxionum was written in 1671 but Newton failed to get
it published and it did not appear in print until John Colson
produced an English translation in 1736.} as early as $1671$. One of
the earliest examples of a first order equation considered by Newton
was
\begin{equation}
\frac{dy}{dx} = 1-3x+y+x^2+xy. \label{eqn:newtode}
\end{equation}
A solution of this equation for the initial condition $y(0) = 0$ can
be obtained as follows: start with
$$y=0+\cdots$$
and insert this into Eq.~(\ref{eqn:newtode}); integrating yields
$$y=x+\cdots,$$
repeating the process with the new value of $y$ gives
$$y=x-x^2+\cdots.$$
One can imagine continuing this process {\em ad infinitum}, leading
to the following solution of Eq.~(\ref{eqn:newtode})
$$y=x-x^2+\frac{1}{3}x^3-\frac{1}{6}x^4+\frac{1}{30}x^5-\frac{1}{45}x^6+\cdots$$
(for further details see~\cite{HairerBook}).

The preceding example demonstrates one of the main differences
between the classical study of differential equations and the
current mindset. The classical study of nonlinear equations was
local, in the sense that individual solutions where sought after.
Most attempts in essence, involved either an approximate series
solution or determining a transformation under which the equation
was reduced either to a known function or to quadrature.

In his work on celestial mechanics~\cite{poincare}, Poincar\'{e}
developed many of the ideas underpinning modern dynamical systems.
By working with sets of initial conditions rather than individual
solutions, he was able to prove that complicated orbits existed for
which no closed solution was available; Poincar\'{e} had caught a
glimpse of what is popularly coined ``chaos'' nowadays.

Although there was continued interest from the mathematical
community; most notably Birkhoff in the 1920s and the Soviet
mathematicians in the 1940s -- Kolmogorov and students thereof -- it
was not until the 1960s that interest from the general scientific
community was rekindled. In 1963 the meteorologist Edward N.~Lorenz
published his now famous paper ``deterministic nonperiodic
flow''~\cite{Lorenz63} where a simple system describing cellular
convection was shown to exhibit extremely complicated dynamics.
Motion was bounded, displayed sensitivity to initial conditions and
was aperiodic; Lorenz had witnessed the first example of a chaotic
attractor.

Around the same time, the mathematician Steve Smale was using
methods from differential topology in order to prove the existence
of a large class of dynamical systems (the so called axiom-A
systems), which were both chaotic and structurally stable at the
same time~\cite{Smale67}. Along with examples such as the Lorenz
model above, scientists where lead to look beyond equilibrium points
and limit cycles in the study of dynamical processes. It became
clear that far from being a mathematical oddity, the chaotic
evolution displayed by many dynamical systems was of great practical
importance.

Today the study of chaotic evolution is widespread throughout the
sciences where the tools of nonlinear analysis are used extensively.
There remain many open questions and the theory of dynamical
systems has a bright and challenging future. The prediction and
control~\cite{Ott90,Controlbook99} of deterministic chaotic systems
is an important area which has received a lot of attention over the
past decade, whilst the extension of the theory to partial
differential equations~\cite{RobinsonBook2,TemamBook} promises to
give fresh insight into the modeling of fully developed turbulence.
However, perhaps the most promising area of future research lies in
the less mathematically minded disciplines such as biology,
economics and the social sciences, to name a few.

\section{Periodic orbits}\label{sec:pot} Periodic orbits play an important role in
the analysis of various types of dynamical systems.  In systems with
chaotic behaviour, unstable periodic orbits form a ``skeleton'' for
chaotic trajectories~\cite{Cvitanovic91}.  A well regarded
definition of chaos~\cite{DevaneyBook} requires the existence of an
infinite number of UPOs that are dense in the chaotic set. Different
geometric and dynamical properties of chaotic sets, such as natural
measure, Lyapunov exponents, fractal dimensions,
entropies~\cite{OttBook}, can be determined from the location and
stability properties of the embedded UPOs.  Periodic orbits are
central to the understanding of quantum-mechanical properties of
nonseparable systems: the energy level density of such systems can
be expressed in a semiclassical approximation as a sum over the UPOs
of the corresponding classical system~\cite{GutzwillerBook}.
Topological description of a chaotic attractor also benefits from
the knowledge of periodic orbits.  For example, a large set of
periodic orbits is highly constraining to the symbolic dynamics and
can be used to extract the location of a generating
partition~\cite{Davidchack00a,Plumecoq00a}.  The significance of
periodic orbits for the experimental study of dynamical systems has
been demonstrated in a wide variety of systems~\cite{Lathrop89},
especially for the purpose of controlling chaotic
dynamics~\cite{Ott90} with possible application in
communication~\cite{Bollt97}.

\subsection{Periodic orbit theory} Briefly put, the {\em periodic
orbit theory} provides a machinery which enables us to use the
knowledge provided by the properties of individual solutions, such
as their periods, location and stabilities, to make predictions
about statistics, e.g. Lyapunov exponents, entropies, and so on. The
dynamical systems to be discussed in this section are smooth
$n$-dimensional maps of the form $x_{i+1} = f(x_i)$, where $x_i$ is
an $n$-dimensional vector in the $n$-dimensional phase space of the
system.

Now, in order for the results to be quoted to hold, we assume that
the attractor of $f$ is both hyperbolic and mixing. A {\em
hyperbolic attractor} is one for which the following two conditions
hold: (i) there exist stable and unstable manifolds at each point of
the attractor whose dimensions, $n_s$ and $n_u$, are the same for
each point on the attractor, with $n_s + n_u = n$,  and (ii) there
exists a constant $K>1$ such that for all points, $x$, on the
attractor, if a vector $u$ is chosen tangent to the unstable
manifold, then
\begin{equation}\label{eqn:stab}
    ||Df(x)u||\geq K||u||,
\end{equation}
and if $u$ is chosen tangent to the stable manifold
\begin{equation}\label{eqn:unstab}
    ||Df(x)u||\leq ||u||/K.
\end{equation}
Here $Df(x)$ denotes the Jacobian matrix of the map $f$ evaluated at
the point $x$. By {\em mixing} we mean that for any two subsets
$A_1$, $A_2$ in the phase space, we have
\begin{equation}\label{eqn:mix}
    \lim_{i\to\infty}\mu[A_1\cap f^i(A_2)] = \mu(A_1)\mu(A_2),
\end{equation}
where $\mu$ is the {\em natural measure} of the attractor. In other
words, the system will evolve over time so that any given open set
in phase space will eventually overlap any other given region.

Let us denote the magnitudes of the eigenvalues of the Jacobian
matrix for the $p$ times iterated map $f^p$ evaluated at the $j$th
fixed point by $\lambda_{1j}, \dots, \lambda_{nj}$. Suppose that the
number of unstable eigenvalues, i.e. $\lambda_{ij}>1$, is given by
$n_u$, and further, that we order them as follows
\begin{equation}\label{eqn:order}
    \lambda_{1j}\geq\cdots\geq\lambda_{n_uj}\geq 1
    \geq\lambda_{(n_u+1)j}\geq\cdots\geq\lambda_{nj}.
\end{equation}
Let $L_j$ denote the product of unstable eigenvalues at the $j$th
fixed point of $f^p$,
\begin{equation}\label{eqn:prod}
    L_j = \lambda_{1j}\lambda_{2j}\cdots\lambda_{n_uj}.
\end{equation}
Then the principal result of the periodic orbit theory is the
following: given a subset $A$ of phase space, one may define its
natural measure to be
\begin{equation}\label{eqn:measure}
    \mu(A) = \lim_{p\to\infty}\mu_p(A),
\end{equation}
where
\begin{equation}\label{eqn:summeas}
    \mu_p(A) = \sum_{j}L_j^{-1}.
\end{equation}
Here the sum is over all fixed points of $f^p$ in $A$; a derivation
of Eq.~(\ref{eqn:measure}) may be found in \cite{Grebogi88}.

This result leads to several important consequences, for example, it
can be shown that the {\em Lyapunov numbers} of $f$ are given by
\begin{equation}\label{eqn:lyap}
    \log\lambda_p =
    \lim_{p\to\infty}\frac{1}{p}\sum_jL_j^{-1}\log\lambda_{pj},
\end{equation}
whilst an analogous result exists for the {\em topological entropy}
\begin{equation}\label{eqn:topent}
    h_{T} = \lim_{p\to\infty}\frac{1}{p}\ln N_p,
\end{equation}
where $N_p$ denotes the number of fixed points of the map $f^p$.
These and similar results obtained within the periodic orbit theory
show that knowledge of the UPOs can yield a great deal of
information concerning the properties of a chaotic dynamical system.
Thus making their efficient detection highly desirable. For further
details, a thorough review of the periodic orbit theory is given in
the book by Cvitanovi\'{c} {\em et al}~\cite{CvitanovicBook}.

At this stage, it is important to point out that most systems of
interest turn out not to be hyperbolic, in particular, the dynamical
systems studied in this thesis are {\em non-hyperbolic}. Hyperbolic
systems, however, remain important due to the fact that they are
more tractable from a mathematical perspective. Indeed, most
rigorous results in dynamical systems are for the case of hyperbolic
systems, and although much of the theory is believed to transfer over to
the non-hyperbolic case there are very few rigorous results.

\subsection{Efficient detection of UPOs}\label{sec:eff}
We have seen that the role of UPOs in chaotic systems is of fundamental
theoretical and practical importance. It is thus not surprising that much
effort has been put into the development of methods for locating periodic
solutions in different types of dynamical systems.  In a limited number of
cases, this can be achieved due to the special structure of the systems.
Examples include the Biham-Wenzel method applicable to H\'{e}non-like
maps~\cite{Biham89}, or systems with known and well ordered symbolic
dynamics~\cite{Hansen95}. For generic systems, however, most methods
described in the literature use some type of an iterative scheme
that, given an initial condition (seed), converges to a periodic
orbit of the chaotic system.  In order to locate all UPOs with a
given period $p$ \newnot{p}, the convergence basin of each orbit for
the chosen iterative scheme must contain at least one seed.  The
seeds are often chosen either at random from within the region of
interest, from a regular grid, or from a chaotic trajectory with or
without close recurrences. Typically, the iterative scheme is chosen
from one of the ``globally'' convergent methods of quasi-Newton or
secant type. However, experience suggests that even the most
sophisticated methods of this type suffer from a common problem:
with increasing period, the basin size of the UPOs becomes so small
that placing a seed within the basin with one of the above listed
seeding schemes is practically impossible~\cite{Miller00}.

A different approach, which appears to effectively deal with the
problem of reduced basin sizes has been proposed by Schmelcher and
Diakonos (SD)~\cite{Schmelcher97,Schmelcher98} \newabb{SD}.  The
basic idea is to transform the dynamical system in such a way that
the UPOs of the original system become stable and can be located by
simply following the evolution of the transformed dynamical system.
That is, to locate period-$p$ orbits of a discrete dynamical system
\newnot{U}
\begin{equation}
  U\!\!:\quad x_{i+1} = f(x_i),\quad
  f\!\!: {\mathbb R}^n \mapsto {\mathbb R}^n\;,
\label{eq:mapf}\end{equation}
one considers an associated flow
\begin{equation}
  \Sigma\!\!: \quad \frac{dx}{ds} = C g(x)\,,
\label{eq:sflow} \end{equation} where $g(x) = f^p(x) - x$ and $C$
\newnot{C} is an $n\times n$ \newnot{n} constant orthogonal matrix.  It is easy to see
that the map $f^p(x)$ \newnot{f}\newnot{g} and flow $\Sigma$
\newnot{sigma} have identical sets of fixed points for any $C$,
while $C$ can be chosen such that unstable period-$p$ orbits of $U$
become stable fixed points of $\Sigma$.

Since it is not generally possible to choose a single matrix $C$
that would stabilise all UPOs of $U$, the idea is to find the
smallest possible set of matrices ${\mathcal C} = \{C_k\}_{k=1}^K$,
such that, for each UPO of $U$, there is at least one matrix $C \in
{\mathcal C}$ that transforms the unstable orbit of $U$ into a
stable fixed point of $\Sigma$.  To this end, Schmelcher and
Diakonos have put forward the following
conjecture~\cite{Schmelcher97}
\begin{conjecture}\label{conj:sd}
Let ${\mathcal C}_{\mathrm{SD}}$ be the set of all $n\times n$
orthogonal matrices with only $\pm 1$ non-zero entries. Then, for
any $n\times n$ non-singular real matrix $G$, there exists a matrix
$C \in {\mathcal C}_{\mathrm{SD}}$ such that all eigenvalues of the
product $CG$ have negative real parts.
\end{conjecture}\newnot{CSD}
\begin{observation}\label{obs:sd}
The set ${\mathcal C}_{\mathrm{SD}}$ forms a group isomorphic to the
Weyl group $B_n$~\cite{HumphreysBook}, i.e. the symmetry group of an
$n$-dimensional hypercube.  The number of matrices in ${\mathcal
C}_{\mathrm{SD}}$ is $K = 2^n n!$.
\end{observation}

The above conjecture has been verified for $n \le
2$~\cite{Schmelcher00}, and appears to be true for $n > 2$, but,
thus far, no proof has been presented. According to this conjecture,
any periodic orbit, whose stability matrix does not have eigenvalues
equal to one, can be transformed into a stable fixed point of
$\Sigma$ with $C \in {\mathcal C}_{\mathrm{SD}}$. In practice, to
locate periodic orbits of the map $U$, we try to integrate the flow
$\Sigma$ from a given initial condition (seed) using different
matrices from the set ${\mathcal C}_{\mathrm{SD}}$. Some of the
resulting trajectories will converge to fixed points, while others
will fail to do so, either leaving the region of interest or failing
to converge within a specified number of steps.

The main advantage of the SD approach is that the convergence basins
of the stabilised UPOs appear to be much larger than the basins
produced by other iterative
schemes~\cite{Davidchack01b,Klebanoff01,Schmelcher98}, making it
much easier to select a useful seed.  Moreover, depending on the
choice of the stabilising transformation, the SD method may converge
to several different UPOs from the same seed.

The flow $\Sigma$ can be integrated by any off-the-shelf numerical
integrator.  Schmelcher and Diakonos have enjoyed considerable
success using a simple Euler method. However, the choice of
integrator for this problem is governed by considerations very
different from those typically used to construct an ODE
\newabb{ODE} solver. Indeed, to locate a fixed point of the flow, it
may not be very efficient to follow the flow with some prescribed
accuracy. Therefore, local error considerations, for example, are
not as important.  Instead, the goal is to have a solver that can
reach the fixed point in as few integration steps as possible.  In
fact, as shown by Davidchack and Lai~\cite{Davidchack99c}, the
efficiency of the method can be improved dramatically when the
solver is constructed specifically with the above goal in mind. In
particular, recognizing the typical stiffness of the flow $\Sigma$,
Davidchack and Lai have proposed a modified semi-implicit Euler
method
\begin{equation}
 x_{i+1} = x_i + [\beta s_i C^{\mathsf T} - G_i]^{-1}g(x_i)\;,
\label{eq:itrDL} \end{equation} where $\beta > 0$ is a scalar
parameter, $s_i = ||g(x_i)||$ is an $L_2$ norm, $G_i =
Dg(x_i)$\newnot{G} is the Jacobian matrix, and ``${\mathsf T}$''
denotes transpose.  Note that, away from the root of $g$, the above
iterative scheme is a semi-implicit Euler method with step size $h =
(\beta s_i)^{-1}$ and, therefore, can follow the flow $\Sigma$ with
a much larger step size than an explicit integrator (e.g. Euler or
Runge-Kutta).  Close to the root, the proposed scheme can be shown
to converge quadratically~\cite{Klebanoff01}, analogous to the
Newton-Raphson method.

Another important ingredient of the algorithm presented
in~\cite{Davidchack99c} is the seeding with already detected
periodic orbits of neighbouring periods.  This seeding scheme
appears to be superior to the typically employed schemes and enables
fast detection of plausibly all\footnote{See \S\ref{sec:numerics}}
periodic orbits of increasingly larger periods in generic
low-dimensional chaotic systems. For example, for the Ikeda map at
traditional parameter values, the algorithm presented
in~\cite{Davidchack99c} was able to locate plausibly all periodic
orbits up to period 22 for a total of over $10^6$ orbit points.
Obtaining a comparable result with generally employed techniques
requires an estimated $10^5$ larger computational effort.

While the stabilisation approach is straightforward for relatively
low-dimensional systems, direct application to higher-dimensional
systems is much less efficient due to the rapid growth of the number
of matrices in ${\mathcal C}_{\mathrm{SD}}$. Even though it appears
that, in practice, far fewer transformations are required to
stabilise all periodic orbits of a given chaotic
system~\cite{Pingel04}, the sufficient subset of transformations is
not known {\em a priori}.  It is also clear that the route of constructing
a universal set of transformations is unlikely to yield substantial
reduction in the number of such transformations.
Therefore, a more promising way of using stabilising transformations
for locating periodic orbits in high-dimensional systems is to
design such transformations based on the information about the
properties of the system under investigation.

\section{Extended systems}\label{sec:extend} The periodic orbit theory is well
developed for low-dimensional chaotic dynamics - at least for
axiom-A systems~\cite{CvitanovicBook}. The question naturally arises
as to whether or not the theory has anything to say for extended
systems. At first glance the transition from low-dimensional chaotic
dynamics to fully developed spatiotemporal chaos may seem rather
optimistic. However, recent results have shown that certain classes
of PDEs
\newabb{PDE} turn out to be less complicated than they initially
appear, when approached from a dynamical systems perspective.
Indeed, under certain conditions their asymptotic evolution can be
shown to lie on a finite dimensional global
attractor~\cite{Robinson95,RobinsonBook2,TemamBook}. Further, by
restricting to equations of the form
\begin{equation}\label{eqn:evolution}
    \frac{du}{dt} + \mathrm{A}u + F(u) = 0,
\end{equation}
where $A$ is a linear differential operator, an even stronger result
may be obtained. Such equations are termed {\em evolution equations}
and their asymptotic dynamics can be shown to lie on a smooth,
finite dimensional manifold, known as the {\em inertial
manifold}~\cite{Robinson95}. In contrast to the aforementioned
global attractor which may have fractal like properties this leads
to a complete description of the dynamics by a finite number of
modes; higher modes being contained in the geometrical constraints
which define the manifold.

A variety of methods for determining {\em all} UPOs up to a given
length exist for low-dimensional dynamical systems (see Chapter
\ref{ch:detection}). For more complex dynamics, such as models of
turbulence in fluids, chemical reactions, or morphogenesis in
biology with high -- possibly infinite -- dimensional phase spaces,
such methods quickly run into difficulties. The most computationally
demanding calculation to date, has been performed by Kawahara and
Kida~\cite{Kida01}. They have reported the detection of two
three-dimensional periodic solutions of turbulent plane Couette flow
using a $15,422$-dimensional discretisation, whilst more recently
Viswanath~\cite{Viswanath06} has been able to detect both periodic
and relative periodic motions in the same system. It is hoped that
such solutions may act as a basis to infer the manner in which
transitions to turbulence can occur.

Our goal is somewhat more modest. We will apply our method to the
model example of an extended system which exhibits spatiotemporal
chaos; the Kuramoto-Sivashinsky equation
\begin{equation}\label{eqn:KSE}
    u_t = -\frac{1}{2}(u^2)_x - u_{xx} - u_{xxxx},\quad t\geq 0,\quad x\in[0,L].
\end{equation}
It was first studied in the context of reaction-diffusion equations
by Kuramoto and Tsuzuki~\cite{Kuramoto76}, whilst Sivashinsky
derived it independently as a model for thermal instabilities in
laminar flame fronts~\cite{Sivashinsky77}. It is one of the simplest
PDEs to exhibit chaos and has played a leading role in studies on
the connection between ODEs and
PDEs~\cite{Foias85,Hyman86,Kevrekidis90}.

It is the archetypal equation for testing a numerical method for
computing periodic solutions in extended systems, and has been
considered in this context in~\cite{Christansen97,Lan04,Zoldi98},
where many UPOs have been detected and several dynamical averages
computed using the periodic orbit theory. Note that the attractor of
the system studied in~\cite{Christansen97} is low dimensional,
whilst those studied in~\cite{Lan04,Zoldi98} have higher intrinsic
dimension. Recently the closely related complex Ginzburg-Landau
equation
\begin{equation}\label{eqn:CGLE}
    A_t = RA +(1+i\nu)A_{xx} -(1+i\mu)A|A|^2,
\end{equation}
has been studied within a similar framework~\cite{Lopez05}, where
Eq.~(\ref{eqn:CGLE}) is transformed into a set of algebraic
equations which are then solved using the Levenberg-Marquardt
algorithm.

\section{A note on numerics}\label{sec:numerics}
Often the physical models put forward by the applied scientists are
extremely complex and, thus, not open to attack via analytical
methods. This necessitates the use of numerical simulations in order
to analyse and understand the models -- particularly in the case
where chaotic behaviour is allowed. However, in those cases one can
always wonder what one is really computing, given the limitations of
floating point systems. This leads to the important question of
whether or not the computed solution is ``close'' to a true solution
of the system of interest. In the case of locating a periodic orbit
on a computer, we would like to know whether the detected orbit
actually exists in the real system. It is well known that in any
numerical calculation accuracy is limited by errors due to roundoff,
discretisation and uncertainty of input data~\cite{StoerBook}. The
difficulty here, lies in the fact that the solutions of a chaotic
dynamical system display extreme sensitivity upon initial
conditions, thus, any tiny error will result in the exponential
divergence of the computed solution from the true one.

For discrete hyperbolic systems, an answer to the question of
validity is provided by the following {\em shadowing lemma} due to
Anosov and Bowen \cite{Anosov67,Bowen75}
\begin{lemma}\label{lem:shadow}
Given a discrete hyperbolic system,
\begin{equation}\label{map}
    x_{i+1} = f(x_i),
\end{equation}
then $\forall \epsilon>0$, $\exists \delta >0$ such that every
$\delta$-pseudo-orbit for $f$ is $\epsilon$-shadowed by a unique real
orbit.
\end{lemma}\noindent
By $\delta$-pseudo-orbit, we mean a computed sequence
$(p_i)_{i\in\mathbb{Z}}$ such that
\begin{equation}\label{eqn:roundofferror}
    |p_{i+1} - f(p_i)| < \delta,
\end{equation}
that is to say, roundoff error at each step of the numerical orbit
is bounded above by $\delta$. Such an orbit is said to be
$\epsilon$-shadowed if there exists a true orbit $(x_i)_{i\in\mathbb{Z}}$
such that
\begin{equation}\label{eqn:trueerror}
    |x_i - p_i| < \epsilon\quad\forall i\in\mathbb{Z}.
\end{equation}

Unfortunately the class of hyperbolic systems is highly restrictive
since such systems are rarely encountered in real problems. For
non-hyperbolic systems - such as those studied in this work -
shadowing results are limited to low dimensional
maps~\cite{Coven88,Hammel88}. Even then, shadowing can only be
guaranteed for a finite number of steps $N$, which is likely to be a
function of the system parameters. Further, it can be shown that
trajectories of non-hyperbolic chaotic systems fail to have long
time shadowing trajectories at all, when unstable dimension
variability persists~\cite{Dawson94,Do04,Sauer97,Sauer02}. Although
the idea of shadowing goes some way towards making sense of
numerical simulations of chaotic systems, it does not answer the
question of whether the numerical orbit corresponds to any real one.
Therefore, we need tools which can rigorously verify the existence
of the corresponding real periodic orbits.

Such methods may also be used to determine the completeness of
sets of periodic orbits, however, this requires that the entire
search is conducted using rigorous numerics and this approach
is inefficient for generic dynamical systems. In general, it is not
possible to prove, within our approach, the completeness of the
detected sets of UPOs. Rather, a stopping criteria must be deduced
after which we can say, with some certainty, that {\em all} UPOs of
period-$p$ have been found; our assertion of completeness will be
based upon the {\em plausibility argument}. The following three
criteria are used for the validation of the argument
\begin{enumerate}
 \item[(i)] Methods based on rigorous numerics (e.g. in~\cite{Galias01})
 have located the same UPOs in cases where such comparison is
 possible (usually for low periods, since these methods are less
 efficient).
 \item[(ii)] Our search strategy scales with the period $p$ (see
 \S\ref{sec:numres} and~\cite{Davidchack01b}).  If we can tune it to
 locate all UPOs for low periods (where we can verify the
 completeness using (i)), it is likely (but not provably) capable of
 locating all UPOs of higher periods as well.
 \item[(iii)] For maps with symmetries, we test the completeness by verifying
 that all the symmetric partners for all detected UPOs have
 been found (see \S\ref{sec:drm} and \S\ref{sec:chm}).
\end{enumerate}
Of course, the preceding discussion can only be applied to discrete
systems.

\subsection{Interval arithmetic} There are several approaches
towards a rigorous computer assisted study of the existence of
periodic orbits. Most make use of the Brouwer fixed point
theorem~\cite{Brouwer1910}, which states that if a convex, compact
set $X\subset\mathbb{R}^n$ is mapped by a continuous function $f$
into itself then $f$ has a fixed point in $X$. Such rigorous methods
tend to fall into one of two classes: (i) topological methods based
upon the index properties of a periodic orbit or (ii) interval
methods. In our discussion we restrict attention to interval methods
since they are the most common in practice. Techniques based on the
index properties are in general less efficient; although recent work
has seen the ideas extended to include infinite dimensional
dynamical systems \cite{Galias,Zgliczynski01}, in particular, in
\cite{Zgliczynski01} several steady states for the
Kuramoto-Sivashinsky equation have been verified rigorously.

Interval methods are based on so called {\em interval arithmetic} --
an arithmetic defined on sets of intervals \cite{Moore79}. Any
computation carried out in interval arithmetic returns an interval
which is guaranteed to contain both the true solution and the
numerical one. Thus, by using properly rounded interval arithmetic,
it is possible to obtain rigorous bounds on any numerical
calculations. In what follows an interval is defined to be a compact
set $\textbf{x}\subset\mathbb{R}$, i.e.
\begin{equation*}
    \textbf{x} = [a,b] = \{x: a\leq x\leq b\},
\end{equation*}
where we use boldface letters to denote interval quantities and
lowercase maths italic to denote real quantities. By an
$n$-dimensional interval vector, we refer to the ordered $n$-tuple
of intervals $\textbf{v} = \{\textbf{x}_1,\dots,\textbf{x}_n\}$.
Note that this leads readily to the definition of higher dimensional
objects.

Arithmetic on the set of intervals is naturally defined in the
following way: let us denote by $\circ$ one of the standard
arithmetic operations $+$, $-$, $\cdot$ and $/$, then the extension
to arbitrary intervals $\textbf{x}_1$ and $\textbf{x}_2$ must
satisfy the condition
\begin{equation*}
    \textbf{x}_1\circ\textbf{x}_2 = \{x = x_1\circ x_2: x_1\in\textbf{x}_1,
    x_2\in\textbf{x}_2\},
\end{equation*}
where, in the case of division, the interval $\textbf{x}_2$ must not
contain the number zero. Importantly, the resulting interval is
always computable in terms of the endpoints, for example, let
$\textbf{x}_1 = [a,b]$ and $\textbf{x}_2 = [c,d]$ then the four
basic arithmetic operations are given by
\begin{eqnarray}
  \textbf{x}_1 + \textbf{x}_2 &=& [a+c,b+d], \\
  \textbf{x}_1 - \textbf{x}_2 &=& [a-d,b-c], \\
  \textbf{x}_1 * \textbf{x}_2 &=& [\mathrm{min}\{ac,ad,bc,bd\},\mathrm{max}\{ac,ad,bc,bd\}],\\
  1/\textbf{x}_1 &=& [1/b,1/a],\\
  \textbf{x}_1 / \textbf{x}_2 &=& \textbf{x}_1*1/\textbf{x}_2.
\end{eqnarray}
This allows one to obtain bounds on the ranges of real valued
functions by writing them as the composition of elementary
operations. For example, if
\begin{equation*}
    f(x) = x(x-1),
\end{equation*}
then
\begin{equation*}
    f([0,1]) = [0,1]([0,1]-1) = [-1,0],
\end{equation*}
note the exact range $[-1/4,0]\subset [-1,0]$ as expected.

Combined with the Brouwer fixed point theorem, interval arithmetic
enables us to prove the existence of solutions to nonlinear systems
of equations. In \S\ref{sec:eff} we saw that the periodic orbit
condition is equivalent to the following system of nonlinear
equations
\begin{equation*}
    g(x) = 0,
\end{equation*}
where $g(x) = f^p(x) - x$. In order to investigate the zeros of the
function $g$ one may apply the {\em Newton operator} to the
n-dimensional interval vector \textbf{x}
\begin{equation}\label{eqn:intnewt}
    N(\textbf{x}) = x_0 -(g'(\textbf{x}))^{-1}g(x_0),
\end{equation}
where $g'(\textbf{x})$ is the interval matrix containing all
Jacobian matrices of the form $g'(x)$ for $x\in\textbf{x}$, and
$x_0$ is an arbitrary point belonging to the interval $\textbf{x}$.
Applying the Brouwer fixed point theorem in the context of the
Newton interval operator leads to the following Theorem.
\begin{theorem}
If $N(\textbf{x})\subset \mathrm{int}(\textbf{x})$ then $g(x) = 0$
has a unique solution in $\textbf{x}$. If
$N(\textbf{x})\cap\textbf{x} = \emptyset$ then there are no zeros of
$g$ in $\textbf{x}$.
\end{theorem}\noindent
For a proof, see for example,~\cite{{Galias01}}.

In practice, the following algorithm may be applied to verify the
existence of a numerical orbit: (i) start by surrounding the orbit
by an $n$-dimensional interval of width $\bar{\epsilon}$, where
$\bar{\epsilon}$ is an integer multiple of the precision,
$\epsilon$, with which the orbit is known, (ii) then apply the
Newton operator to the interval, if $N(\textbf{x})\subset\textbf{x}$
there is exactly one orbit in \textbf{x}, else if
$N(\textbf{x})\cap\textbf{x} = \emptyset$ no orbit of $g$ lies in
\textbf{x}, (iii) if neither of the above hold then either the orbit
is not a true one, or else, $\bar{\epsilon}$ needs to be increased.

In~\cite{Galias01,Galias02} interval arithmetic has been applied to
various two-dimensional maps, note however, that in applications the
Newton operator is replaced by the following method due to Krawczyk
\begin{equation}
    K(\textbf{x}) = x_0 - Ag(x_0) - (Ag'(\textbf{x}) -
    I)(\textbf{x}-x_0),\label{eqn:krawczyk}
\end{equation}
here $A$ is a {\em preconditioning matrix}. The Krawczyk operator of
Eq.~(\ref{eqn:krawczyk}) has the advantage that it does not need to
compute the inverse of $g'$, thus it can be used for a wider class
of systems than the Newton operator.

\section{Overview}\label{sec:overview}
In this thesis, we present an extension of the method of stabilising
transformations to high-dimensional systems. Using periodic orbits
as seeds, we construct  stabilising transformations based upon our
knowledge of the respective stability matrices. The major advantage
of this approach as compared with the method of Schmelcher and
Diakonos is in a substantial reduction of the number of
transformations. Since in practice, high-dimensional systems studied
in dynamical systems typically consist of low-dimensional chaotic
dynamics embedded within a high-dimensional phase space, we are able
to greatly increase the efficiency of the algorithm by restricting
the construction of transformations to the low-dimensional dynamics.
An important aspect of our treatment of high-dimensional flows is
that we do not restrict to a Poincar\'{e} surface of section (PSS).
\newabb{PSS} This is a particularly nice feature, since the phase space
topology for a high-dimensional flow is extremely complex, and the
correct placement of such a surface is a nontrivial task.

In Chapter \ref{ch:detection} we review common techniques for
detecting UPOs, keeping with the theme of the present work our
arrangement is biased towards those methods which are readily
applicable in higher dimensions. We begin Chapter \ref{ch:stabtrans}
by introducing the method of stabilisation transformations (ST)
\newabb{ST} in its original form. In \S\ref{sec:stab2d} we study the
properties of the STs for $n = 2$. We extend our analysis to higher
dimensional systems in \S\ref{sec:ext}, and show how to construct
STs using the knowledge of the stability matrices of already
detected periodic orbit points. In particular, we argue that the
stabilising transformations depend essentially on the signs of
unstable eigenvalues and the directions of the corresponding
eigenvectors of the stability matrices. Section \ref{sec:numres}
illustrates the application of the new STs to the detection of
periodic orbits in a four-dimensional kicked double rotor map and a
six-dimensional system of three coupled H\'{e}non maps. In Chapter
\ref{ch:kse} we propose and implement an extension of the method of
STs for detecting UPOs of flows as well as unstable spatiotemporal
periodic solutions of extended systems. We will see that for
high-dimensional flows -- where the choice of PSS is nontrivial --
it will pay to work in the full phase space. In \S\ref{sec:sub} we
adopt the approach often taken in subspace iteration
methods~\cite{Lust98}, we construct a decomposition of the tangent
space into unstable and stable orthogonal subspaces, and construct
STs without the knowledge of the UPOs. This is particularly useful
since in high dimensional systems it may prove difficult to detect
even a single periodic orbit. In particular, we show that the use of
singular value decomposition to approximate the appropriate
subspaces is preferable to that of Schur decomposition, which is
usually employed within the subspace iteration approach. The
proposed method for detecting UPOs is tested on a large system of
ODEs representing odd solutions of the Kuramoto-Sivashinsky equation
in \S\ref{sec:imp}. Chapter \ref{ch:summary} summarises this work
and looks at further work that should be undertaken to apply the
methods presented to a wider range of problems.

\section{Thesis results}\label{sec:results}
The main results of this thesis are published in
\cite{Crofts05,Crofts06,Crofts07}, the important points of which are
detailed below.

\vspace{0.5cm} \noindent {\bf 1) Efficient detection of periodic
orbits in chaotic systems by stabilising transformations.}
\begin{enumerate}
  \item[$\bullet$] A proof of Conjecture \ref{conj:sd} for the case $n=2$
  is presented. In other words, we show that any two by two matrix may be
  stabilised by at least one matrix belonging to the set proposed by
  Schmelcher and Diakonos.
  \item[$\bullet$] Analysis of the stability matrices for the
  two-dimensional case is provided.
  \item[$\bullet$] The above analysis is used to construct a smaller
  set of stabilising transformations. This enables us to efficiently
  apply the method to high-dimensional systems.
  \item[$\bullet$] Experimental evidence is provided showing the
  successful application of the new set of transformations to
  high-dimensional ($n\geq 4$) discrete dynamical systems. For
  the first time, plausibly complete sets of periodic orbits are
  detected for high-dimensional systems.
\end{enumerate}
\vspace{0.5cm}\noindent{\bf 2) On the use of stabilising
transformations for detecting unstable periodic orbits for the
Kuramoto-Sivashinsky equation.}
\begin{enumerate}
  \item[$\bullet$] The extension of the method of stabilising transformations
  to large systems of ODEs is presented.
  \item[$\bullet$] We construct stabilising transformations using
  the local stretching factors of an arbitrary -- not periodic --
  point in phase space. This is particularly important, since for
  very high-dimensional systems, finding small sets of UPOs to
  initiate the search becomes increasingly difficult.
  \item[$\bullet$] The number of such transformations is shown to be
  determined by the system's dynamics. This contrasts to the
  transformations introduced by Schmelcher and Diakonos which grow
  with system size.
  \item[$\bullet$] In contrast to traditional methods we do not use
  a Poincar\'{e} surface of section, rather, we supply an extra
  equation in order to determine the period.
  \item[$\bullet$] Experimental evidence for the applicability of
  the aforementioned scheme is provided. In particular, we are able
  to calculate many time-periodic solutions of the Kuramoto-Sivashinsky
  equation using a fraction of the computational effort of generally
  employed techniques.
\end{enumerate}

\chapter{Conventional techniques for detecting periodic orbits}
\label{ch:detection}
\begin{quote}
Science is built up of facts, as a house is with stones. But a
collection of facts is no more a science than a heap of stones is a
house.\\\emph{H. J.~Poincar\'{e}}
\end{quote}
The importance of efficient numerical schemes to detect periodic
orbits has been discussed in the Introduction, where we have seen
that the periodic orbits play an important role in our ability to
understand a given dynamical system. In the following chapter we
give a brief review of the most common techniques currently in use.
In developing numerical schemes to detect unstable periodic orbits
(UPO) there is much freedom. Essentially, the idea is to transform
the system of interest to a new dynamical system which possesses the
sought after orbit as an attracting fixed point. Most methods in the
literature are designed to detect UPOs of discrete systems, the
application to the continuous setting is then made by the correct
choice of Poincar\'{e} surface of section (PSS). For that reason in
this chapter, unless stated otherwise, the term dynamical system
will refer to a discrete dynamical system.

\section{Special cases}
In a select number of cases, efficient methods may be designed based
on the special structure inherent within a particular system. In
this section we discuss such methods, with particular interest in
the method due to Biham and Wenzel~\cite{Biham89} applicable to
H\'{e}non-type maps. In Chapter \ref{ch:stabtrans} we apply our
method to a system of coupled H\'{e}non maps and validate our
results against a method which is an extension of the Biham-Wenzel
method.

\subsection{One-dimensional maps}
Perhaps one of the simplest methods to detect UPOs in
one-dimensional maps is that of {\em inverse iteration}. By
observing that the unstable orbits of a one-dimensional map are
attracting orbits of the inverse map, one may simply iterate the
inverse map forward in time in order to detect UPOs. Since the
inverse map is not one-to-one, at each iteration we have a choice of
branch to make. By choosing the branch according to the symbolic
code of the orbit we wish to find, we automatically converge to the
desired cycle.

The method cannot be directly applied to higher-dimensional systems
since they typically have both expanding and contracting directions.
However, if in the contracting direction the chaotic attractor is
thin enough so as to be treated approximately as a zero-dimensional
object, then it may be possible to build an expanding
one-dimensional map by projecting the original map onto the unstable
manifold and applying inverse iteration to the model system. Orbits
determined in this way will typically be ``close'' to orbits of the
full system, and may be used to initiate a search of the full system
using more sophisticated routines.

Methods can also be constructed due to the fact that for
one-dimensional maps well ordered symbolic dynamics exists. For ease
of exposition, we shall describe one such method in the case of a
unimodal mapping, $f$, that is, a mapping of the unit interval such
that $f(0)=f(1)=0$, $f'(c)=0$ and $f''(c)<0$, where $c\in(0,1)$ is
the unique turning point of $f$.

The symbolic dynamics description for a point $x\in[0, 1]$ is given
by $\{s_k\}$ where
\begin{equation}\label{eqn:symbol}
    s_k = \left\{\begin{array}{ll}
    1\quad\mathrm{if}\quad f^{(k-1)}(x)\geq x_c,\\
    0\quad\mathrm{otherwise}.\\
    \end{array} \right.
\end{equation}
Here $x_c$ is the unique turning point of the map $f$. Note that the
order along the $x$-axis of two points $x$ and $y$ can be determined
from their respective itineraries  $\{s_k\}$ and $\{s_k'\}$. To see
this, let us define the {\em well ordered symbolic future} $\gamma$
of the point $x$ to be
\begin{equation}\label{eqn:gamma}
    \gamma(S) = 0.w_1w_2\dots = \sum_{k=1}^{\infty}w_12^{-k},
\end{equation}
where $S$ denotes the symbolic code of the point $x$ and
$$w_k = \sum_{i=1}^ks_i~(\mathrm{mod}2).$$
Now suppose that $s_1 = s_1',\dots, s_k=s_k'$ and $s_{k+1}=0$ and
$s_{k+1}'=1$. Then it can be shown that
\begin{equation}\label{eqn:ordering}
    x<y \Longleftrightarrow  \sum_{i=1}^{k+1}s_i~(\mathrm{mod}2) = 0.
\end{equation}
For a proof see for example~\cite{cvitanovic88}.

Thus in order to detect UPOs of the map $f$ one begins by
determining the symbolic value $\gamma_c$ for the orbit $\bar{S} =
\cdots s_1s_2\cdots s_ps_1s_2\cdots s_p\cdots$. Choosing a starting
point $x_0$ with symbolic value $\gamma$ from the unit interval, one
may update the starting point by comparing its symbolic value,
$\gamma$, against $\gamma_c$, the value for the cycle. Using a
binary search, this procedure will quickly converge to the desired
orbit.

The method can also be extended to deal with certain two-dimensional
systems, in that case one must also define the well ordered symbolic
past in order to uniquely identify orbits of the system. This idea
has been applied to a number of different models, such as the
H\'{e}non map, different types of billiard systems and the
diamagnetic Kepler problem, to name a few~\cite{Hansen95}.

We conclude this section by mentioning that for one-dimensional maps
it is always possible to determine UPOs as the roots to the
nonlinear equation $g = f^p(x)-x$. Since it is straightforward to
bracket the roots of a nonlinear equation in one dimension and thus
apply any of a number of solvers to detect UPOs. We shall discuss
methods for solving nonlinear equations in some detail in
\S\ref{sec:nonlin}.

\subsection{The Biham-Wenzel method}
The Biham-Wenzel (BW) \newabb{BW} method has been developed and
successfully applied in the detection of UPOs for H\'{e}non-type
systems \cite{Biham89,Biham90,Politi92}. It is based on the
observation that for maps such as the H\'{e}non map there exists a
one-to-one correspondence between orbits of the map and the extremum
configurations of a local potential function. For ease of
exposition, we describe the method in the case of the H\'{e}non map
which has the following form
\begin{equation}\label{eqn:Henon}
    x_{i+1} = 1 - ax_i^2 +bx_{i-1},
\end{equation}
when expressed as a one-dimensional recurrence relation.

In order to detect a closed orbit of length $p$ for the H\'{e}non
map, we introduce a $p$-dimensional vector field, $v(x)$, which
vanishes on the periodic orbit
\begin{equation}\label{eqn:bw}
    \frac{dx_i}{d\tau} = v_i(x) = x_{i+1} - 1 + ax_i^2 - bx_{i-1},
    \quad i = 1,\dots, p.
\end{equation}
Now for fixed $x_{i+1}$, $x_{i-1}$ the equation $v_i(x) = 0$ has two
solutions which may be viewed as representing extremal points of a
local potential function
\begin{equation}\label{eqn:potbw}
    v_i(x) = \frac{\partial}{\partial x_i}V_i(x),\quad V_i(x) =
    x_i(x_{i+1} - bx_{i-1} - 1) + \frac{a}{3}x_i^3.
\end{equation}
Assuming the two extremal points to be real, one is a local minimum
of $V_i(x)$ and the other is a local maximum. The idea of BW was to
integrate the flow (\ref{eqn:bw}) with an essential modification of
the signs of its components
\begin{equation}\label{eqn:modbw}
    \frac{dx_i}{d\tau} = s_iv_i,\quad i = 1,\dots , p,
\end{equation}
where $s_i = \pm 1$. Note that Eq.~(\ref{eqn:modbw}) is solved subject
to the periodic boundary condition $x_{p+1} = x_0$.

Loosely speaking, the modified flow will be in the direction of the
local maximum of $V_i(x)$ if $s_i = +1$, or in the direction of the
local minimum if $s_i = -1$. The differential equations
(\ref{eqn:modbw}) then drive an approximate initial guess towards a
steady state of (\ref{eqn:bw}). Since the potential defined in
Eq.~(\ref{eqn:potbw}) is unbounded for large $|x_i|$, the flow will
diverge for initial guesses far from the true trajectory. However,
the basins of attraction for the method are relatively large, and it
can be shown that convergence is achieved for all initial conditions
as long as $|x_i|$, $i = 1,\dots , p$, are small with respect to
$\sqrt{a}$. For the standard parameter values $a=1.4$, $b=0.3$, BW
report the detection of {\em all} UPOs for $p\leq 28$.

An additional feature of the BW method is that the different
sequences, $\{s_i\}$, when read as a binary code, turn out to be
related to the symbolic code of the UPOs. Consider a periodic
configuration $\{x_i\}$, and the corresponding sequence $\{s_i\}$.
Then if we define
\begin{equation}\label{eqn:symbw}
    S_i = (-1)^is_i,\quad i = 1,\dots ,p,
\end{equation}
it can be seen that for most trajectories, the sequence $S_i$
coincides with the symbolic dynamics $\bar{S}_i$ of the H\'{e}non
map, which we define as $\bar{S}_i = 1$ if $x_i > 0$ and $\bar{S}_i
= -1$ if $x_i < 0$. Further, for a particular sequence $\{s_i\}$ the
corresponding UPO does not always exist. In this case, the solution
of (\ref{eqn:modbw}) diverges, thus in principle the BW method
detects all UPOs that exist and indicates which ones do not.

To conclude, the BW method is an efficient method for detecting UPOs
of H\'{e}non-type maps. However, it has been shown that for certain
parameter values the method may fail to converge in some cases
\cite{Grassberger89,Hansen92}. In \cite{Grassberger89} it was shown
that failure occurs in one of two ways: either the solution of
Eq.~(\ref{eqn:modbw}) converges to a limit cycle rather than a
steady state, or uniqueness fails, i.e. two sequences converge to
the same UPO. Extensions of the method include detection of all
$2^p$ UPOs for the complex H\'{e}non map \cite{Biham90}, as well as
detection of UPOs for systems of weakly coupled H\'{e}non maps
\cite{Politi92}.

\section{UPOs: determining roots of nonlinear functions}
\label{sec:nonlin} One of the most popular detection strategies is
to recast the search for UPOs in phase space into the equivalent
problem of detecting zeros of a highly nonlinear function
\begin{equation}\label{eqn:nlin}
    g(x) = f^p(x) - x,\quad x\in\mathbb{R}^n,
\end{equation}
where $p\in\mathbb{Z}^+$ denotes the period. This enables us to use
tools developed for root finding to aid in our search. These methods
usually come in the form of an iterative scheme
\begin{equation}\label{eqn:iter}
    x_{i+1} = \Phi(x_i),\quad i = 0,1,2,\dots,
\end{equation}
which, under certain conditions and for a sufficiently good guess,
$x_0$, guarantee convergence \cite{StoerBook}. Note that the concern
of the current work is with high-dimensional systems, thus we will
not discuss those methods which are not readily applicable in this
case. In particular we do not deal with those approaches designed
primarily for one-dimensional systems. Examples and further
references concerning root finding for one-dimensional systems may
be found, for example, in \cite{NumericalRecipes}.

\subsection{Bisection}
Perhaps the most basic technique for detecting zeros of a function
of one variable is the bisection method. Let
$g:\mathbb{R}\rightarrow\mathbb{R}$ be continuous on the interval
$[a, b]$, then if $g(a)<0$ and $g(b)>0$, it is an immediate
consequence of the {\em intermediate value theorem} that at least
one root of $g$ must lie within the interval $(a, b)$. Assuming the
existence of a unique zero of $g\in(a,b)$, we may proceed as
follows: set $x_1 = \frac{1}{2}(a+b)$ and evaluate $g(x_1)$, if
$g(x_1)>0$ then the root lies in the interval $(a,x_1)$, otherwise
it lies in the interval $(x_1,b)$. The iteration of this process
leads to a robust algorithm which converges linearly in the presence
of an isolated root. Problems may however occur when -- as in our
case -- the detection of several roots is necessary.

An extension of the bisection method to higher dimensional problems
known as {\em characteristic bisection} (CB)
\newabb{CB} has recently been put forward by Vrahatis~\cite{Tassos02,Vrahatis95}.
Before giving a description of the CB method, we need to define the
concept of a {\em characteristic polygon}. Let us denote by $S$ the set of all
$n$-tuples consisting of $\pm 1$'s. Clearly, the number of distinct elements
of $S$ is given by $2^n$. Now we form the $2^n\times n$ matrix $\Lambda_n$ so that
its rows are the elements of $S$ without repetition. Next, consider
an oriented $n$-polyhedron, $\Pi_n$, with vertices $x_i, i = 1,2,\dots,2^n$.
We may construct the associated matrix $S(g;\Pi_n)$, whose $i$th row is given
by
\begin{equation}\label{eqn:signvec}
    \mathrm{sgn}g(x_i) = (\mathrm{sgn}g_1(x_i),
    \mathrm{sgn}g_2(x_i),\dots, \mathrm{sgn}g_n(x_i))^T,
\end{equation}
where sgn denotes the sign function and $g$ is given by
Eq.~(\ref{eqn:nlin}). The polyhedron $\Pi_n$ is said to be a
characteristic polyhedron if the matrix $S(g;\Pi_n)$ is equal to
$\Lambda_n$ up to a permutation of rows. A constructive process for
determining whether a polyhedron is characteristic, is to check that
all combination of signs are present at its vertices. In Figure
\ref{fig:char_poly}a, one can immediately conclude that the polygon
$ABDC$ is not characteristic, whilst the polygon $AEDC$ is.

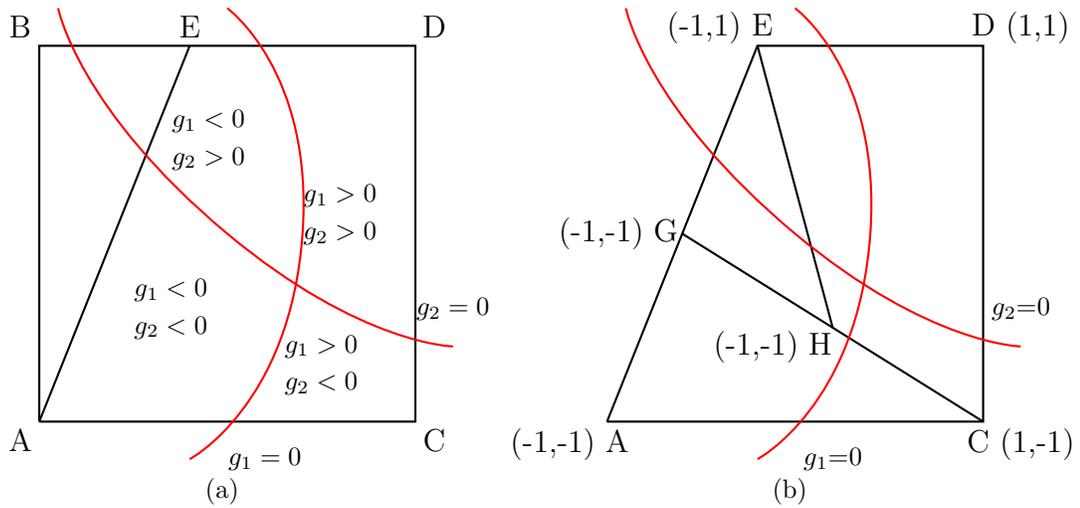
\begin{figure}
 \begin{center}
  \subfigure[]{
   \begin{pspicture}(0,0)(6,6)
   \rput(0.25,0.25){A}\rput(0.25,5.75){B}
   \rput(5.75,0.25){C}\rput(5.75,5.75){D}\rput(2.5,5.75){E}
   \rput(2.75,4.5){\footnotesize $g_1<0$}\rput(2.75,4){\footnotesize $g_2>0$}
   \rput(2.25,2.25){\footnotesize $g_1<0$}\rput(2.25,1.75){\footnotesize $g_2<0$}
   \rput(4.5,3.5){\footnotesize $g_1>0$}\rput(4.5,3){\footnotesize $g_2>0$}
   \rput(4.25,1.5){\footnotesize $g_1>0$}\rput(4.25,1){\footnotesize $g_2<0$}
   \rput(3.5,0){\footnotesize $g_1=0$}\rput(6,2){\footnotesize $g_2=0$}

   \psline(0.5,0.5)(0.5,5.5)(5.5,5.5)(5.5,0.5)(0.5,0.5)
   \psline(0.5,0.5)(2.5,5.5)
   \pscurve[curvature=2 0.1 0,linecolor=red](0.75,6)(3,3)(6,1.5)
   \pscurve[curvature=2 0.1 0,linecolor=red](3,6)(4,3)(2.5,0)
   \end{pspicture}
  }
  \hspace{1cm} \subfigure[]{
   \begin{pspicture}(0,0)(6,6)
   \rput(0,0.2){(-1,-1) A}\rput(6,0.2){C (1,-1)}\rput(6,5.75){D (1,1)}
   \rput(2,5.75){(-1,1) E}\rput(0.65,3){(-1,-1) G}\rput(2.7,1.5){(-1,-1) H}
   \rput(3.5,0){\footnotesize $g_1$=0}\rput(6,2){\footnotesize $g_2$=0}

   \psline(0.5,0.5)(2.5,5.5)(5.5,5.5)(5.5,0.5)(0.5,0.5)
   \psline(5.5,0.5)(1.5,3)
   \psline(3.5,1.75)(2.5,5.5)
   \pscurve[curvature=2 0.1 0,linecolor=red](0.75,6)(3,3)(6,1.5)
   \pscurve[curvature=2 0.1 0,linecolor=red](3,6)(4,3)(2.5,0)
   \end{pspicture}
  }
  \caption{(a) The polyhedron ABDC is non-characteristic while the polyhedron AEDC is characteristic,
  (b) Application of the CB method to AEDC leads to succesive characteristic polyhedra GEDC and HEDC.}
  \label{fig:char_poly}
 \end{center}
\end{figure}

In our discussion we assume that we know the whereabouts of an
isolated root of Eq.~(\ref{eqn:nlin}); in practice rigorous
techniques such as those discussed in \S\ref{sec:numerics} can be
used in order to locate the so called {\em inclusion regions} to
initiate the search. The idea of the CB method is to surround the
region in phase space which contains the root by a succession of
$n$-polyhedra, $\Pi^{(k)}_n$. At each stage the polyhedra are
bisected in such a way that the new polyhedra maintain the quality
of being characteristic. This refinement process is continued until
the required accuracy is achieved. The idea is illustrated for the
two-dimensional case in Figure \ref{fig:char_poly}b, where three
iterates of the process are shown: $\Pi^{(0)}_2 \equiv AEDC$,
$\Pi^{(1)}_2 \equiv GEDC$ and $\Pi^{(2)}_2 \equiv HEDC$.

Similar to the one-dimensional bisection, CB gives an unbeatably
robust algorithm once an inclusion region for the  UPO has been
determined. It is independent of the stability properties of the UPO
and guarantees success as long as the initial polyhedron is
characteristic. However, there lies the crux of the method. In
general to initiate the search, one needs to embed $2^n$ points into
an $n$-dimensional phase space, such that all combination of
$\pm1$'s are represented. This is clearly a nontrivial task for
high-dimensional problems such as those studied in this work.

\subsection{Newton-type methods}\label{subsec:newt}
In this section we discuss a class of iterative schemes which are
popular in practice -- the Newton-Raphson (NR) \newabb{NR} method
and variants thereof. Due to excellent convergence properties,
namely that convergence is ensured for sufficiently good guesses and
that in the linear neighbourhood quadratic convergence is guaranteed,
NR is the method of choice for many applications.

\subsubsection{Newton-Raphson method}
Let $x^*$ be a solution of the equation $g(x^*) = f^p(x^*) - x^* =
0$, and assume that we have an initial guess $x_i$ sufficiently
close to $x^*$, i.e. $x^* - x_i = \delta x_i$ where $\delta x_i$ is
small. Replacing $x^*$ in the above equation gives
\begin{equation}\label{zero}
f^p(x_i + \delta x_i) - (x_i + \delta x_i) = 0,
\end{equation}
which, after Taylor series expansion of $f$ and some algebraic
manipulation yields \newnot{I}
\begin{equation}\label{eqn:taylor}
\delta x_i + (Df^p(x_i) - \mathrm{I}_n)^{-1}(f^p(x_i)-x_i) +
\emph{O}(\delta x_i^2) = 0,
\end{equation}
where $Df^p(x_i)$\newnot{Df} is the Jacobian of $f^p(x)$ evaluated
at $x = x_i$. Restricting to the linear part of Eq.~\ref{eqn:taylor}
and setting $x_{i+1} = x_i + \delta x_i$, leads to the familiar
Newton-Raphson algorithm
\begin{eqnarray}\label{eqn:newt}\nonumber
    x_{i+1} &=& x_i -
    (Df^p(x_i) - \mathrm{I}_n)^{-1}(f^p(x_i)-x_i),\\
    &=& x_i - [Dg(x_i)]^{-1}g(x_i).
\end{eqnarray}
\newnot{Dg}In the case $n=1$, Newton's method has the following geometrical
interpretation: given an initial point $x_i$, we approximate $g$ by
the linear function
\begin{equation}\label{eqn:lin}
    L(x) = g(x_i) + Dg(x_i)(x-x_i),
\end{equation}
which is tangent to $g(x)$ at the point $x_i$. We then obtain the
updated guess $x_{i+1}$ by solving the linear Eq.~(\ref{eqn:lin});
see Figure~\ref{fig:newt}.

\begin{figure}
 \begin{center}
  \begin{pspicture}(0,0)(7,6)
  \psline(0,1)(7,1)
  \psline(2.05,0)(6.55,4.5)
  \psline[linestyle=dashed](4.8,1)(4.8,2.75)
  \pscurve[linecolor=red](0.5,0.5)(2,1)(5,3)(6,5)
  \psdot[dotscale=1.5](2,1)\psdot[dotscale=1.5](3.05,1)\psdot[dotscale=1.5](4.8,1)

  \rput(2,1.35){$x^*$} \rput(3.25,0.65){$x_{n+1}$} \rput(4.8,0.65){$x_n$}
  \rput(4,3){$g(x)$} \rput(6.25,3.5){$L(x)$}

  \end{pspicture}
  \caption{Newton-Raphson method}
  \label{fig:newt}
 \end{center}
\end{figure}
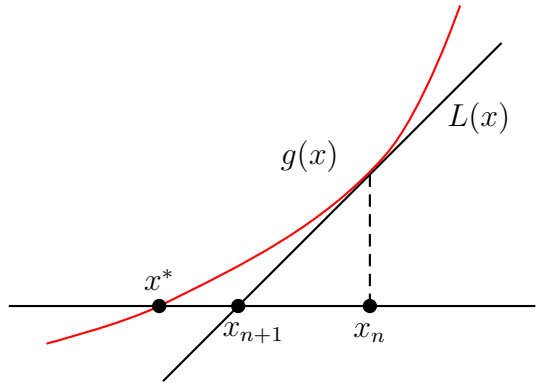

In order to speed up convergence, Householder proposed the following
iterative scheme~\cite{HouseholderBook}
\begin{equation}\label{eqn:House}
    x_{i+1} = x_i +
    (k+1)\Big(\frac{(1/g)^{(k)}}{(1/g)^{(k+1)}}\Big),\quad
    k=1,2,3,\dots.
\end{equation}
Note that this is a generalization of the Newton algorithm
applicable in the one-dimensional setting, where one keeps $k+1$
terms in the Taylor series expansion of $f$. For $k=0$, Newton's
method is restored, whilst $k = 1$, gives the following third order
method first discovered by the astronomer E. Halley in 1694
\begin{equation}\label{eqn:Halley}
    x_{i+1} = x_i - \frac{2gg'}{2(g')^2-gg''}.
\end{equation}

The use of NR to detect periodic orbits is limited mainly due to the
unpredictable nature of its global convergence properties, in
particular for high dimensions. Typically the basins of attraction
will not be simply connected regions, rather, their geometries will
be given by complicated, fractal sets, rendering any systematic
detection strategy useless. Nevertheless, the method is unbeatably
efficient when a good initial guess is known, or for speeding
convergence in inaccurately determined zeros.

\subsubsection{Multiple shooting algorithms}
Multiple shooting is a variant of the NR method, and has been
developed for detecting period-$p$ orbits for
maps~\cite{CvitanovicBook}, particularly for increasing period,
where the nonlinear function $g$ becomes increasingly complex and
can fluctuate excessively. In those cases the periodic orbit is more
easily found as a zero of the following $n\times p$ dimensional
vector function
\begin{equation}\label{eqn:shooting}
    F(\textbf{x}) = \left[\!\!\begin{array}{c}f(x_1) - x_2\\ f(x_2) - x_3\\ \cdot\\
    f(x_{p}) - x_1\end{array}\!\!\right].
\end{equation}
We now write Newton's method in the following more convenient form
\begin{equation}\label{eqn:newtshoot}
  DF(\mathbf{x})\delta \mathbf{x} = -F(\mathbf{x}),
\end{equation}
where we define the {\em Newton step} $\delta \mathbf{x} =
\mathbf{x}_{i+1} - \mathbf{x}_i$, and $DF(\mathbf{x})$ is the
$np\times np$ matrix
\begin{equation}\label{eqn:multiflow}
\left[\begin{array}{cccccc}
    f'(x_1) & -\mathrm{I}_n &&&\\
      & f'(x_2) & -\mathrm{I}_n &&\\
    &&\cdots&\cdots&\\
    &&& f'(x_{p-1}) & -\mathrm{I}_n\\
    -I_n &&&& f'(x_p)\\
\end{array}\right]
\end{equation}
where $\mathrm{I}_n$ is the $n\times n$ identity matrix. Note that,
due to the sparse nature of the Jacobian (\ref{eqn:newtshoot}) can
be solved efficiently.

The technique of multiple shooting gives a robust method for
determining long period UPOs in discrete dynamical systems. The idea
also proves useful for continuous systems. Here there are several
approaches one may take, for example, one may define a sequence of
Poincar\'{e} surface of sections (PSS) in such away that an orbit
leaving one section reaches the next one in a predictable manner,
without traversing other sections along the way, thus reducing the
flow to a set of maps. However, the topology of high-dimensional
flows is hard to visualise and such a sequence of PSSs will usually
be hard to construct; see \cite{CvitanovicBook} for further details.
We present an alternative approach when we discuss the application
of Newton's method to flows in \S \ref{sec:newtflow}.

\subsubsection{The damped Newton-Raphson method}
The global convergence properties of the NR method are wildly
unpredictable, and in general the basins of attraction cannot be
expected to be simply connected regions. To improve upon this
situation, it is useful to introduce the following function
\begin{equation}\label{eqn:norm}
    h(x) = \frac{1}{2}g(x)\cdot g(x),
\end{equation}
proportional to the square of the norm of $g(x)$. Noting that the NR
correction, $\delta x$, is a direction of descent for $h$, i.e.
\newnot{T}
\begin{eqnarray}
  \nabla{h}\cdot\delta x &=& (g^{\mathsf{T}}Dg)\cdot(-Dg^{-1}g), \\
    &=&  -g^{\mathsf{T}}\cdot g < 0,
\end{eqnarray}
suggests the following damped NR scheme
\begin{equation}\label{eqn:damp}
    x_{i+1} = x_i - \lambda_i[Dg(x_i)]^{-1}g(x_i),
\end{equation}
where $0 < \lambda_i \leq 1$ are chosen such that at each step
$h(x)$ decreases. The existence of such a $\lambda_i$ follows from
the fact that the NR correction is in the direction of decreasing
norm.

At this stage it is important to make the distinction between the
Newton step $\delta x$ and the {\em Newton direction} $d =
-[Dg(x)]^{-1}g(x)$. In the following, the Newton step will be a
positive multiple of the Newton direction. The full Newton step
corresponds to $\lambda_i = 1$ in the above equation. In practice
one starts by taking the full NR step since this will lead to
quadratic convergence close to the root. In the case that the norm
increases we backtrack along the Newton direction until we find a
value of $\lambda_i$ for which $h(x_{i+1}) < h(x_i)$.

One such implementation, is given by the Newton-Armijo rule, which
expresses the damping factor, $\lambda_i$, in the form
\begin{equation*}\label{eqn:damping}
    \lambda_i = 2^{-k},\quad k = 0,1,\dots,
\end{equation*}
during each iteration successive values of $k$ are tested until a
value of $k$ is found such that
\begin{equation}\label{eqn:residual}
    h(x_i + \lambda_id) < (1 - \alpha\lambda_i)^2h(x_i),
\end{equation}
where the parameter $\alpha\in(0,1)$ is a small number intended to
make (\ref{eqn:residual}) as easy as possible to satisfy;
see~\cite{KelleyBook} for details and references
concerning the Armijo rule.

Methods like the Armijo rule are often called {\em line searches}
due to the fact that one searches for a decrease in the norm along
the line segment $[x_i, x_i + d]$. In practice some problems can
respond well to one or two decreases in the step length by a modest
amount (such as $1/2$), whilst others require a more drastic
reduction in the step-length. To address such issues, more
sophisticated line searches can be constructed. The strategy for a
practical line search routine is as follows: if after two reductions
by halving, the norm still does not decrease sufficiently, we define
the function
\begin{equation}\label{eqn:lambda}
    \Phi(\lambda) := h(x_i + \lambda d).
\end{equation}
We may build a quadratic polynomial model of $\Phi$ based upon
interpolation at the three most recent values of $\lambda$. The next
$\lambda$ is the minimiser of the quadratic model. For example, let
$\lambda_0 = 1, \lambda_1 = 1/2$ and $\lambda_2 = 1/4$, then we can
model $\Phi(\lambda)$ by
\begin{equation}\label{eqn:quadratic}
    \Phi(\lambda) \approx (\frac{8}{3}\Phi_0 - 8\Phi_1 +\frac{16}{3}\Phi_2)\lambda^2
     +(-2\Phi_0 + 10\Phi_1 -8\Phi_2)\lambda + (\frac{1}{3}\Phi_0 - 2\Phi_1 +
     \frac{8}{3}\Phi_2),
\end{equation}
where $\Phi_i = \Phi(\lambda_i)$ for $i = 0,1,2$. Taking the
derivative of this quadratic, we find that it has a minimum when
\begin{equation}\label{eqn:minimumquad}
    \lambda = \frac{6\Phi_0 - 30\Phi_1 + 24\Phi_2}{16\Phi_0 - 48\Phi_1 +32\Phi_2}.
\end{equation}
These ideas can of course be further generalised through the use of
higher-order polynomials in the modeling of the function
$\Phi(\lambda)$.

Damped NR is a very powerful detection routine. By using information
contained in the norm of $g$ to determine the correct step-size, a
globally convergent method is obtained. Further, close to the root
it reverts to the full NR step thus recovering quadratic
convergence. It should, however, be pointed out that, in practice,
the method may occasionally get stuck in a local minimum, in which
case, detection should be restarted from a new seed.

\subsubsection{Quasi-Newton}
For an $n$-dimensional system, calculation of the Jacobian matrix
requires $n^2$ partial derivative evaluations whilst approximating
the Jacobian matrix by finite differences requires $O(n^2)$ function
evaluations. Often, a more computationally efficient method is
desired.

Consider the case of the one-dimensional NR algorithm for solving
the scalar equation $g(x) = 0$. Recall that the {\em secant
method}~\cite{NumericalRecipes} approximates the derivative
$g'(x_i)$ with the difference quotient
\begin{equation}\label{eqn:diffquotient}
    b_i = \frac{g(x_i) - g(x_{i-1})}{x_i - x_{i-1}},
\end{equation}
and then takes the step
\begin{equation}\label{eqn:rewritequot}
    x_{i+1} = x_i - b_i^{-1}g(x_i).
\end{equation}
The extension to higher dimensions is made by using the basic
property of the Jacobian $Dg(x)\delta x = \delta g$ to write the
equation
\begin{equation}\label{eqn:deltajac}
    Dg(x_i)(x_i - x_{i-1}) = g(x_i) - g(x_{i-1}).
\end{equation}
For scalar equations, (\ref{eqn:diffquotient}) and
(\ref{eqn:deltajac}) are equivalent.  For equations in more than one
variable, (\ref{eqn:diffquotient}) is meaningless, so a wide variety
of methods that satisfy the secant condition (or quasi-Newton
condition) (\ref{eqn:deltajac}) have been designed.

Broyden's method is the simplest example of the quasi-Newton
methods. In the case of Broyden's method, if $x_i$ and $B_i$ are the
current approximate solution and Jacobian, respectively, then
\begin{equation}\label{eqn:broyden}
    x_{i+1} = x_i -\lambda_iB_i^{-1}g(x_i),
\end{equation}
where $\lambda_i$ is the step length for the approximate Newton
direction $d_i = -B_i^{-1}g(x_i)$. After each iteration, $B_i$ is
updated to form $B_{i+1}$ using the Broyden update
\begin{equation}\label{eqn:update}
    B_{i+1} = B_i + \frac{(y-B_is)s^{\mathsf{T}}}{s^{\mathsf{T}}s}.
\end{equation}
Here $y = g(x_{i+1}) - g(x_i)$ and $s = \lambda_id_i$. It is a
straightforward calculation to show that the above update formula
satisfies the quasi-Newton condition (\ref{eqn:deltajac}). Different
choices for the update formula are possible, indeed, depending on
the form of the update a different quasi-Newton scheme results; see
\cite{KelleyBook} for further details and references.

\subsection{Newton's method for flows}\label{sec:newtflow}
Given an autonomous flow
\begin{equation}\label{eqn:flow}
    \frac{dx}{dt} = v(x),
\end{equation}
we can write the periodic orbit condition as
\begin{equation}\label{eqn:poc}
    \phi^{T}(x) - x = 0.
\end{equation}
Here the flow map $\phi^{t}(x)\equiv x(t)$ \newnot{phi} corresponds
to the solution of (\ref{eqn:flow}) at time $t$, and $T$ \newnot{TT}
defines the period. In order to determine UPOs of the flow we must
determine the $(n+1)$ vector $(x,T)$ satisfying Eq.~(\ref{eqn:poc}).
Immediately however, we run into a problem; namely, that the system
in (\ref{eqn:poc}) is under determined. To solve this problem we add
an equation of constraint by way of a Poincar\'{e} surface of
section (PSS). As long as the flow crosses the PSS transversely
everywhere in phase space, we obtain a $(n-1)$-dimensional system of
equations of full rank. In what follows we assume for simplicity
that the PSS is given by a hyperplane $\mathcal{P}$, that is, it
takes the following linear form
\begin{equation}\label{eqn:pss}
    \xi\cdot(x - x_0) = 0,
\end{equation}
where $x_0 \in \mathcal{P}$ and $\xi$ is a vector perpendicular to
$\mathcal{P}$. The action of the constraint given by
Eq.~(\ref{eqn:pss}), is to reduce the $n$-dimensional flow to a
$(n-1)$-dimensional map defined by successive points of directed
intersection of the flow with the hyperplane $\mathcal{P}$.

We may now proceed in similar fashion to the discrete case by
linearising Eq.~(\ref{eqn:poc}), we then obtain an improved guess by
solving the resulting linear system in conjunction with
Eq.~(\ref{eqn:pss}). The NR correction $(\delta x, \delta T)$ is
given by the solution of the following system of linear equations
\begin{equation}
   \left[\begin{array}{cc}
        J - \mathrm{I}_n & v(\phi^T(x)) \\
        \xi^{\mathsf{T}} & 0 \\
      \end{array}\right]
    \left[\begin{array}{c}
         \delta x \\
         \delta T \\
    \end{array}\right] = - \left[\begin{array}{c}
        \phi^T(x) - x \\
        0 \\
    \end{array}\right],
    \label{eqn:pssnewt}
\end{equation}
where the corresponding Jacobian, $J$\newnot{J}, is obtained by
simultaneous integration of the flow and the variational form
\cite{CvitanovicBook}
\begin{equation}\label{eqn:var}
    \frac{dJ}{dt} = AJ, \quad
    J(0) = \mathrm{I}_n,
\end{equation}
here $A = dv/dx$ and as usual $\mathrm{I}_n$ denotes the $n\times n$
identity matrix. Note that the second term in the vector on the RHS
of Eq.~(\ref{eqn:pssnewt}) has been equated to zero, this is since
our initial guess for the NR method lies on the PSS. By
construction, Eq.~(\ref{eqn:pss}) is equal to zero on the PSS.

Since the ODE (\ref{eqn:flow}) is autonomous, it follows that if
$x(t)$ is a solution of (\ref{eqn:flow}), (\ref{eqn:poc}), then any
phase shifted function $x(t+q)$, $q\in\mathbb{R}$, also solves
(\ref{eqn:flow}), (\ref{eqn:poc}). It is this arbitrariness which
forces us to introduce a PSS, by restricting to the surface
$\mathcal{P}$ we eliminate corrections along the flow. Constraints
can of course arise in other more natural ways, for example, if the
flow has an integral of motion, such as the energy in a Hamiltonian
system. In this case one may define a map of dimension $(n-2)$, by
successive points of directed intersection of the flow, with the
$(n-2)$-manifold given by the intersection of the plane,
$\mathcal{P}$, and the corresponding energy shell. This leads to the
following form for the NR correction
\begin{equation}\label{eqn:NRhampss}
    \left[\begin{array}{ccc}
        J - \mathrm{I}_n & v(\phi^T(x)) & \nabla H \\
        \xi^{\mathsf{T}} & 0 & 0 \\
    \end{array}\right]
    \left[\begin{array}{c}
        \delta x \\
        \delta T \\
        \delta E \\
    \end{array}\right] = -
    \left[\begin{array}{c}
        \phi^T(x) - x \\
        0 \\
    \end{array}\right],
\end{equation}
which must be solved simultaneously with
\begin{equation}\label{eqn:hampss}
    H(x) - E = 0,
\end{equation}
in order to remain on the energy shell. Here $H$ is the Hamiltonian
function and $E$ defines its energy.

\begin{figure}
 \begin{center}
  \begin{pspicture}(0,0)(8,5)
   \psline(1,0)(1,4)\psline(7,0)(7,4)
   \psdot[dotsize=7pt](1,2.5)\psdot[dotsize=7pt](7,2.5)
   \pscurve[curvature=2 0.1 0,linecolor=blue](1,2.5)(4,2.25)(7,2.5)
   \pscurve[curvature=2 0.1 0,linecolor=blue](1,1.5)(4,1.25)(5.4,1.3)
   \pscurve[curvature=2 0.1 0,linestyle=dashed,linecolor=blue](5.4,1.3)(6.2,1.35)(7,1.5)
   \psline[arrowscale=1.5,linewidth=1pt]{->}(1,2.5)(1,1.5)
   \psline[arrowscale=1.5,linewidth=1pt]{->}(7,2.5)(5.4,1.3)
   \psline[arrowscale=1.5,linewidth=1pt]{->}(7,2.5)(7,1.5)
   \psarc[arrowscale=1]{->}(0.8,4){0.5}{90}{270}
   \psarc[arrowscale=1]{->}(6.8,4){0.5}{90}{270}

   \rput(1.5,2.6){$x_i$}\rput(0.5,2){$\delta x_i$}
   \rput(7.6,2.6){$x_{i+1}$}\rput(7.6,2){$\delta x_{i+1}$}
   \rput(5.25,1.76){$\delta x(\tau _i)$}
   \rput(1,4.5){$\mathcal{P}$}\rput(7,4.5){$\mathcal{P}$}
  \end{pspicture}
  \caption{After integration time $\tau_i$ the point $x_i$ returns to the Poincar\'{e} surface of section (PSS),
  however the nearby point $x_i+\delta x$ does not. Thus, the matrix needed to map
  an arbitrary deviation $\delta x_i$ on the PSS to the subsequent one $\delta x_{i+1}$
  needs to take into account the implicit dependance of the return time on $x$.}
  \label{fig:pss}
 \end{center}
\end{figure}
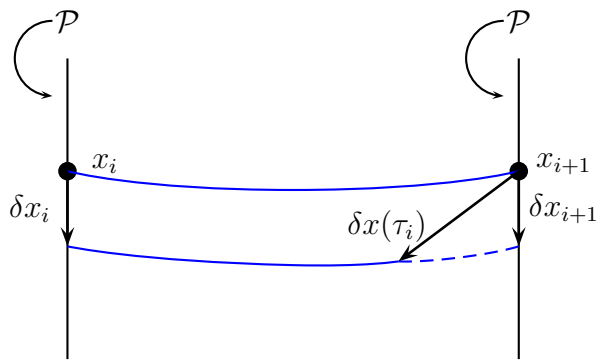

To complete we mention that the preceding discussion seems somewhat
strange, in that, in order to detect fixed points of the
Poincar\'{e} map which has dimension $(n-1)$, we solve a system of
$(n+1)$ equations. This is because, in general, the equations of
constraint are nonlinear, recall that we chose to work with linear
equations to simplify the exposition. In the case when the equations
of constraint can be solved explicitly, one may work directly with
the Poincar\'{e} map, however, care must be taken when computing the
corresponding Jacobians, in particular, the total derivative should
be taken due to the dependence of the constrained variables and the
return times on $x$; see Figure \ref{fig:pss}. This may naturally be
extended to the Hamiltonian case given above, and in fact to a flow
with an arbitrary number of constraints; see~\cite{Doyon05} for
further details.

\subsubsection{Multiple shooting revisited}
As in the discrete case, the detection of long period UPOs can be
aided by the introduction of multiple shooting. We have seen that
the problem of detecting UPOs of a flow may be regarded as a
two-point boundary value problem; the boundary conditions being the
relations closing the trajectory in phase space after time T, i.e.
the periodic orbit condition, $\phi^T(x) = x$. In order to apply
multiple shooting, we begin by discretising the time evolution into
$N$ time intervals
\begin{equation}\label{eqn:disc}
    0=\tau_0 < \tau_1 < \cdots < \tau_{N-1} < \tau_N = T.
\end{equation}
For ease of exposition we partition the time interval $[0, T]$ into
$N$ intervals of equal length, $\Delta = T/N$. The aim of multiple
shooting is then to integrate $N$ trajectories -- one for each
interval -- and to check if the final values coincide with the
initial conditions of the next one up to some precision. If not,
then we apply NR to update the initial conditions.

Let us denote by $x_i$ the initial value of the trajectory at time
$\tau_i$, and the final value of the trajectory at time $\tau_{i+1}$
by $\phi^{\Delta}(x_i)$, then ($N-1$) continuity conditions exist
\begin{equation}\label{eqn:cont}
    C_i(x_i,x_{i+1},T) = \phi^{\Delta}(x_i) - x_{i+1} = 0,\quad i =
    0,1,\dots, N-2,
\end{equation}
together with the boundary conditions
\begin{equation}\label{eqn:bound}
    B(x_{N-1},x_0,T) = \phi^{\Delta}(x_{N-1}) - x_0 = 0.
\end{equation}

Now, we must solve $N$ initial value problems, and for that we adopt
the NR method
\begin{equation}\label{eqn:m}
    C_i(x_i,x_{i+1},T) + \frac{\partial C_i}{\partial x_i}\delta x_i +
    \frac{\partial C_i}{\partial x_{i+1}}\delta x_{i+1} +\frac{\partial C_i}{\partial T}\delta T = 0.
\end{equation}
These equations become
\begin{equation}\label{eqn:cont2}
    C_i(x_i,x_{i+1},T) + J_i\delta x_i + \delta x_{i+1} + \frac{1}{N}v_{i+1}\delta T =
    0,
\end{equation}
and using the boundary conditions (\ref{eqn:bound}), we get
\begin{equation}\label{eqn:bound2}
    B(x_{N-1}, x_0, T) + J_{N-1}\delta x_{N-1}
    -\delta x_0 + \frac{1}{N}v_{N}\delta T= 0.
\end{equation}
Here
\begin{equation}\label{eqn:multjac}
    J_i = \frac{\partial\phi^{\Delta}(x_i)}{\partial
    x_i} \quad \mathrm{and} \quad v_{i+1} = v(\phi^{\Delta}(x_i)).
\end{equation}
We can then write Eqs.~(\ref{eqn:cont2}), (\ref{eqn:bound2}) in a
matrix form of dimension $[nN \times n(N+1)]$
\begin{equation}\label{eqn:multiflow}
\left[\begin{array}{ccccccc}
    J_0 & -\mathrm{I}_n & 0 & \cdots & 0 & 0 & v_1\\
    0 & J_1 & -\mathrm{I}_n & \cdots & 0 & 0 & v_2\\
    \cdots & \cdots & \cdots & \cdots & \cdots & \cdots & \cdots\\
    0 & 0 & 0 & \cdots & J_{N-2} & -\mathrm{I}_n & v_{N-1}\\
    -I_n & 0 & 0 & \cdots & 0 & J_{N-1} & v_{N}\\
\end{array}\right]
\left[\begin{array}{c}
    \delta x_0 \\ \delta x_1 \\ \cdots \\ \delta x_{N-2} \\ \delta
    x_{N-1}\\ \frac{1}{N}\delta T\\
\end{array}\right] =
\left[\begin{array}{c}
    C_0 \\ C_1 \\ \cdots \\ C_{N-2} \\ B \\
\end{array}\right],
\end{equation}
which must be solved simultaneously with the PSS condition
\begin{equation}\label{eqn:bc}
    P(x;x_0) = 0.
\end{equation}
As in the case of simple shooting, the equation (\ref{eqn:bc})
ensures the Newton correction lies on the PSS.

For a more detailed description in applying multiple shooting
algorithms to flows -- in particular conservative ones --
see~\cite{Farantos98}. We conclude by mentioning the {\tt POMULT}
software described in~\cite{Farantos98} which is written in Fortran
and provides routines for a rather general analysis of dynamical
systems. The package was designed specifically for locating UPOs and
steady states of Hamiltonian systems by using two-point boundary
value solvers which are based on similar ideas to the multiple
shooting algorithms discussed here. However, it also includes tools
for calculating power spectra with fast fourier transform, maximum
Lyapunov exponents, the classical correlation function, and the
construction of PSS.

\section{Least-square optimisation tools}\label{sec:optim}
In \S\ref{subsec:newt} we saw that by using the information
contained within the norm of $g(x) = f^p(x)-x$ it is possible to
greatly improve the convergence properties of the Newton-Raphson
(NR) method. Let us now rewrite the function $h(x)$ defined in
Eq.~(\ref{eqn:norm}) in a slightly different form
\begin{equation}\label{eqn:nonlinleast}
    h(x) = \frac{1}{2}\sum_{j = 1}^{n}[g_j(x)]^2.
\end{equation}
Here the $g_j$ are the components of the vector valued function $g =
f^p(x)-x$. Since $h$ is a nonnegative function for all $x$, it
follows that every minimum $x^*$ of $h$ satisfying $h(x^*) = 0$ is a
zero of $g$, i.e. $g(x^*) = f^p(x^*) - x^* = 0$.

In order to detect UPOs of $g$ we can apply the NR algorithm to the
gradient of (\ref{eqn:nonlinleast}), the corresponding Newton
direction is given by
\begin{equation}\label{eqn:nonlinnewtdir}
    d_i = -H^{-1}(x_i)\nabla h(x_i),
\end{equation}
where $H$ is the Hessian matrix of mixed partial derivatives. Note
that the components of the Hessian matrix depend on both the first
derivatives and the second derivatives of the $g_j$:
\begin{equation}\label{eqn:secder}
    H_{kl} = \sum_{j=1}^n\left[\frac{\partial g_j}{\partial x_k}\frac{\partial g_j}{\partial x_l}
    + g_j(x)\frac{\partial g_j}{\partial x_k\partial x_l}\right].
\end{equation}
In practice the second order derivatives can be ignored. To motivate
this, note that the term multiplying the second derivative in
Eq.~(\ref{eqn:secder}) is $g_j(x) = f_j^p(x) - x_j$. For a
sufficiently good initial guess the $g_j$ will be small enough that
the term involving only second order derivatives becomes negligible.
Dropping the second order derivatives leads to the Gauss-Newton
direction
\begin{equation}\label{eqn:gaussnewt}
    d_i = -\tilde{H}^{-1}(x_i)\nabla h(x_i).
\end{equation}
Here $\tilde{H} = Dg^{\mathsf{T}}Dg$ where $Dg$ is the Jacobian of
$g$ evaluated at $x_i$.

In the case where a good initial guess is not available, the
application of Gauss-Newton (GN) \newabb{GN} will in general be
unsuccessful. One way to remedy this is to move in the direction
given by the gradient whenever the GN correction acts as to increase
the norm, i.e. use steepest descent far from the root.

Based on this observation, Levenberg proposed the following
algorithm~\cite{Levenberg44}. The update rule is a blend of the
aforementioned algorithms and is given by
\begin{equation}\label{eqn:levenberg}
    x_{i+1} = x_i - [\tilde{H} + \lambda\mathrm{I}_n]^{-1}\nabla
    h(x_i).
\end{equation}
After each step the error is checked, if it goes down, then
$\lambda$ is decreased, increasing the influence of the GN
direction. Otherwise, $\lambda$ is increased and the gradient
dominates. The above algorithm has the disadvantage that for
increasingly large $\lambda$ the matrix $(\tilde{H} +
\lambda\mathrm{I}_n)$ becomes more and more diagonally dominant,
thus the information contained within the calculated Hessian is
effectively ignored. To remedy this, Marquardt had the idea that an
advantage could still be gained by using the curvature information
contained within the Hessian matrix in order to scale the components
of the gradient. With this in mind he proposed the following update
formula
\begin{equation}\label{eqn:LM}
    x_{i+1} = x_i - [\tilde{H} + \lambda\cdot\mathrm{diag}(\tilde{H})]^{-1}\nabla
    h(x_i),
\end{equation}
known as the Levenberg-Marquardt (LM) \newabb{LM}
algorithm~\cite{Marquardt63}. Since the Hessian is proportional to
the curvature of the function $h$, Eq.~(\ref{eqn:LM}) implies larger
steps in the direction of low curvature, and smaller steps in the
direction with high curvature.

LM is most commonly known for its application to nonlinear least
square problems for modeling data sets. However, it has recently
been applied in the context of UPO detection. Lopez {\em et
al}~\cite{Lopez05} used it to detect both periodic and relative
periodic orbits of the complex Ginzburg-Landau equation. In this
work Lopez {\em et al} used the {\tt lmder} implementation of the LM
algorithm contained in the MINPACK package~\cite{Minpack}. They
found that in the case of the complex Ginzburg-Landau equation the
method out-performs the quasi-Newton methods discussed earlier, as
well as other routines in the MINPACK package such as {\tt hybrj}
which is based on a modification of Powell's hybrid
method~\cite{Powell70}.

\section{Variational methods}\label{subsub:var}
For those equations whose dynamics are governed by a variational
principle, periodic orbits are naturally detected by locating
extrema of a so called cost function. In the case of a classical
mechanical system such a cost function is given by the action, which
is the time integral of the Lagrangian, $L(q,\dot{q},t)$, where
$q\in\mathbb{R}^n$. Starting from an initial loop, that is, a smooth
closed curve, $q(t)$, satisfying $q(t) - q(t+T) = 0$, one proceeds
to detect UPOs by searching for extrema of the action functional
\begin{equation}\label{eqn:lagrange}
    S[q] = \int_{0}^{T} L(q,\dot{q},t)dt.
\end{equation}
One advantage of this approach is that UPOs with a certain topology
can be detected since the initial guess is not a single point but a
whole orbit. This method has recently been applied to the classical
$n$-body problem, where new families of solutions have been
detected~\cite{Simo}.

In the case of a general flow -- such as that defined by
Eq.~(\ref{eqn:flow}) -- one would still like to take advantage of
the robustness offered by replacing an initial point by a rough
guess of the whole UPO. To this end, Civtanovi\'{c} {\em et
al}~\cite{Lan04} have introduced the {\em Newton descent method},
which can be viewed as a generalisation of the multiple shooting
algorithms discussed earlier. The basic idea is to make an informed
guess of what the desired UPO looks like globally, and then to use a
variational method in order to drive the initial loop, $L$, towards
a true UPO.

To begin, one selects an initial loop $L$, a smooth, differentiable
closed curve $\tilde{x}(s)$ in phase space, where $s\in [0, 2\pi]$
is some loop parameter. Assuming $L$ to be close to a UPO, one may
pick $(N-1)$ pairs of nearby points along the loop and the orbit
\begin{eqnarray}
  \tilde{x}_i &=& \tilde{x}(s_i),\quad 0\leq s_1\leq\dots\leq s_{N-1} \leq 2\pi, \\
  x_i &=& x(t_i),\quad 0\leq t_1\leq\dots\leq t_{N-1} \leq T.
\end{eqnarray}
Denote by $\delta\tilde{x}_i$ the deviation of a point $x_i$ on the
periodic orbit from the point $\tilde{x}_i$. Let us define the {\em
loop velocity field}, as the set of s-velocity vectors tangent to
the loop $L$, i.e. $\tilde{v}(\tilde{x}) =
\mathrm{d}\tilde{x}/\mathrm{d}s$. Then the goal is to continuously
deform the loop $L$ until the directions of the loop velocity field
are aligned with the vector field $v$ evaluated on the UPO. Note
that the magnitudes of the tangent vectors depend upon the
particular choice of parameter $s$. Thus, to match the magnitudes, a
local time scaling is introduced
\begin{equation}\label{eqn:varscale}
    \lambda(s_i) = \Delta t_i/\Delta s_i.
\end{equation}

Now, since $x_i = \tilde{x}_i + \delta\tilde{x}_i$ lies on the
periodic orbit, we have
\begin{equation}\label{eqn:varpoc}
    \phi^{\Delta t_i + \delta t}(\tilde{x}_i + \delta\tilde{x}_i) =
    \tilde{x}_{i+1} + \delta\tilde{x}_{i+1}.
\end{equation}
Linearisation of (\ref{eqn:varpoc}) yields the multiple shooting NR
method
\begin{equation}\label{eqn:varlin}
    \delta\tilde{x}_{i+1} - J(\tilde{x}_i,\Delta
    t_i)\delta\tilde{x}_i - v_{i+1}\delta t_i = \phi^{\Delta t_i}(\tilde{x}_i)
    - \tilde{x}_{i+1},
\end{equation}
which, for a sufficiently good guess, generates a sequence of loops
L with a decreasing cost function given by
\begin{equation}\label{eqn:varcost}
    F^2(\tilde{x}) = \frac{N-1}{2\pi}\sum_{i=1}^{N-1}
    (\phi^{\Delta t_i}(\tilde{x}_i)-\tilde{x}_{i+1})^2.
\end{equation}
As with any NR method, a decrease in the cost function is not always
guaranteed if the full NR-step is taken. However, a decrease in
$F^2$ is ensured if infinitesimal steps are taken.

Fixing $\Delta s_i$ one may proceed by $\delta\tau$, where $\tau$ is
a fictitious time which parameterises the Newton Descent, thus the
RHS of Eq.~(\ref{eqn:varlin}) is multiplied by $\delta\tau$
\begin{equation}\label{eqn:varms}
    \delta\tilde{x}_{i+1} - J(\tilde{x}_i,\Delta
    t_i)\delta\tilde{x}_i - v_{i+1}\delta t_i = \delta\tau(\phi^{\Delta t_i}(\tilde{x}_i)
    - \tilde{x}_{i+1}).
\end{equation}
According to Eq.~(\ref{eqn:varscale}) the corresponding change in
$\Delta t_i$ is given by
\begin{eqnarray}
  \nonumber
  \delta t_i &=& \Delta t_i(\tau + \delta\tau) - \Delta
  t_i(\tau),\\\nonumber
   &=& \frac{\partial\Delta t_i}{\partial\tau}\delta\tau, \\
   &=& \Delta s_i\frac{\partial\lambda}{\partial\tau}\delta\tau,
   \label{eqn:varvar}
\end{eqnarray}
and similarly $\delta\tilde{x} =
(\partial/\partial\tau)\tilde{x}(s_i,\tau)\delta\tau$. Dividing
Eq.~(\ref{eqn:varms}) by $\delta\tau$ and substituting the above
expressions for $\delta t_i$, $\delta\tilde{x}_i$ yields
\begin{equation}\label{eqn:varflow}
    \frac{d\tilde{x}_{i+1}}{d\tau} - J(\tilde{x}_i,\Delta
    t_i)\frac{d\tilde{x}_i}{d\tau} -
    v_{i+1}\frac{\partial\lambda}{\partial\tau}(s_i,\tau)\Delta s_i
    = \phi^{\Delta t_i}(\tilde{x}_i) - \tilde{x}_{i+1}.
\end{equation}
Now in the limit $N \rightarrow \infty$, the step sizes  $\Delta
s_i$, $\Delta t_i = \emph{O}(\frac{1}{N})\rightarrow 0$, giving
$$v_{i+1}\approx v_i,\quad \tilde{x}_{i+1} \approx \tilde{x}_i + \tilde{v}_i\Delta s_i,$$
$$J(\tilde{x}_i,\ \Delta t_i) \approx I + A(\tilde{x}_i)\Delta t_i, \quad
\phi^{\Delta t_i}(\tilde{x}_i) \approx \tilde{x}_i + v_i\Delta
t_i.$$ Substituting into Eq.~(\ref{eqn:varflow}) and using the
relation (\ref{eqn:varscale}) results in the following PDE which
describes the evolution of a loop $L(\tau)$ toward a UPO
\begin{equation}\label{eqn:pdeloop}
    \frac{\partial^2\tilde{x}}{\partial{s}\partial{\tau}} -
    \lambda A \frac{\partial{\tilde{x}}}{\partial{\tau}} -
    v\frac{\partial{\lambda}}{\partial{\tau}} = \lambda v -
    \tilde{v},
\end{equation}
where $v$ is the vector field given by the flow, and $\tilde{v}$ is
the loop velocity field. The importance of Eq.~(\ref{eqn:pdeloop})
becomes transparent if we rewrite the equation in the following form
\begin{equation}\label{eqn:pdeloopode}
    \frac{\partial}{\partial\tau}(\tilde{v} - \lambda v) = -(\tilde{v} - \lambda
    v),
\end{equation}
from which it follows that the fictitious time flow decreases the
cost functional
\begin{equation}\label{eqn:varficcost}
    F^2[\tilde{x}] = \frac{1}{2\pi}\int_{L(\tau)}[\tilde{v}(\tilde{x}) - \lambda v(\tilde{x})]^2d\tilde{x}.
\end{equation}

The Newton descent method is an infinitesimal step variant of the
damped NR method. It uses a large number of initial points in phase
space as a seed which allows for UPOs of a certain topology to be
detected, in practice it is very robust and has been applied in
particular to the Kuramoto-Sivashinsky equation in the weakly
turbulent regime, where many UPOs have been detected. The main
problem -- as with most multiple shooting methods -- is that of
efficiency. In \cite{Lan04} they state that most of the
computational effort goes into inverting the $[(nd+1)\times(nd+1)]$
matrix resulting from the numerical approximation of
Eq.~(\ref{eqn:pdeloop}), here $d$ is the number of loop points and
will typically be large.

\section{Summary}
Our review is not exhaustive and further methods exist for the
detection of unstable periodic orbits (UPO) in a chaotic dynamical
system. For example, in those cases where the systems dynamics
depend smoothly upon some parameter the method of continuation can
be applied. Here one uses the fact that their exists a window in
parameter space where one can easily detect solutions, the idea is
then to slowly vary the parameter into a previously unaccessible
region in parameter space whilst carefully following the
corresponding solution branch. These ideas have recently been
applied to determine relative periodic motions to the classical
fluid mechanics problem of pressure-driven flow through a circular
pipe~\cite{Kerswell04}. Another recent proposal by Parsopoulos and
Vrahatis~\cite{Parsopoulos02} uses the method of {\em particle swarm
optimization} in order to detect UPOs as global minima of a properly
defined objective function. The preceding algorithm has been tested
on several nonlinear mappings including a system of two coupled
H\'{e}non maps and the Predator-Prey mapping, and has been found to
be an efficient alternative for computing UPOs in low-dimensional
nonlinear mappings.

For generic maps (or flows) no computational algorithm is guaranteed
to find all solutions up to some period $p$ (time
$T_{\mathrm{max}}$). For systems where the topology is well
understood one can build sufficiently good initial guesses so that
the Newton-Raphson method can be applied successfully. However, in
high-dimensional problems the topology is hard to visualize and
typically Newton-Raphson will fail. Here the variational method
offers a robust alternative, methods that start with a large number
of initial guesses along an orbit, do not suffer from the numerical
instability caused by exponential sensitivity of chaotic
trajectories. Unfortunately, these methods tend to result in large
systems of equations, the solution of which is a costly process,
this has the knock on effect of slowing the convergence of such
methods considerably. This makes it extremely difficult to find UPOs
with larger periods or detect them for higher dimensional systems.
In Chapters \ref{ch:stabtrans} and \ref{ch:kse} we shall introduce and
discuss the method of stabilising transformations in some detail, due
to the fact that the method possesses excellent global convergence
properties and needs only marginal {\em a priori} knowledge of the
system, it is ideally suited to determining UPOs in high-dimensional
systems.

\chapter{Stabilising transformations}
\label{ch:stabtrans}
\begin{quote}
Tis plain that there is not in nature a point of stability to be
found: everything either ascends or declines.\\
\emph{Sir Walter Scott}
\end{quote}

An algorithm for detecting unstable periodic orbits based on
stabilising transformations (ST) has had considerable success in
low-dimensional chaotic systems~\cite{Davidchack99c}. Applying the
same ideas in higher dimensions is not trivial due to a rapidly
increasing number of required transformations. In this chapter we
introduce the idea behind the method of STs before going onto to
analyse their properties. We then propose an alternative approach
for constructing a smaller set of transformations.  The performance
of the new approach is illustrated on the four-dimensional kicked
double rotor map and the six-dimensional system of three coupled
H\'{e}non maps~\cite{Crofts06}.

\section{Stabilising transformations as a tool for detecting UPOs}
The method of STs was first introduced in 1997 by Schmelcher and
Diakonos (SD) as a tool for determining unstable periodic orbits
(UPOs) for general chaotic maps~\cite{Schmelcher97}. For the first
time UPOs of high periods where determined for complicated
two-dimensional maps, for example, the Henon map~\cite{Henon76} and
the Ikeda-Hammel-Jones-Maloney map~\cite{Ikeda79,Hammel85}. One of
the main advantages of the method is that the basin of attraction
for each UPO extends far beyond its linear neighbourhood and is given
by a simply connected region.

Let us consider the following $n$-dimensional system:
\begin{equation}
  U\!\!:\quad x_{i+1} = f(x_i),\quad
  f\!\!: {\mathbb R}^n \mapsto {\mathbb R}^n\;.
\end{equation}
The basic idea of the SD method is to introduce a new dynamical
system $\bar{U}$, with fixed points in exactly the same position in
phase space as the UPOs of $U$ but with differing stability
properties, ideally the fixed points of $\bar{U}$ will be
asymptotically stable. In general however, no such system exists.
Fortunately, this can be rectified by choosing a set of associated
systems, such that, for each UPO of $U$, there is at least one
system which possesses the UPO as a stable fixed point. To this end,
SD have put forward the following set of systems
\begin{equation}\label{eqn:assmap}
  \bar{U}_k\!\!:\quad x_{i+1} = x_i + \lambda C_k[f^p(x_i) - x_i],\quad k
  = 1,\dots,2^nn!,
\end{equation}
here $\lambda$ is a small positive number, $p$ is the period and
$C_k$ is an $n\times n$ matrix with elements
$[C_{k}]_{ij}\in\{0,\pm1\}$ such that each row or column contains
only one nonzero element. To motivate Eq.~(\ref{eqn:assmap}), it is
useful to look at the Jacobian of the system $\bar{U}_k$
\begin{equation}\label{assjac}
    \frac{dx_{i+1}}{dx_i} = \mathrm{I}_n + \lambda C_k[Df^p(x_i) - \mathrm{I}_n].
\end{equation}
For a fixed point of $\bar{U}_k$ to be stable all eigenvalues of the
above expression must have absolute value less than one. This can be
achieved in the following way: (i) choose a linear transformation
$C_k$ such that all eigenvalues of the matrix
$C_k[Df^p(x_i)-\mathrm{I}_n]$ have negative real part, and (ii) pick
$\lambda$ sufficiently small so that the eigenvalues of $\mathrm{I}
+ \lambda C_k[Df^p(x_i) - \mathrm{I}_n]$ are scaled so as to all
have absolute value less than one. The fact that such a $\lambda$
exists is clear, whilst step (i) follows from Conjecture
\ref{conj:sd}.

This leads to the following algorithm to detect {\em all} period-$p$
orbits of a chaotic map. We begin by placing seeds over the
attractor using any standard seeding strategy, for example, a
chaotic trajectory or a uniform grid. Using the different matrices
$C_k$, we construct the $2^nn!$ associated systems $\bar{U}_k$ --
see Observation \ref{obs:sd}. For a sufficiently small value of
$\lambda$, we propagate each seed using the iteration scheme in
(\ref{eqn:assmap}) to compute a sequence $\{x_i\}$ for each of the
$2^nn!$ maps, $\bar{U}_k$. If a sequence converges, we check whether
a new period-$p$ orbit point has been found, and if so, proceed to
detect the remaining $(p-1)$ orbit points by iterating the map $f$.
In order to test for the completeness of a set of period-$p$ orbits,
SD suggest iterating a multiple of the initial seeds used in the
detection process. If no new orbits are found then they claim that
the completion of the set is assured.

From a computational perspective the SD method has two major
failings. Firstly, and most importantly, the application to
high-dimensional systems is restricted due to the fact that the
number of matrices in the set $\mathcal{C}_{\mathrm{SD}}$ increases
very rapidly with system dimension. Secondly, it follows from our
prior discussion that the parameter $\lambda$ scales with the
magnitude of the largest eigenvalue of (\ref{assjac}). Since the
instability of a UPO increases with the period $p$, very small
values of $\lambda$ must be taken in order to detect long period
orbits. This has the effect of increasing the number of steps --
hence function evaluations -- needed to obtain convergence, which in
turn leads to a rise in the computational costs.

The problem of extending the method of STs to high-dimensional
systems is the primary concern of this thesis and will be dealt with
in some detail in this chapter and the next. The problem of
efficiency is solved to some extent by the following observation
in~\cite{Schmelcher98}.
\begin{observation}\label{obs:euler}
For a given map $f$. Taking the limit $\lambda\to 0$ in
Eq.~(\ref{eqn:assmap}) leads to the following continuous flow
\begin{equation}\label{eqn:vec}
\lim_{\lambda\to 0}\frac{x_{i+1}-x_i}{\lambda} = \frac{dx}{ds} =
C_k(f^p(x) - x) = C_kg(x).
\end{equation}
\end{observation}\noindent
The equilibria of Eq.~(\ref{eqn:vec}) are located at exactly the
same positions and share the same stability properties as the fixed
points of the associated systems $\bar{U}_k$ in the limit of small
$\lambda$. By transferring to the continuous setting, the dependency
on $\lambda$ has been removed, this allows us to use our favourite
off-the-shelf numerical integrator to enable us to detect UPOs of
$f$; the SD method is precisely Euler's method to solve the flow of
(\ref{eqn:vec}) with step-size $\Delta t = \lambda$.

\begin{figure}[t]
 \begin{center}
  \psset{xunit=1.5cm,yunit=2.25cm}
  \begin{pspicture}(-3,-2)(3,2)
   \psaxes[Dx=1,Dy=1,Ox=-3]{-}(-3,0)(-3,-1.25)(2.99999,1.25)
   \psline(-3,-2)(3,-2)\psline(-3,1.25)(3,1.25)\psline(-3,-1.25)(-3,-2)\psline(3,1.25)(3,-2)
   \psline[linewidth=1pt](-3,-1.25)(3,-1.25)\psline[linewidth=1pt](-3,-1.5)(3,-1.5)
   \psline[linewidth=1pt](-3,-1.75)(3,-1.75)
   \psplot[linecolor=blue,plotstyle=curve,plotpoints=3000]{-3}{3}{/k 180 3.14159265 div def x 2 exp k
   mul cos}

   \psset{arrowscale=1.25,linewidth=1pt,linecolor=red,tbarsize=4pt,dotscale=1}

   \psdot(-2.8,-1.25)\psdot(-2.8,-1.75)\psdot(-2.175,-1.25)\psdot(-2.175,-1.5)
   \psdot(-2.175,-1.75)\psdot(-1.25,-1.25)\psdot(-1.25,-1.75)\psdot(1.25,-1.25)
   \psdot(1.25,-1.5)\psdot(1.25,-1.75)\psdot(2.175,-1.25)\psdot(2.175,-1.75)
   \psdot(2.8,-1.25)\psdot(2.8,-1.5)\psdot(2.8,-1.75)

   \psline{->}(-3,-1.25)(-2.8,-1.25)
   \psline[ArrowInside=-<]{-|}(-2.8,-1.25)(-2.55,-1.25)
   \psline{|->}(-2.4,-1.25)(-2.175,-1.25)
   \psline[ArrowInside=-<]{-|}(-2.175,-1.25)(-1.9,-1.25)
   \psline[ArrowInside=->]{|-}(-1.7,-1.25)(-1.25,-1.25)
   \psline[ArrowInside=-<]{-|}(-1.25,-1.25)(-0.75,-1.25)
   \psline{->}(3,-1.25)(2.8,-1.25)
   \psline[ArrowInside=-<]{-|}(2.8,-1.25)(2.55,-1.25)
   \psline{|->}(2.4,-1.25)(2.175,-1.25)
   \psline[ArrowInside=-<]{-|}(2.175,-1.25)(1.9,-1.25)
   \psline[ArrowInside=->]{|-}(1.7,-1.25)(1.25,-1.25)
   \psline[ArrowInside=-<]{-|}(1.25,-1.25)(0.75,-1.25)
   \psline{<-|}(-3,-1.5)(-2.8,-1.5)
   \psline[ArrowInside=->](-2.8,-1.5)(-2.175,-1.5)
   \psline[ArrowInside=-<]{-|}(-2.175,-1.5)(-1.25,-1.5)
   \psline[ArrowInside=->,ArrowInsidePos=0.75](-1.25,-1.5)(1.25,-1.5)
   \psline[ArrowInside=-<]{-|}(1.25,-1.5)(2.175,-1.5)
   \psline[ArrowInside=->](2.175,-1.5)(2.8,-1.5)
   \psline{->}(3,-1.5)(2.8,-1.5)
   \psline{|->}(-2.95,-1.75)(-2.8,-1.75)
   \psline[ArrowInside=-<]{-|}(-2.8,-1.75)(-2.65,-1.75)
   \psline[ArrowInside=->,ArrowInsidePos=0.75]{|-}(-2.55,-1.75)(-2.175,-1.75)
   \psline[ArrowInside=-<]{-|}(-2.175,-1.75)(-1.55,-1.75)
   \psline[ArrowInside=->]{|-}(-1.45,-1.75)(-1.25,-1.75)
   \psline[ArrowInside=-<]{-|}(-1.25,-1.75)(-1.05,-1.75)
   \psline[ArrowInside=->,ArrowInsidePos=0.75]{|-}(-.95,-1.75)(1.25,-1.75)
   \psline[ArrowInside=-<]{-|}(1.25,-1.75)(1.875,-1.75)
   \psline{|->}(2.375,-1.75)(2.175,-1.75)
   \psline[ArrowInside=-<]{-|}(2.175,-1.75)(1.974,-1.75)
   \psline[ArrowInside=->,ArrowInsidePos=0.75]{|-}(2.475,-1.75)(2.8,-1.75)
   \psline{->}(3,-1.75)(2.8,-1.75)
   \rput(0,-1.125){(a)}\rput(0,-1.375){(b)}\rput(0,-1.625){(c)}
\end{pspicture}
\end{center}
\caption{(colour online) Shown in red are the basins of convergence
of (a) the Newton method, (b) the Schmelcher and Diakonos method
with $0<\lambda<0.3568$ and $C=\mathrm{I}$, and (c) the Davidchack
and Lai method with $\beta = 4.0$ and $C = \mathrm{I}$ to the zeros
of a function $g(x) = \cos(x^2)$ in the interval $(-3, 3)$. Arrows
indicate the direction of convergence, and large dots are the zeros
to which the method converges.}\label{fig:compbasin}
\end{figure}
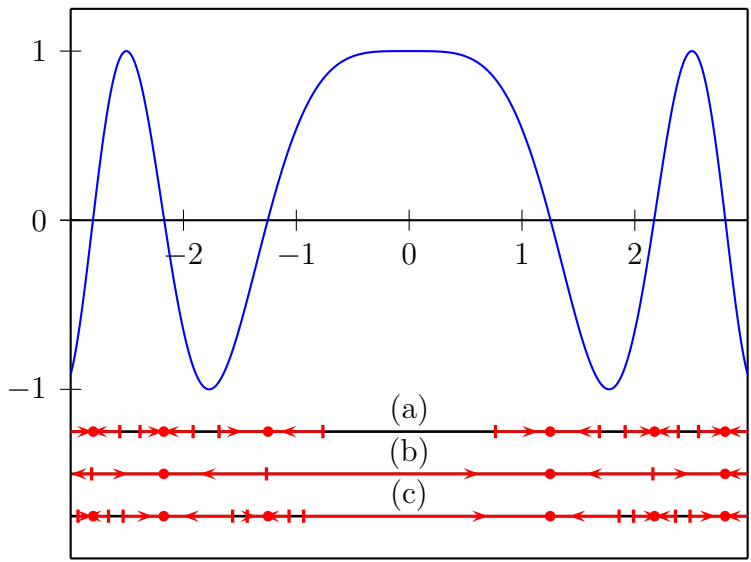

Taking into account the typical stiffness of the flow in
Eq.~(\ref{eqn:vec}), Davidchack and Lai proposed a modified scheme
employing a semi-implicit Euler method \cite{Davidchack99c}. The
semi-implicit scheme is obtained from the fully implicit Euler
method in the following way: starting from the implicit scheme as
applied to Eq.~(\ref{eqn:vec})
\begin{equation}\label{eqn:impeul}
    x_{i+1} = x_i + hCg(x_{i+1}),
\end{equation}
one obtains the semi-implicit Euler routine by expanding the term
$g(x_{i+1})$ in a Taylor series about $x_i$ and retaining only those
terms which are linear
\begin{eqnarray}
    x_{i+1}&=&x_i + hCg(x_i) + hCg'(x_i)\Delta x_i + O(\Delta x_i^2), \\
    &\doteq& x_i + [\frac{1}{h}C^{\mathsf{T}}-g'(x_i)]^{-1}g(x_i).
\end{eqnarray}
The Davidchack-Lai method is completed by choosing the step-size $h
= 1/(\beta s_i)$, here $\beta > 0$ is a scalar parameter and $s_i =
||g(x_i)||$ is an $L_2$ norm. This leads to the following algorithm
for detecting UPOs
\begin{equation}\label{eqn:davlai}
    x_{i+1} = x_i + [\beta s_i C^{\mathsf{T}} - G_i]^{-1}g(x_i),
\end{equation}
where $G_i = Dg(x_i)$ is the Jacobian matrix, and ``${\mathsf T}$''
denotes transpose.

In the vicinity of an UPO, the function $g(x)$ tends to zero and
Newton-Raphson is restored, indeed it can be shown that the method
retains quadratic convergence \cite{Klebanoff01}. Away from the UPO
and for sufficiently large $\beta$, the flow is accurately 
reproduced thus conserving the global convergence properties of the
SD method. A nice property of the above algorithm is that the basin
sizes can be controlled via the parameter $\beta$. In particular in
\cite{Davidchack01b} it is shown that increasing the value of the
parameter $\beta$ results in a larger basin size. A large value of
$\beta$, however, requires more time steps for the iteration to
converge to a UPO. This leads to a trade off between enlarging the
basins and the speed of convergence. A comparison of the basins of
attraction for the Newton-Raphson, the Schmelcher-Diakonos and the
Davidchack-Lai methods is given in Figure \ref{fig:compbasin}. Here
the roots of the function $g(x) = \cos(x^2)$ are shown on the
interval $(-3, 3)$ along with the corresponding basins of attraction
for the three methods; see \cite{Davidchack99c} for further details.

\subsection{Seeding with periodic orbits}
An important ingredient of the Davidchack-Lai algorithm lies in the
selection of initial seeds. Davidchack and Lai claim that the most
efficient strategy for detecting UPOs of period $p$ is to use UPOs
of other periods as seeds \cite{Davidchack99c}. This can be
understood due to the distribution of the UPOs. It is well known
that orbit points cover the attractor in a systematic manner, which
in turn reflects the foliation of the function $f^p(x)$ and its
iterates. For low-dimensional maps, Davidchack and Lai have
considerable success with the aforementioned seeding strategy as
compared to traditional seeding algorithms. Indeed, for the
H\'{e}non and the Ikeda-Hammel-Jones-Maloney maps, {\em all}
period-$p$ orbits were detected using period-$(p-1)$ orbits,
provided they exist; in the case that no period-$(p-1)$ orbits
exist, one can try either $(p-2)$ or $(p+1)$.

The use of UPOs as seeds is also crucial to our work. We shall see
that it is the information contained in the UPOs -- or rather their
Jacobian's  -- that will enable us to construct new sets of STs, in
turn allowing for efficient detection of UPOs in high-dimensional
systems.

\section{Stabilising transformations in two dimensions}
\label{sec:stab2d} The stability of a fixed point $x^\ast$ of the
flow
\begin{equation}
  \Sigma\!\!: \quad \frac{dx}{ds} = C g(x)\,,
\end{equation} is determined by the real parts of the eigenvalues
of the matrix $C G$, where $G = Dg(x^\ast)$ is the Jacobian matrix
of $g(x)$ evaluated at $x^\ast$. For $x^\ast$ to be a stable fixed
point of $\Sigma$, the matrix $C$ has to be such that all the
eigenvalues of $CG$ have negative real parts.  In order to
understand what properties of $G$ determine the choice of a
particular stabilising transformation $C$, we use the following
parametrisation for the general two-dimensional orthogonal matrices
\begin{equation}
C_{s,\alpha} = \left[\!\!\begin{array}{cc}
s \cos \alpha & \sin \alpha \\
-s \sin \alpha & \cos \alpha \end{array} \!\!\right],
\label{eq:par2d} \end{equation} where $s = \pm 1$ and $-\pi < \alpha
\le \pi$. When $\alpha = -\pi/2,~0,~\pi/2$, or $\pi$, we obtain the
set of matrices ${\mathcal C}_{\mathrm{SD}}$. For example,
\begin{equation}\label{eqn:mat1}
    C_{1,\pi/2} = \left[\begin{array}{cc}
                     0 & 1 \\
                     -1 & 0 \\
                  \end{array}\right],
\end{equation}
and
\begin{equation}\label{eqn:mat2}
    C_{-1,\pi} = \left[\begin{array}{cc}
                     1 & 0 \\
                     0 & -1 \\
                 \end{array}\right].
\end{equation}
If we write
\begin{equation}\label{eqn:Gjac}
    G := \{G_{ij}\},\quad i,j = 1,2,
\end{equation}
then the eigenvalues of $C_{s,\alpha}G$ are given by the following
equations:
\begin{equation}
  \sigma_{1,2} = -A\cos(\alpha - \theta) \pm
  \sqrt{A^2\cos^2(\alpha - \theta) - s \det G}
\label{eq:eigs1} \end{equation} where $\det G = G_{11} G_{22} -
G_{12} G_{21}$, $~A = \frac{1}{2}\sqrt{(s G_{11} + G_{22})^2
    + (s G_{12} - G_{21})^2}$, and
\begin{equation}\label{eq:theta}
  \tan\theta = \frac{s G_{12} - G_{21}}{-s G_{11} - G_{22}}\:,
  \qquad -\pi < \theta \le \pi\,.
\end{equation}
Note that the signs of the numerator and denominator are significant
for defining angle $\theta$ in the specified range and should not be
canceled out.  It is clear from Eq.(\ref{eq:eigs1}) that both
eigenvalues have negative real parts when
\begin{equation}\label{eq:cond}
  s = \bar{s} := {\mathrm{sgn}}\,\det G,\qquad
  \mbox{and} \qquad | \alpha - \theta | < \textstyle\frac{\pi}{2}\,,
\end{equation}
provided that $\det G \neq 0$.  This result proves the validity of
Conjecture \ref{conj:sd} for $n = 2$.  Moreover, it shows that there
are typically two matrices in ${\mathcal C}_{\mathrm{SD}}$ that
stabilise a given fixed point.

Parameter $\theta$ clearly plays an important role in the above
analysis.  The following two theorems show its relationship to the
eigenvalues and eigenvectors of the stability matrix of a fixed
point.
\begin{theorem}
Let $x^\ast$ be a saddle fixed point of $f^p(x): {\mathbb R}^2
\mapsto {\mathbb R}^2$ whose stability matrix $\Dfp(x^\ast)$ has
eigenvalues $\lambda_{1,2}$ such that $|\lambda_2| < 1 <
|\lambda_1|$ and eigenvectors defined by the polar angles $0 \le
\phi_{1,2} < \pi$, i.e. $v_{1,2} = (\cos\phi_{1,2},
\sin\phi_{1,2})^{\mathsf T}$. Then the following is true for the
angle $\theta$ defined in Eq.~(\ref{eq:theta}):
\begin{description}
  \item[{\sc Case 1.} $\lambda_1 < -1$:]
   \begin{equation}\label{eq:th1case1a}
    \textstyle \theta \in \left(-\frac{\pi}{2}, \frac{\pi}{2}\right)\,.
   \end{equation}
   Moreover, if $|\lambda_1| \gg 1$, then
   \begin{equation}
    \theta = (\phi_1 - \phi_2) (\mathrm{mod~}\pi) -
   \textstyle \frac{\pi}{2} + \mathcal{O}(\frac{1}{|\lambda_1|})\,.
  \label{eq:th1case1b} \end{equation}
 \item[{\sc Case 2.} $\lambda_1 > 1$:]
  \begin{equation}\label{eq:th1case2}
   \theta = \left\{\begin{array}{cc}
   \frac{3\pi}{2} - \phi_1 - \phi_2\,, & 0 < \phi_1 - \phi_2 < \pi\,,\\
   \frac{\pi}{2} - \phi_1 - \phi_2\,, & -\pi < \phi_1 - \phi_2 < 0\,.\\
   \end{array} \right.
  \end{equation}
\end{description}
\end{theorem}
\begin{proof}
Matrix $G = \Dfp(x^\ast) - I_2$, where $I_2$ is the identity matrix,
can be written as follows:
\begin{equation}\label{eq:eigdec}
  G = \left[\!\!\begin{array}{cc} G_{11} \!&\! G_{12} \\
  G_{21} \!&\! G_{22}\end{array} \!\!\right] =
  \left[\!\!\begin{array}{cc} \cos\phi_1 \!&\! \cos\phi_2 \\
  \sin\phi_1 \!&\! \sin\phi_2\end{array} \!\!\right]
  \left[\!\!\begin{array}{cc} \lambda _{1}-1 \!&\! 0 \\
  0 \!&\! \lambda _{2}-1\end{array} \!\!\right]
  \left[\!\!\begin{array}{cc} \cos\phi_1 \!&\! \cos\phi_2 \\
  \sin\phi_1 \!&\! \sin\phi_2\end{array} \!\!\right]^{-1}
\end{equation}
\begin{description}
  \item[Case 1:] Since $\det G = (\lambda_1 - 1)(\lambda_2 - 1)
  > 0$ we set $s = 1$ and obtain from Eq.~(\ref{eq:theta}):
  \begin{equation}
    \tan\theta = \frac{(\lambda_1-\lambda_2)\cot(\phi_1-\phi_2)}
    {2-\lambda_1-\lambda_2}\,,
    \label{eq:theta_c1}
  \end{equation}
  where, just like in Eq.~(\ref{eq:theta}), as well as in
  Eqs.~(\ref{eq:theta_c2}) and
  (\ref{eq:theta_t1}) below, the signs of the numerator and
  denominator should not be canceled out. Since $2-\lambda_1-\lambda_2
  > 0$, we have that $\cos\theta > 0$ or
  \[  \textstyle\theta\in\left(-\frac{\pi}{2}, \frac{\pi}{2}\right)\,. \]
  For $|\lambda_1| \gg 1$, Eq.~(\ref{eq:theta_c1}) yields:
  \[  \tan\theta = \left[-1 + \mathcal{O}(\frac{1}{|\lambda_1|})\right]\cot(\phi_1-\phi_2)\,, \]
  and, given Eq.~(\ref{eq:th1case1a}), the result in
  Eq.~(\ref{eq:th1case1b}) immediately follows.

  \item[Case 2:] In this case $\det G = (\lambda_1 - 1)(\lambda_2
  - 1) < 0$, so, from Eq.~(\ref{eq:theta}) with $s = -1$:
  \begin{eqnarray}\label{eq:theta_c2}
  \tan\theta &=& \frac{(\lambda_2-\lambda_1)\cos(\phi_1+\phi_2)/
  \sin(\phi_1-\phi_2)}{(\lambda_2-\lambda_1)\sin(\phi_1+\phi_2)/
  \sin(\phi_1-\phi_2)}\\ &=& \frac{-\cos(\phi_1+\phi_2)/
  \sin(\phi_1-\phi_2)}{-\sin(\phi_1+\phi_2)/\sin(\phi_1-\phi_2)}\,,
  \nonumber
  \end{eqnarray}
  since $\lambda_2-\lambda_1 < 0$. The result in
  Eq.~(\ref{eq:th1case2}) follows.
\end{description}
\end{proof}
\begin{theorem}
Let $x^\ast$ be a spiral fixed point of $f^p(x): {\mathbb R}^2
\mapsto {\mathbb R}^2$ whose stability matrix $\Dfp(x^\ast)$ has
eigenvalues $\lambda_{1,2} = \lambda \pm \mathrm{i}\omega$.  Then
\begin{eqnarray}\label{eq:th2}
  \theta & \in & \textstyle \left(-\frac{\pi}{2},
\frac{\pi}{2}\right)
  \hspace{1.9cm}\mathrm{if}\quad \lambda < 1\,,\\
  \theta & \in & \textstyle \left(-\pi, -\frac{\pi}{2}\right)\cup
  \left(\frac{\pi}{2}, \pi\right)
  \quad\mathrm{if}\quad \lambda > 1\,.\nonumber
\end{eqnarray}
\end{theorem}
\begin{proof}
 The stability matrix can be decomposed as follows:
\begin{equation}
  \Dfp(x^\ast)=\left[\!\!\begin{array}{cc}
  \cos\phi & \mathrm{e}^\eta \\
  \sin\phi & 0\end{array} \!\!\right]
  \left[\!\!\begin{array}{cc} \lambda & \omega \\
  -\omega & \lambda \end{array} \!\!\right]
  \left[\!\!\begin{array}{cc} \cos\phi & \mathrm{e}^\eta \\
  \sin\phi & 0\end{array} \!\!\right]^{-1}\,,
\end{equation}
where $\eta \in \mathbb{R}$.  Given that $G = \Dfp(x^\ast) - I_2$,
we have from Eq.~(\ref{eq:theta}):
\begin{equation}\label{eq:theta_t1}
  \tan\theta = \frac{-\omega\cosh\eta/\sin\phi}{1-\lambda}\,.
\end{equation}
The result in Eq.~(\ref{eq:th2}) follows from the sign of the
denominator.\qquad\end{proof}

The key message of the above theorems is that the ST matrix depends
mostly on the directions of the eigenvectors and the signs of the
unstable\footnote{That is, eigenvalues whose magnitude is larger
than one.} eigenvalues of $Df^p$ (or their real parts), and only marginally 
on the actual magnitudes of the eigenvalues. This means that a
transformation that stabilises a given fixed point $x^\ast$ of $f^p$
will also stabilise fixed points of all periods with similar
directions of eigenvectors and signs of the unstable eigenvalues. In
the next Section, we will show how this observation can be used to
construct STs for efficient detection of periodic orbits in systems
with $n > 2$.

\section{Extension to higher-dimensional systems}
\label{sec:ext} To extend the analysis of the preceding Section to
higher-dimensional systems, we note that the matrix
$C_{\bar{s},\theta}$, as defined by Eqs.~(\ref{eq:par2d}),
(\ref{eq:theta}),  and (\ref{eq:cond}), is closely related to the
orthogonal part of the {\em polar decomposition} of $G$; see
Appendix \ref{ch:appendix4}.  Recall that any non-singular $n\times
n$ matrix can be uniquely represented as a product
\begin{equation}\label{eq:polar}
  G = QB\,,
\end{equation}
where $Q$ is an orthogonal matrix and $B$ is a symmetric positive
definite matrix.  The following theorem provides the link between
$C_{\bar{s},\theta}$ and $Q$ for $n = 2$:
\begin{theorem}
Let $G \in \mathbb{R}^{2\times 2}$ be a non-singular matrix with the
polar decomposition $G = QB$, where $Q$ is an orthogonal matrix and
$B$ is a symmetric positive definite matrix. Then matrix
$C_{\bar{s},\theta}$, as defined by Eqs.~(\ref{eq:par2d}),
(\ref{eq:theta}) and (\ref{eq:cond}), is related to $Q$ as follows:
\begin{equation}\label{eq:qt1}
  C_{\bar{s},\theta} = -Q^{\mathsf T}
\end{equation}
\end{theorem}
\begin{proof}
Since $C_{\bar{s}, \theta}$ is an orthogonal matrix by definition,
it is sufficient to prove that $C_{\bar{s}, \theta}G$ is symmetric
negative definite.  Then, by the uniqueness of the polar
decomposition, it must be equal to $-B$.

Denote by $b_{ij}$ the element in the $i$-th row and $j$-th column
of $C_{\bar{s}, \theta}G$.  We must show that $b_{12} = b_{21}$.
Using Eq.~(\ref{eq:theta}), we have that
\begin{eqnarray}
  b_{12} &=& \bar{s}G_{12}\cos\theta + G_{22}\sin\theta\\
  &=& \left[\bar{s}G_{12} + G_{22}\frac{\bar{s}G_{12}-G_{21}}
  {-\bar{s}G_{11}-G_{22}}\right]\cos\theta \nonumber \\
  &=& \left[\frac{G_{11}G_{12}+G_{21}G_{22}}
  {\bar{s}G_{11}+G_{22}}\right]\cos\theta\,, \nonumber
\end{eqnarray}
and similarly
\begin{eqnarray}
  b_{21} &=& G_{21}\cos\theta - \bar{s}G_{11}\sin\theta\\
  &=& \left[G_{21} - \bar{s}G_{11}\frac{\bar{s}G_{12} -
  G_{21}}{-\bar{s}G_{11} - G_{22}}\right]\cos\theta\nonumber \\
  &=& \left[\frac{G_{11}G_{12}+G_{21}G_{22}}
  {\bar{s}G_{11}+G_{22}}\right]\cos\theta\,, \nonumber
\end{eqnarray}
hence the matrix $C_{\bar{s}, \theta}G$ is symmetric.  Since, by
definition, $\theta$ and $\bar{s}$ are chosen such that the
eigenvalues of $C_{\bar{s}, \theta}G$ are negative, the matrix
$C_{\bar{s}, \theta}G$ is negative definite.  Finally, by the
uniqueness of the polar decomposition,
\[ C_{\bar{s}, \theta}G = -B = -Q^{\mathsf T}G\,, \]
which completes the proof.\qquad\end{proof}

For $n > 2$, we can always use the polar decomposition to construct
a transformation that will stabilise a given fixed point. Indeed, if
a fixed point $x^\ast$ of an $n$-dimensional flow has a non-singular
matrix $G = Dg(x^\ast)$, then we can calculate the polar
decomposition $G = QB$ and use
\begin{equation}\label{eq:qt2}
  C = -Q^{\mathsf T}\;,
\end{equation}
to stabilise $x^\ast$. Moreover, by analogy with the two-dimensional
case, we can expect that the same matrix $C$ will also stabilise
fixed points $\tilde{x}$ with the matrix $\tilde{G} =
Dg(\tilde{x})$, as long as the orthogonal part $\tilde{Q}$ of the
polar decomposition $\tilde{G} = \tilde{Q}\tilde{B}$ is sufficiently
close to $Q$. More precisely,
\begin{observation} \label{obs:stab}
$C$ will stabilise $\tilde{x}$, if all eigenvalues of the product
$Q\tilde{Q}^\mathsf{T}$ have positive real parts.
\end{observation}

We base this observation on the following corollary of Lyapunov's
stability theorem; see Appendix \ref{ch:appendix4}
\begin{corollary}
Let $B \in \mathbb{R}^{n\times n}$ be a positive definite symmetric
matrix.  If $Q \in \mathbb{R}^{n\times n}$ is an orthogonal matrix
such that all its eigenvalues have positive real parts, then all the
eigenvalues of the product $QB$ have positive real parts as well.
\label{corr:lyapunov}
\end{corollary}
\begin{proof}
According to Lyapunov's theorem, a matrix $A \in \mathbb{R}^{n\times
n}$ has all eigenvalues with positive real parts if and only if
there exists a symmetric positive definite $G \in
\mathbb{R}^{n\times n}$ such that $GA + A^\mathsf{T}G = H$ is
positive definite.

Let $A = QB$ and let's choose $G$ in the form $G =\frac{1}{2}
QB^{-1}Q^\mathsf{T}$.  Since $B$ is positive definite, its inverse
$B^{-1}$ is also positive definite, and, since $G$ and $B^{-1}$ are
related by a congruence transformation, according to Sylvester's
inertia law -- see Appendix \ref{ch:appendix4} -- $G$ is also positive
definite. Now,
\[ GQB + (QB)^\mathsf{T}G =
\textstyle\frac{1}{2}QB^{-1}Q^\mathsf{T}QB +
\frac{1}{2}BQ^\mathsf{T}QB^{-1}Q^\mathsf{T} = \frac{1}{2}[Q +
Q^\mathsf{T}]\;.\] Therefore, $QB$ has eigenvalues with positive
real parts if and only if $\frac{1}{2}[Q + Q^\mathsf{T}]$ is
positive definite.  The proof is completed by observing that, for
orthogonal matrices, the eigenvalues of $\frac{1}{2}[Q +
Q^\mathsf{T}]$ are equal to the real parts of the eigenvalues of
$Q$.\qquad\end{proof}

Note that Observation \ref{obs:stab} is a direct generalisation of
conditions in Eq.~(\ref{eq:cond}) which are equivalent to requiring
that the eigenvalues of $C_{s,\alpha}C_{\bar{s}, \theta}^\mathsf{T}$
have positive real parts.

In the scheme where already detected periodic orbits are used as
seeds to detect other orbits~\cite{Davidchack99c}, we can use $C$ in
Eq.~(\ref{eq:qt2}) as a stabilising matrix for the seed $x^\ast$.
Based on the analysis in \S\ref{sec:stab2d}, this will allow us to
locate a periodic orbit in the neighbourhood of $x^\ast$ with
similar invariant directions and the same signs of the unstable
eigenvalues.  Note, however, that the neighbourhood of the seed
$x^\ast$ can also contain periodic orbits with the similar invariant
directions but with some eigenvalues having the opposite sign
(i.e.~orbits with and without reflections).  To construct
transformations that would stabilise such periodic orbits, we can
determine the eigenvalues and eigenvectors of the stability matrix
of $x^\ast$
\begin{equation}
  \Dfp(x^\ast) = V\Lambda V^{-1}\,,
\label{eq:gev} \end{equation} where $\Lambda :=
\mathrm{diag}(\lambda_1,\ldots,\lambda_n)$ is the diagonal matrix of
eigenvalues of $\Dfp(x^\ast)$ and $V$ is the matrix of eigenvectors,
and then calculate the polar decomposition of the matrix
\begin{equation}
  \hat{G} = V(S\Lambda - \mathrm{I}_n)V^{-1}\,,
\label{eq:gpev} \end{equation} where $S = \mathrm{diag}(\pm 1, \pm
1,\ldots,\pm 1)$.  Note that, as follows from the analysis in
\S\ref{sec:stab2d} for $n = 2$ and numerical evidence for $n > 2$,
changing the sign of a stable eigenvalue will not result in a
substantially different ST.  Therefore, we restrict our attention to
the following subset of $S$:
\begin{equation}\label{eq:signs}
  S_{ii} = \left\{ \begin{array}{rl} \pm 1, & |\lambda_i| > 1\,,\\
  1, & |\lambda_i| < 1\,,   \end{array} \right.
  \quad\mathrm{for}\quad i=1,\ldots,n\,.
\end{equation}
For a seed with $k$ real unstable eigenvalues, this results in $2^k$
possible transformations.  Note that, on the one hand, this set is
much smaller than ${\mathcal C}_{\mathrm{SD}}$, while, on the other
hand, it allows us to target all possible types of periodic orbits
that have invariant directions similar to those at the seed.

\section{Numerical results}\label{sec:numres}

In this Section we illustrate the performance of the new STs on a
four-dimensional kicked double rotor map~\cite{Romeiras92} and a
six-dimensional system of three coupled H\'{e}non
maps~\cite{Politi92}. Both systems are highly chaotic and the number
of UPOs is expected to grow rapidly with increasing period. The goal
is to locate {\em all} UPOs of increasingly larger period. Of
course, the completeness of the set of orbits for each period cannot
be guaranteed, but it can be established with high degree of
certainty by using the plausibility criteria outlined in the
Introduction.

In order to start the detection process, we need to have a small set
of periodic orbits (of period $p > 1$) that can be used as seeds.
Such orbits can be located using, for example, random seeds and the
standard Newton-Raphson method (or the scheme in
Eq.~(\ref{eqn:davlai}) with $\beta = 0$).  We can then use these
periodic orbits as seeds to construct the STs and detect more UPOs
with higher efficiency.  The process can be iterated until we find
no more orbits of a given period. In previous work
DL~\cite{Davidchack99c,Davidchack01b} showed that for
two-dimensional maps such as H\'{e}non and Ikeda it is sufficient to
use period-$(p-1)$ orbits as seeds to locate plausibly all
period-$p$ orbits.  For higher-dimensional systems, such as those
considered in the present work, these seeds may not be sufficient.
However, it is always possible to use more seeds by, for example,
locating some of the period-$(p+1)$ orbits, which can then be used
as seeds to complete the detection of period-$p$ orbits. The
following recipe can be used as a general guideline for developing a
specific detection scheme for a given system: {\em
\begin{enumerate}
\item Find a set of orbit points of low period using random seeds
and the iterative scheme in Eq.~(\ref{eqn:davlai}) with $\beta = 0$
(i.e. the Newton-Raphson scheme).
\item To locate period-$p$ orbits, first use period-$(p-1)$ orbits
as seeds. For each seed $x_0$, construct $2^{k}$ STs $C$ using
Eqs.~(\ref{eq:gev}-\ref{eq:signs}), where $k$ is the number of
unstable eigenvalues of $D\!f^{p-1}(x_0)$.
\item Starting from $x_0$ and with a fixed value of $\beta > 0$
use the iteration scheme in Eq.~(\ref{eqn:davlai}) to construct a
sequence $\{x_{i}\}$ for each of the $2^{k}$ STs. If a sequence
converges to a point, check whether it is a new period-$p$ orbit
point, and if so, proceed to find a complete orbit by iterating the
map $f$.
\item Repeat steps $2-3$ for several $\beta$ in order to determine
the optimal value of this parameter (see explanation below).
\item Repeat steps $2-4$ using newly found period-$p$ points
as seeds to search for period-$(p+1)$ orbits.
\item Repeat steps $2-4$ using incomplete set of period-$(p+1)$
orbits as seeds to find any missing period-$p$ orbits.
\end{enumerate} }

Although we know that the action of $\beta$ is to increase the basin
size of the stabilised points, it is not known {\em a priori} what
values of $\beta$ to use for a given system and period.  Monitoring
the fraction of seeds that converge to periodic orbits, we observe
that it grows with increasing $\beta$ until it reaches saturation,
indicating that the iterative scheme faithfully follows the flow
$\Sigma$.  On the other hand, larger $\beta$ translates into smaller
integration steps and, therefore, longer iteration sequences.  Thus
the optimal value of $\beta$ is just before the saturation point. As
demonstrated previously by DL~\cite{Davidchack99c} and observed in
the numerical examples presented in the following sections, this
value appears to scale exponentially with the period and can be
estimated based on the information about the detection pattern at
lower periods.

The stopping criteria in step 3, which we use in the numerical
examples discussed below, are as follows.  The search for UPOs is
conducted within a rectangular region containing a chaotic invariant
set.  The sequence $\{x_{i}\}$ is terminated if (i) $x_{i}$ leaves
the region, (ii) $i$ becomes larger than a pre-defined maximum
number of iterations (we use $i > 100+5\beta$ ), (iii) the sequence
converges, such that $\|g(x_i)\| < \mbox{\em Tol}_g$.  In cases (i)
and (ii) a new sequence is generated from a different seed and/or
with a different stabilising matrix.  In case (iii) five Newton
iterations are applied to $x_i$ to allow convergence to a fixed
point to within the round-off error.  A point $x^\ast$ for which
$\|g(x^\ast)\|$ is the smallest is identified with a fixed point of
$f^p$.  The maximum round-off error over the set ${\cal X}_p$ of all
detected period-$p$ orbit points
\begin{equation}\label{eq:maxe}
  \epsilon_\mathrm{max}(p)=\max\{\|g(x^\ast)\|: x^\ast\in{\cal X}_p\}
\end{equation}
is monitored in order to assess the accuracy of the detected orbits.

To check if the newly detected orbit is different from those already
detected, its distance to other orbit points is calculated: if
$\|x^\ast-y^\ast\|_\infty > \mbox{\em Tol}_x$ for all previously
detected orbit points $y^\ast$, then $x^\ast$ is a new orbit point.
Even for a large number of already detected UPOs, this check can be
done very quickly by pre-sorting the detected orbit points along one
of the system coordinates and performing a binary search for the
points within $\mbox{\em Tol}_x$ of $x^\ast$.  The infinity norm in
the above expression is used for the computational efficiency of
this check.

The minimum distance between orbit points
\begin{equation}\label{eq:mind}
    d_\mathrm{min}(p)=\min\{\|x^\ast-y^\ast\|_\infty :
    x^\ast, y^\ast \in{\cal X}_p\}
\end{equation}
is monitored and the algorithm is capable of locating all isolated
UPOs of a given period $p$ as long as $\epsilon_\mathrm{max}(p) <
\mbox{\em Tol}_g \lesssim \mbox{\em Tol}_x < d_\mathrm{min}(p)$.
Since typically $\epsilon_\mathrm{max}(p)$ increases and
$d_\mathrm{min}(p)$ decreases with $p$ (see Tables \ref{tab:drmp}
and \ref{tab:chm}), the above conditions can be satisfied up to some
period, after which higher-precision arithmetics needs to be used in
the evaluation of the map. For the numerical examples presented in
the following sections we use double-precision computation with
$\mbox{\em Tol}_g = 10^{-6}$ and $\mbox{\em Tol}_x = 10^{-5}$.

\subsection{Kicked double rotor map}\label{sec:drm}
The kicked double rotor map describes the dynamics of a mechanical
system known as the double rotor under the action of a periodic
kick; a derivation is given in Appendix \ref{ch:appendix3}. It is a
4-dimensional map defined by
\begin{equation}
  \left[\!\!\begin{array}{c}x_{i+1}\\y_{i+1}\end{array}\!\!\right] =
  \left[\!\!\begin{array}{l}My_{i} + x_{i}~(\mbox{mod}~2\pi)\\
  Ly_{i} + c\sin{x_{i+1}}\end{array}\!\!\right],
\end{equation}
where $x_{i}\in \mathbb{S}^2$ are the angle coordinates and
$y_{i}\in\mathbb{R}^2$ are the angular velocities after each kick.
Parameters $L$ and $M$ are constant $2\times 2$ matrices that depend
on the masses, lengths of rotor arms, and friction at the pivots,
while $c\in\mathbb{R}^2$ is a constant vector whose magnitude is
proportional to the kicking strength $f_0$.  In our numerical tests
we have used the same parameters as in~\cite{Romeiras92}, with the
kicking strength $f_{0} = 8.0$.

The following example illustrates the stabilising properties of the
transformations constructed on the basis of periodic orbits.  Let us
take a typical period-3 orbit point $x^\ast = (0.6767947,
5.8315697)$, $y^\ast = (0.9723920, -7.9998313)$ as a seed for
locating period-4 orbits.  The Jacobian matrix
$D\!f^3(x^\ast,y^\ast)$ of the seed has eigenvalues $\Lambda =
\mbox{diag}(206.48, -13.102, -0.000373, 0.000122)$. Therefore, based
on the scheme discussed in \S\ref{sec:ext}
Eqs.~(\ref{eq:gev}-\ref{eq:signs}), we can construct four STs $C$
corresponding to $(S_{11},S_{22})$ in Eq.~(\ref{eq:signs}) being
equal to $(+,+)$, $(-,+)$, $(+,-)$ and $(-,-)$.  Of the total of
2190 orbit points of period-4 (see Table \ref{tab:drmp}), the
transformations $C_1$, $C_2$, $C_3$, and $C_4$ stabilise \#$(1) =
532$, \#$(2)=544$, \#$(3)=474$, and \#$(4)=516$ orbit points,
respectively, and these sets of orbits are almost completely
non-overlapping.  That is, the number of orbits stabilised by both
$C_1$ and $C_2$ is \#$(1\cap 2) = 2$.  Similarly, \#$(1\cap 3) =
16$, \#$(1\cap 4) = 0$, \#$(2\cap 3) = 0$, \#$(2\cap 4) = 14$, and
\#$(3\cap 4) = 0$. On the other hand, the number of period-4 orbits
stabilised by at least one of the four transformations is \#$(1\cup
2\cup 3\cup 4) = 2034$.  This is a typical picture for other seeds
of period-3 as well as other periods.

This example provides evidence for the validity of our approach to
constructing the STs in high-dimensional systems based on periodic
orbits.  It also shows that, in the case of the double rotor map, a
single seed is sufficient for constructing transformations that
stabilise majority of the UPOs. Of course, in order to locate the
UPOs, we need to ensure that the seeds are in the convergence basins
of the stabilised periodic orbits.  That is why we need to use more
seeds to locate plausibly all periodic orbits of a given period.
Still, because of the enlarged basins of the stabilised orbits, the
number of seeds is much smaller than that required with iterative
schemes that do not use the STs.

\begin{table}
\caption{Number $n(p)$ of prime period-$p$ UPOs, and the number
$N(p)$ of fixed points of $p$-times iterated map for the kicked
double rotor map.  The asterisk for $p=8$ indicates that this set of
orbits is not complete.  Parameters $\epsilon_\mathrm{max}(p)$ and
$d_\mathrm{min}(p)$ are defined in Eqs.~(\ref{eq:maxe}) and
(\ref{eq:mind}).}
\begin{center}
\begin{tabular}{|l|r|r|c|c|} \hline
$p$ & $n(p)$~~ & $N(p)$~~ & $\epsilon_\mathrm{max}(p)$ &
$d_\mathrm{min}(p)$\\\hline
1 &      12 &       12 & $1.0\cdot10^{-14}$ & $1.3\cdot10^{0}$\\
2 &      45 &      102 & $5.9\cdot10^{-14}$ & $3.4\cdot10^{-1}$\\
3 &     152 &      468 & $5.8\cdot10^{-13}$ & $6.2\cdot10^{-2}$\\
4 &     522 &     2190 & $2.7\cdot10^{-12}$ & $6.9\cdot10^{-3}$\\
5 &    2200 &  11\,012 & $2.6\cdot10^{-11}$ & $1.1\cdot10^{-3}$\\
6 &    9824 &  59\,502 & $1.6\cdot10^{-10}$ & $1.8\cdot10^{-4}$\\
7 & 46\,900 & 328\,312 & $9.7\cdot10^{-10}$ & $9.1\cdot10^{-5}$\\
8$^\ast\!\!$ & 229\,082 & 1\,834\,566 & $1.2\cdot10^{-8}$ &
$5.5\cdot10^{-5}$\\\hline
\end{tabular}
\end{center}
\label{tab:drmp}
\end{table}

Compared to the total of 384 matrices in ${\mathcal
C}_{\mathrm{SD}}$, we use only two or four transformations for each
seed, depending on the number of unstable directions of the seed
orbit points. Yet, the application of the detection scheme outlined
in \S\ref{sec:numres} allows us to locate plausibly all periodic
orbits of the double rotor map up to period 7.  Table \ref{tab:drmp}
also includes the number of detected period-8 orbits that were used
as seeds to complete the detection of period 7.

The confidence with which we claim to have plausibly complete sets
of periodic orbits for each period is enhanced by the symmetry
consideration. That is, since the double rotor map is invariant
under the change of variables $(x, y) \mapsto (2\pi - x, -y)$, a
necessary condition for the completeness of the set of orbits for
each period is that for any orbit point $(x^{*},y^{*})$ the set also
contains an orbit point $(2\pi - x^{*},-y^{*})$.  Even though this
condition was not used in the detection scheme, we find that the
detected sets of orbits (apart from period 8) satisfy this symmetry
condition.  Of course, this condition is not sufficient to prove the
completeness of the detected sets of UPOs, but, combined with the
exhaustive search procedure presented above, provides a strong
indication of the completeness.

\subsection{Coupled H\'{e}non maps}\label{sec:chm}
Another system we use to test the efficacy of our approach is a
six-dimensional system of three coupled H\'{e}non maps (CHM)
\newabb{CHM},
\begin{equation}\label{eq:coupled1}
  x^{j}_{i+1} = a - (\tilde{x}^{j}_{i})^{2} + bx^{j}_{i-1},
  \quad\mathrm{for}\quad j=1,2,3\,,
\end{equation}
where $a = 1.4$ and $b = 0.3$ are the standard parameter values of
the H\'{e}non map and the coupling is given by
\begin{equation}
  \tilde{x}^{j}_{i} = (1-\epsilon)x^{j}_{i} + \frac{1}{2}\epsilon(
   x^{j+1}_{i} + x^{j-1}_{i}),
\end{equation}
with $x^0_i = x^3_i$ and $x^4_i = x^1_i$. We have chosen the
coupling parameter $\epsilon = 0.15$. Our choice of this system is
motivated by the work of Politi and Torcini~\cite{Politi92} in which
they locate periodic orbits in CHM for a small coupling parameter by
extending the method of Biham and Wenzel (BW)~\cite{Biham89}.  This
makes the CHM an excellent test system, since we can compare our
results against those for the BW method. The BW method defines the
following artificial dynamics
\begin{equation}\label{eq:EBWflow}
\dot{x}_{i}^{j}(t) = (-1)^{s(i,j)}\{x^{j}_{i+1}(t) - a +[
                      \tilde{x}^{j}_{i}(t)]^{2} - bx^{j}_{i-1}(t)\},
\end{equation}
with $s(i,j)\in\{0,1\}$.  Given the boundary condition $x_{p+1}^j =
x_1^j$, the equilibrium states of Eq.~(\ref{eq:EBWflow}) are the
period-$p$ orbits for the CHM.  The BW method is based on the
property that every equilibrium state of Eq.~(\ref{eq:EBWflow}) can
be made stable by one of the $2^{3p}$ possible sequences of $s(i,j)$
and, therefore, can be located by simply integrating
Eq.~(\ref{eq:EBWflow}) to convergence starting from the same initial
condition $x_i^j = 0.0$.  It is also found that, for the vast
majority of orbits, each orbit is stabilised by a unique sequence of
$s(i,j)$.

In order to reduce the computational effort, Politi and Torcini
suggest reducing the search to only those sequences $s(i,j)$ which
are allowed in the uncoupled system, i.e. with $\epsilon = 0$. This
reduction is possible because the introduction of coupling has the
effect of pruning some of the orbits found in the uncoupled
H\'{e}non map without creating any new orbits.

\begin{table}
\caption{The number of prime UPOs for the system of three coupled
H\'{e}non maps (CHM) detected by three different methods: BW -- full
Biham-Wenzel, BW-r -- reduced Biham-Wenzel, ST -- our method based
on stabilising transformations, Max -- maximum number of detected
UPOs obtained from all three methods and the system symmetry. See
text for details.}
\begin{center}\footnotesize
\begin{tabular}{|l|r|r|r|r|c|c|} \hline
$p$ &  BW & BW-r &   ST &  Max & $\epsilon_\mathrm{max}(p)$ &
$d_\mathrm{min}(p)$\\ \hline
 1 &    8 &    8 &    8 &    8 & $1.3\cdot10^{-14}$ & $9.9\cdot10^{-1}$\\
 2 &   28 &   28 &   28 &   28 & $4.6\cdot10^{-14}$ & $5.2\cdot10^{-1}$\\
 3 &    0 &    0 &    0 &    0 &    -               &       -          \\
 4 &   34 &   34 &   40 &   40 & $2.7\cdot10^{-8}$  & $4.2\cdot10^{-2}$\\
 5 &    0 &    0 &    0 &    0 &    -               &       -          \\
 6 &   74 &   74 &   72 &   74 & $9.5\cdot10^{-10}$ & $8.6\cdot10^{-3}$\\
 7 &   28 &   28 &   28 &   28 & $1.0\cdot10^{-8}$  & $5.6\cdot10^{-3}$\\
 8 &  271 &  271 &  285 &  286 & $1.1\cdot10^{-6}$  & $5.5\cdot10^{-3}$\\
 9 &    - &   63 &   64 &   66 & $9.9\cdot10^{-7}$  & $2.6\cdot10^{-4}$\\
10 &    - &  565 &  563 &  568 & $1.3\cdot10^{-8}$  & $4.1\cdot10^{-4}$\\
11 &    - &  272 &  277 &  278 & $7.1\cdot10^{-9}$  & $5.4\cdot10^{-4}$\\
12 &    - & 1972 & 1999 & 1999 & $2.5\cdot10^{-6}$  & $4.3\cdot10^{-4}$\\
13$^\ast$& - & - & 1079 &   -  & $8.6\cdot10^{-8}$  & $4.0\cdot10^{-4}$\\
14$^\ast$& - & - & 6599 &   -  & $2.3\cdot10^{-6}$  & $3.5\cdot10^{-4}$\\
15$^\ast$& - & - & 5899 &   -  & $7.0\cdot10^{-6}$  &
$1.5\cdot10^{-4}$\\\hline
\end{tabular}
\end{center}
\label{tab:chm}
\end{table}

We have implemented the BW method with both the full search and the
reduced search (BW-r) up to as high a period as is computationally
feasible (see Table~\ref{tab:chm}).  In the case of the full search
we detect UPOs up to period 8 and in the case of the reduced search
up to period 12.  The seed $x_i^j = 0.0$ was used for all periods
except for period 4, where it was found that with this seed both BW
and BW-r located only 28 orbits.  We found a maximum of 34 orbits
using the seed $x_i^j = 0.5$.  It is possible that more orbits can
be found with different seeds for other periods as well, but we have
not investigated this.  The example of period 4 illustrates that,
unlike for a single H\'{e}non map, the Biham-Wenzel method fails to
detect all orbits from a single seed.

Even though our approach (labeled ``ST'' in Table~\ref{tab:chm}) is
general and does not rely on the special structure of the H\'{e}non
map, its efficiency far surpasses the full BW method and is
comparable to the reduced BW method.  Except for periods 6 and 10,
the ST method locates the same or larger number of
orbits.\footnote{The precise reason for the failure of the ST method
to detect all period 6 and 10 orbits needs further investigation. We
believe that the orbits that were not detected have
uncharacteristically small convergence basins with any of the
applied stabilising transformations.}

Unlike the double rotor map, the CHM possesses very few periodic
orbits for small $p$, particularly for odd values of $p$. Therefore,
we found that the direct application of the detection strategy
outlined at the beginning of \S\ref{sec:numres} would not allow us
to complete the detection of even period orbits. Therefore, for even
periods $p$ we also used $p+2$ as seeds and, in case of period 12, a
few remaining orbits were located with seeds of period 15.  We did
not attempt to locate a maximum possible number of UPOs for $p >
12$.  The numbers of such orbits (labeled with asterisks) are listed
in Table~\ref{tab:chm} for completeness.

As with the double rotor map, we used the symmetry of the CHM to
test the completeness of the detected sets of orbits.  It is clear
from the definition of the CHM that all its UPOs are related by the
permutation symmetry (i.e., six permutations of indices $j$).  The
column labeled ``Max'' in Table~\ref{tab:chm} lists the maximum
number of UPOs that we were able to find using all three methods and
applying the permutation symmetry to find any UPOs that might have
been missed.  As can be seen in Table~\ref{tab:chm}, only a few
orbits remained undetected by the ST method.

Concluding this Section, we would like to point out that the high
efficiency of the proposed method is primarily due to the fact that
each ST constructed based on the stability properties of the seed
orbit substantially increases the basins of convergence of orbits
stabilised by this transformation.  This is apparent in a typical
increase of the fraction of converged seeds with the increasing
value of parameter $\beta$ in Eq.~(\ref{eqn:davlai}). For example,
when detecting period-10 orbits of CHM using period-12 orbits as
seeds, the fraction of seeds that converge to periodic orbits grows
from 25-30\% for small $\beta$ (essentially the Newton-Raphson
method) to about 70\% for the optimal value of $\beta$.

\section{Summary}
In this chapter we have presented a new scheme for constructing
stabilising transformations~\cite{Crofts06} which can be used 
to locate periodic orbits in chaotic maps with the iterative scheme 
given by Eq.~(\ref{eqn:davlai}).  The scheme is based on the understanding 
of the relationship between the STs and the properties  of eigenvalues 
and eigenvectors of the stability matrices of the periodic orbits. 
Of particular significance is the observation that only the unstable 
eigenvalues are important for determining the STs. 
Therefore, unlike the original set of transformations proposed by 
Schmelcher and Diakonos, which grows with the system size as $2^n n!$, 
our set has cardinality of at most $2^k$, where $k$ is the maximum 
number of unstable eigenvalues (i.e.~the maximum dimension of the 
unstable manifold). It is also apparent that, while the SD set contains 
a large fraction of transformations that do not stabilise any UPOs of
a given system, all of our transformations stabilise a significant
subset of UPOs. The dependence of the number of transformations on
the dimensionality of the unstable manifold rather than on the
system dimensionality is especially important in cases when we study
low-dimensional chaotic dynamics embedded in a high-dimensional
phase space.  This is often the case in systems obtained from
time-space discretisation of nonlinear partial differential
equations (e.g. the Kuramoto-Sivashinsky equation). Application of
the STs approach to such high-dimensional chaotic systems will be
the subject of Chapter \ref{ch:kse}.

\chapter{Extended systems: Kuramoto-Sivashinsky equation}
\label{ch:kse}
\begin{quote}
The mind uses its faculty for creativity only when experience forces
it to.\\ \emph{H. J.~Poincar\'{e}}
\end{quote}
In this chapter we extend the ideas presented in Chapter
\ref{ch:stabtrans} so as to efficiently compute unstable periodic
orbits (UPO) in large-scale dynamical systems arising from the
spatial discretisation of parabolic PDEs~\cite{Crofts07}. Following the
approach often adopted in subspace iteration methods (see~\cite{Lust98,Shroff93}
and references therein) we construct a decomposition of the tangent
space into unstable and stable orthogonal subspaces. On the unstable
subspace we apply the method of stabilising transformations (ST),
whilst Picard iteration is performed on the complement. The method is
extremely effective when the dimension of the unstable subspace is
small compared to the system dimension. We apply the new scheme to the
model example of a chaotic spatially extended system -- the
Kuramoto-Sivashinsky equation.

\section{Subspace decomposition}\label{sec:sub}
Consider the solution of the nonlinear system
\begin{equation}\label{ch3:eqn:nlin}
    f(x) - x = 0, \quad x\in\mathbb{R}^n, \quad
    f:\mathbb{R}^n\rightarrow\mathbb{R}^n,
\end{equation}
where $f(x)$ is assumed twice differentiable in the neighbourhood of
$x^*$, an isolated root of Eq.~(\ref{ch3:eqn:nlin}). We can
approximate the solution of (\ref{ch3:eqn:nlin}) by a recursive {\em
fixed point} procedure of the form
\begin{equation}\label{ch3:eqn:picard1}
    x_{i+1} = f(x_i), \quad i = 1,2,3,\dots.
\end{equation}
It is well known that the iteration (\ref{ch3:eqn:picard1})
converges locally in the neighbourhood of a solution $x^*$, as long
as all the eigenvalues of the Jacobian $Df(x^*)$ lie within the unit
disc $\{z\in\mathbb{C} : |z| < 1\}$. In contrast,
(\ref{ch3:eqn:picard1}) typically diverges if $Df(x^*)$ has an
eigenvalue outside the unit disc. In that case, a popular
alternative -- as discussed in Chapter \ref{ch:detection} -- is to
employ Newton's method
\begin{eqnarray}\label{ch3:eqn:newton}
    (Df(x_i) - I_n)\delta x_i &=& -(f(x_i) - x_i),\\
    x_{i+1} &=& x_i +\delta x_i, \quad i = 1, 2, 3, \dots.
\end{eqnarray}

The idea of subspace iterations is to exploit the fact that the
divergence of the fixed point iteration (\ref{ch3:eqn:picard1}) is
due to a small number of eigenvalues, $n_u$, lying outside the unit
disc. By decomposing the space $\mathbb{R}^n$ into the direct sum of
the unstable subspace spanned by the eigenvectors of $Df(x^*)$
\begin{equation}\label{ch3:eqn:unstable}
    \mathbb{P} = \mathrm{Span}\{e_k\in\mathbb{R}^n: Df(x^*)e_k = \lambda_k e_k, |\lambda_k| > 1\}
\end{equation}
and its orthogonal complement, $\mathbb{Q}$, a modified iterative
scheme is obtained. The application of Newton's method to the
subspace $\mathbb{P}$ whilst continuing to use the relatively cheap
fixed point iteration on the subspace $\mathbb{Q}$, results in a
highly efficient scheme provided $\mathrm{dim}(\mathbb{P}) \ll
\mathrm{dim}(\mathbb{Q})$.

To this end, let $V_p\in\mathbb{R}^{n\times n_u}$ be a basis for the
subspace $\mathbb{P}\subset\mathbb{R}^n$ spanned by the eigenvectors
of $DF(x^*)$ corresponding to those eigenvalues lying outside the
unit disc, and $V_q\in\mathbb{R}^{n\times n_s}$ a basis for
$\mathbb{Q}$, where $n_u+n_s=n$. Then, we can define orthogonal
projectors $P$ and $Q$ onto the respective subspaces, $\mathbb{P},
\mathbb{Q}$, as follows
\begin{eqnarray}
  P &=& V_pV_p^{\mathsf{T}},\\
  \label{ch3:eqn:orthog}
  Q &=& V_qV_q^{\mathsf{T}} = I_n - P.
\end{eqnarray}
Note that any $x\in\mathbb{R}^n$ admits the following unique
decomposition
\begin{equation}\label{ch3:eqn:xdec}
    x = V_p\bar{p} + V_q\bar{q} = p + q, \quad p := V_p\bar{p} = Px,
    \quad q := V_q\bar{q} = Qx,
\end{equation}
with $\bar{p}\in\mathbb{R}^{n_u}$ and $\bar{q}\in\mathbb{R}^{n_s}$.
Substituting (\ref{ch3:eqn:xdec}) in Eq.~(\ref{ch3:eqn:newton}) and
multiplying the result by $[V_q, V_p]^{\mathsf{T}}$ on the left, one
obtains
\begin{equation}\label{ch3:eqn:newtsub}
\left[
  \begin{array}{cc}
    V_q^{\mathsf{T}}DfV_q - I_{n_s} & 0 \\
    V_p^{\mathsf{T}}DfV_q & V_p^{\mathsf{T}}DfV_p - I_{n_u} \\
  \end{array}
\right] \left[
  \begin{array}{c}
    \Delta\bar{q} \\
    \Delta\bar{p} \\
  \end{array}
\right] = -\left[
  \begin{array}{c}
    V_q^{\mathsf{T}}f - \bar{q} \\
    V_p^{\mathsf{T}}f - \bar{p}\\
  \end{array}
\right].
\end{equation}
Here we have used the fact that $V_p^{\mathsf{T}}V_q = 0_{n_u\times
n_s}$, $V_q^{\mathsf{T}}V_p = 0_{n_s\times n_u}$, and
$V_q^{\mathsf{T}}DfV_p = 0_{n_s\times n_u}$ the latter holding due
to the invariance of $Df$ on the subspace $\mathbb{P}$. Now, the
first $n_s$ equations in (\ref{ch3:eqn:newtsub}) may be solved using
the following fixed point iteration scheme
\begin{eqnarray}\label{ch3:eqn:picard2}\nonumber
  \Delta\bar{q}^{[0]} &=& 0, \\
  \Delta\bar{q}^{[i]} &=& V_q^{\mathsf{T}}DfV_q\Delta\bar{q}^{[i-1]} + V_q^{\mathsf{T}}f - \bar{q},
  \\\nonumber
  \Delta\bar{q} &=& \Delta\bar{q}^{[l]} =
  \sum_{i=0}^{l-1}(V_q^{\mathsf{T}}DfV_q)^i(V_q^{\mathsf{T}}f - \bar{q}),
\end{eqnarray}
where $l$ denotes the number of fixed point iterations taken per
Newton-Raphson step. Since $r_{\sigma}[V_q^{\mathsf{T}}DfV_q] < 1$
by construction, the iteration (\ref{ch3:eqn:picard2}) will be
locally convergent on $\mathbb{Q}$ in some neigbourhood of
$\Delta\bar{q}$ -- here $r_{\sigma}[\cdot]$ denotes the spectral
radius. In order to determine $\Delta\bar{p}$ one solves
\begin{equation}
    (V_p^{\mathsf{T}}DfV_p - I_{n_u})\Delta\bar{p} =
    -V_p^{\mathsf{T}}f + \bar{p} - V_p^{\mathsf{T}}DfV_q\Delta\bar{q}.
\end{equation}
Note that in practice only one iteration of (\ref{ch3:eqn:picard2})
is performed~\cite{Lust98}, i.e. $l = 1$, this leads to the
following simplified system to solve for the correction
$[\Delta\bar{q}, \Delta\bar{p}]^{\mathsf{T}}$
\begin{equation}
\left[\begin{array}{cc}
    -I_{n_s} & 0 \\
    V_p^{\mathsf{T}}DfV_q & V_p^{\mathsf{T}}DfV_p - I_{n_u} \\
  \end{array}
\right]
\left[\begin{array}{c}
    \Delta\bar{q} \\
    \Delta\bar{p} \\
  \end{array}
\right] = -
\left[\begin{array}{c}
    V_q^{\mathsf{T}}f - \bar{q} \\
    V_p^{\mathsf{T}}f - \bar{p}\\
  \end{array}
\right].
\end{equation}

Key to the success of the above algorithm is the accurate
approximation of the eigenspace corresponding to the unstable modes.
In order to construct the projectors $P$, $Q$, the Schur
decomposition is used. However, primary concern of the work
in~\cite{Lust98,Shroff93} is the continuation of branches of
periodic orbits, where it is assumed that a reasonable approximation
to a UPO is known. Since we have no knowledge {\em a priori} of the
orbits whereabouts we shall need to accommodate this into our
extension of the method to detecting UPOs.

\subsection{Stabilising transformations}\label{ch3:sec:stabtrans}
An alternative approach is supplied by the method of STs, where in
order to detect equilibrium solutions of Eq.~(\ref{ch3:eqn:nlin}) we
introduce the associated flow
\begin{equation}\label{ch3:eqn:flow1}
    \frac{dx}{ds} = g(x).
\end{equation}
Here $g(x) = f(x) - x$. With this setup we are able to stabilise all
UPOs $x^*$ of Eq.~(\ref{ch3:eqn:nlin}) such that all the eigenvalues
of the Jacobian $Df(x^*)$ have real part smaller than one. In order
to stabilise all possible UPOs we study the following flow
\begin{equation}\label{ch3:eqn:flow2}
 \frac{dx}{ds} = Cg(x),
\end{equation}
where $C\in\mathbb{R}^{n\times n}$ is a constant matrix introduced
in order to stabilise UPOs with the Jacobians that have eigenvalues
with real parts greater than one.

Substituting (\ref{ch3:eqn:xdec}) in Eq.~(\ref{ch3:eqn:flow1}) and
multiplying the result by $[V_q, V_p]^{\mathsf{T}}$ on the left, one
obtains
\begin{eqnarray}\label{ch3:eqn:1}
  \frac{d\bar{q}}{ds} &=&   V_q^{\mathsf{T}}g, \\\label{ch3:eqn:2}
  \frac{d\bar{p}}{ds} &=& V_p^{\mathsf{T}}g.
\end{eqnarray}
Thus we have replaced the original Eq.~(\ref{ch3:eqn:flow1}) by a
pair of coupled equations, Eq.~(\ref{ch3:eqn:1}) of dimension $n_s$
and Eq.~(\ref{ch3:eqn:2}) of dimension $n_u$. Since
\begin{eqnarray*}
  \frac{\partial}{\partial\bar{q}}(V_q^{\mathsf{T}}g) &=& V_qDg\frac{\partial x}{\partial\bar{q}}, \\
            &=& V_q^{\mathsf{T}}DgV_q,
\end{eqnarray*}
and $r_{\sigma}[V_q^{\mathsf{T}}DgV_q] < 0$ by construction, it
follows that in order to detect all UPOs of
Eq.~(\ref{ch3:eqn:nlin}), it is sufficient to solve
\begin{eqnarray}\label{ch3:eqn:3}
  \frac{d\bar{q}}{ds} &=& V_q^{\mathsf{T}}g, \\\label{ch3:eqn:4}
  \frac{d\bar{p}}{ds} &=& \tilde{C}V_p^{\mathsf{T}}g,
\end{eqnarray}
where $\tilde{C}\in\mathbb{R}^{n_u\times n_u}$ is a constant matrix.
\begin{figure}[t]
\begin{center}
\subfigure[Schur decomposition]{
\includegraphics[height=8.5cm,width=6.5cm]{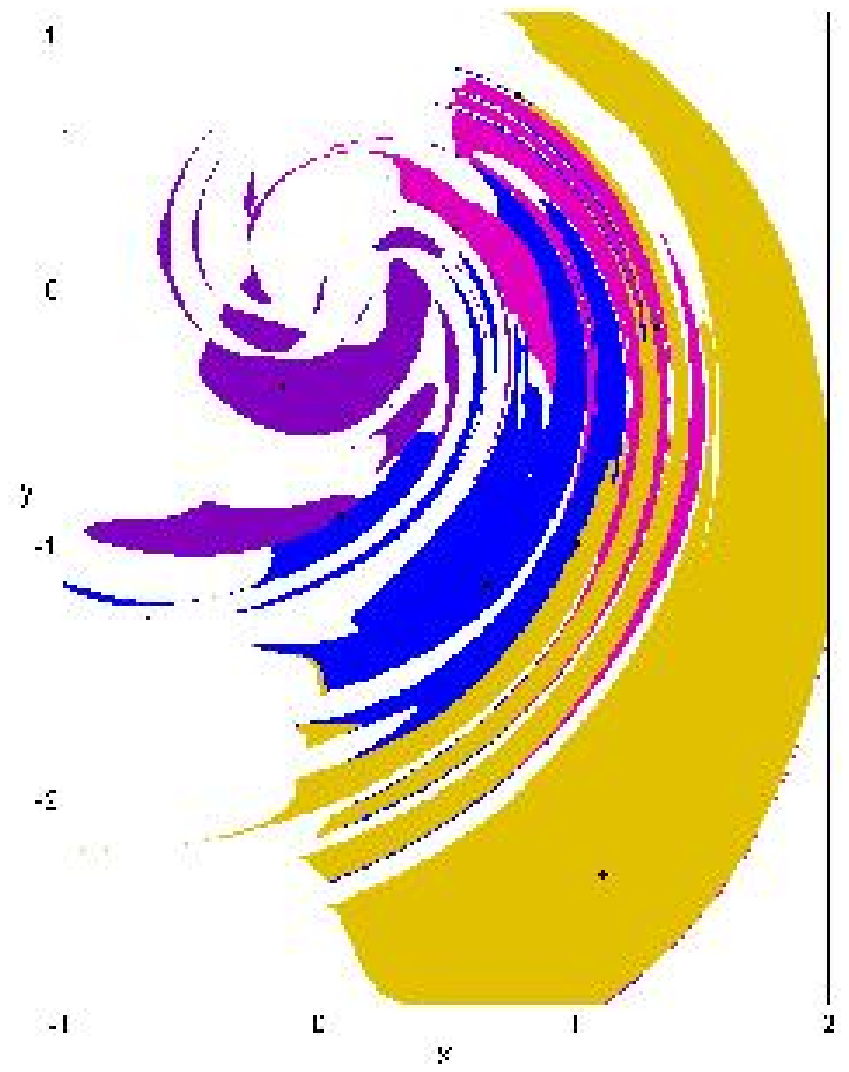}}
\hspace{0.5cm} \subfigure[Singular value decomposition]{
\includegraphics[height=8.5cm,width=6.5cm]{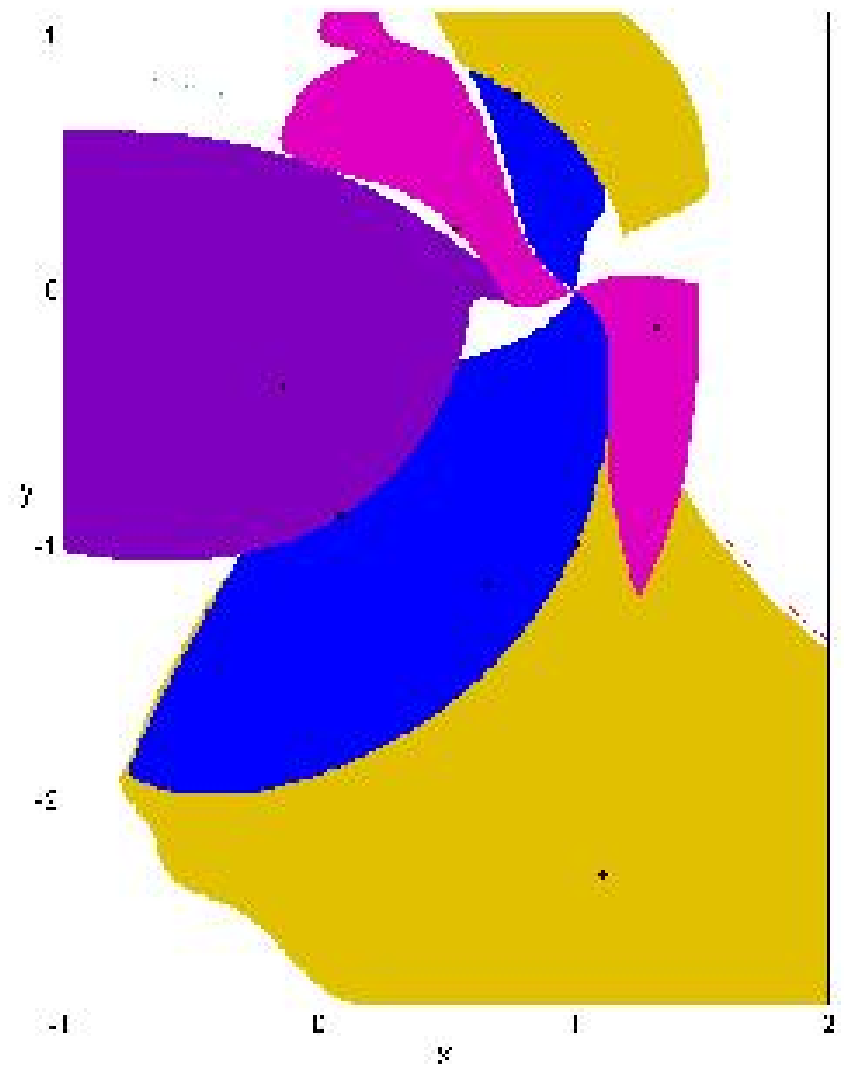}}
\caption{Basins of attraction for the period$-3$ orbits of the Ikeda
map with parameter values $a = 1.0$, $b = 0.9$, $k = 0.4$ and $\eta
= 6.0$. Here we have chosen $\tilde{C} = -1$ since in this example
the unstable subspace is one-dimensional.} \label{ch3:fig:basin}
\end{center}
\end{figure}

In~\cite{Lust98,Shroff93} the Schur decomposition (SchD)
\newabb{SchD} is used in order to construct the projectors $P$ and
$Q$. This is fine for continuation problems since one may assume
from the offset that they possess an initial condition $x_0$
sufficiently close to a UPO such that the SchD of $Df(x_0)$ gives a
good approximation to the eigenspace of $Df(x^*)$. However, it is
well known that the eigenvectors of the perturbed Jacobian
$Df(x^*+\delta x)$ behave erratically as we increase $\delta x$. In
order to enlarge the basins of attraction for the UPOs we propose
that singular value decomposition (SVD)
\newabb{SVD} be used instead. That is we choose an initial condition
$x_0$ and construct the SVD of its preimage, i.e. $Df(f^{-1}(x_0)) =
USW^{\mathsf{T}}$ (or in the continuous case
$D\phi^T(\phi^{-T}(x_0)) = USW^{\mathsf{T}}$ for some time $T$), the
columns of $U$ give the stretching directions of the map at $x_0$,
whilst the singular values determine whether the directions are
expanding or contracting. It is these directions which we use to
construct the projectors $P$ and $Q$. Due to the robustness of the
SVD we expect the respective basins of attraction to increase.

It is not necessary in practice to decompose
Eq.~(\ref{ch3:eqn:flow1}) in order to apply the new ST. Rather we
can express $C$ in terms of $\tilde{C}$ and $V_p$. To see this we
add $V_q$ times Eq.~(\ref{ch3:eqn:3}) to $V_p$ times
Eq.~(\ref{ch3:eqn:4}) to get
\begin{eqnarray}\nonumber
\frac{dx}{ds} &=& V_qV_q^{\mathsf{T}}g(x) + V_p\tilde{C}V_p^{\mathsf{T}}g(x),\\
              &=& [I_n + V_p(\tilde{C} - I_{n_u})V_p^{\mathsf{T}}]g(x),\label{ch3:eqn:C}
\end{eqnarray}
where the second line follows from Eq.~(\ref{ch3:eqn:orthog}). From
this we see that the following choice of $C$ is equivalent to the
preceding decomposition
\begin{equation}\label{ch3:eqn:Cmat}
C = I_n + V_p(\tilde{C} - I_{n_u})V_p^{\mathsf{T}}.
\end{equation}
Thus in practice we compute $V_p$ and $\tilde{C}$ at the seed $x_0$
in order to construct $C$ and then proceed to solve
Eq.~(\ref{ch3:eqn:flow2}).

The advantage of using the SVD rather than the SchD can be
illustrated by the following example. Consider the Ikeda
map~\cite{Ikeda79}:
\begin{equation}
f(\mathbf{x}) := \left[\!\!\begin{array}{c} x_{i+1}\\
y_{i+1}\end{array}
\!\!\right] = \left[\!\!\begin{array}{c} a + b(x_i\cos{(\phi_i)} - y_i\sin{(\phi_i)})\\
 b(x_i\sin{(\phi_i)} + y_i\cos{(\phi_i)})\end{array}\!\!\right],\label{ch3:eqn:ikeda}
\end{equation}
where $\phi_i = k-\eta/(1+x_i^2+y_i^2)$ and the parameters are
chosen such that the map has a chaotic attractor: $a=1.0$, $b=0.9$,
$k=0.4$ and $\eta = 6.0$.  For this choice of parameters the Ikeda
map possesses eight period$-3$ orbit points (two period-3 orbits and
two fixed points, one of which is on the attractor basin boundary).
In our experiments we have covered the attractor for
Eq.~(\ref{ch3:eqn:ikeda}) with a grid of initial seeds and solved
the associated flow for $p=3$, i.e., $g(x) = f^3(x) - x$. This is
done twice, firstly in the case where the projections $P$ and $Q$
are constructed via the SchD and secondly when they are constructed
through the SVD. Since all UPOs of the Ikeda map are of saddle type,
the unstable subspace is one-dimensional and we need only two
transformations: $\tilde{C} = 1$ and $\tilde{C} = -1$. Figure
\ref{ch3:fig:basin} shows the respective basins of attraction for
the two experiments with $\tilde{C} = -1$. It can be clearly seen
that the use of SVD corresponds to a significant increase in basin
size compared to the SchD. Note that with $\tilde{C} = -1$ we
stabilise four out of eight fixed points of $f^3$. The other four
are stabilised with $\tilde{C} = 1$. The corresponding basins of
attraction are shown in Figure \ref{ch3:fig:id}. Note that for this
choice of ST the choice of basis vectors is not important, since
Eq.~(\ref{ch3:eqn:Cmat}) yields $C = I$, so that the associated flow
is given by Eq.~(\ref{ch3:eqn:flow1}).
\begin{figure}[t]
\begin{center}
\includegraphics[height=9cm,width=7cm]{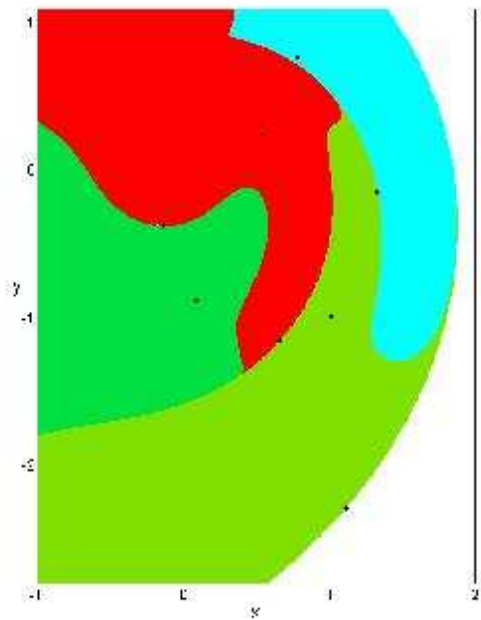}
\caption{The basins of attraction for the Ikeda map for the choice
of $\tilde{C} = 1$. Fixed points of $f^3$ with negative unstable
eigenvalues are stable stationary solutions of the associated flow,
while those with positive eigenvalues are saddles located at the
basin boundaries.} \label{ch3:fig:id}
\end{center}
\end{figure}

\section{Implementation} \label{sec:imp}
We now wish to apply these ideas to parabolic PDEs of the form
(\ref{eqn:evolution}). The numerical solution of such PDEs are
typically based upon their representation in terms of a truncated
system of nonlinear ODEs. Thus, we will be concerned with the
detection of UPOs for large systems of ODEs:
\begin{equation}\label{ch3:eqn:ode1}
    \frac{dx}{dt} = v(x), \quad x\in \mathbb{R}^n,
\end{equation}
where $v$ is derived from an appropriate discretisation procedure
(i.e. finite differences, finite element, spectral method) and $n\gg
1$.

A typical approach in the determination of UPOs for flows is via a
Poincar\'e surface of section (PSS). By ``clever'' placement of an
$(n-1)$-dimensional manifold in the phase space, the problem is
reduced to a discrete map defined via intersections with the
manifold. However, a correct choice of PSS is a challenging problem
in itself. Due to the complex topology of a high-dimensional phase
space, the successful detection of UPOs will be highly dependent
upon the choice of section. When the choice of a suitable PSS is not
obvious {\em a priori}, we found it preferable to work with the full
flow, adding an auxiliary equation to determine the integration time
$T$.

Let $x\mapsto\phi^t(x)$ denote the flow map of
Eq.~(\ref{ch3:eqn:ode1}). Then we define the associated flow as
follows
\begin{eqnarray}\label{ch3:eqn:assflow}\nonumber
    \frac{dx}{ds} &=& Cg\\ &=& C(\phi^T(x) - x).
\end{eqnarray}
The additional equation for $T$ is constructed such that $T$ is
always changing in the direction that decreases $|\phi^T(x) - x|$,
i.e.
\begin{equation*}
    \frac{dT}{ds} \varpropto -\frac{\partial |g|^2}{\partial T},
\end{equation*}
or, more precisely
\begin{equation}\label{ch3:eqn:Teqn}
    \frac{dT}{ds} = -\alpha
    v(\phi^T(x))\cdot(\phi^T(x)-x).
\end{equation}
Here $\alpha > 0$ is a constant which controls the relative speed of
convergence of Eq.~(\ref{ch3:eqn:Teqn}). This leads to the following
augmented flow which must be solved to detect UPOs of
(\ref{ch3:eqn:ode1}):
\begin{equation}\label{ch3:eqn:augflow}
\frac{d}{ds}
    \left[
      \begin{array}{c}
        x \\
        T \\
      \end{array}
    \right] =
    \left[
      \begin{array}{c}
        C(\phi^T(x) - x) \\
        -\alpha v(\phi^T(x))\cdot(\phi^T(x) - x) \\
      \end{array}
    \right].
\end{equation}
Note that the augmented flow (\ref{ch3:eqn:augflow}) is derived by
integrating a nonlinear PDE -- in our case the KSE -- for some time
$T$ and will become increasingly stiff for larger $T$.

Several approaches have been proposed for the solution of stiff
systems of ODEs; see for example, the review by Shampine and Gear
\cite{Gear79}. Of all these techniques the general-purpose codes
contained within the ODEPACK software package \cite{ODEPACK} are
regarded as some of the best available routines for the solution of
such systems.  Thus, in our numerical experiments we use the stiff
solver {\tt dlsodar} from the ODEPACK toolbox to integrate
(\ref{ch3:eqn:augflow}). {\tt dlsodar} is a variable step-size
solver which automatically changes between stiff and nonstiff modes.
In particular, as we approach a steady state of
Eq.~(\ref{ch3:eqn:augflow}) {\tt dlsodar} will take increasingly
larger time-steps, leading to super linear convergence in the
neighbourhood of the solution.

To use the solver {\tt dlsodar}, we must provide a routine that
returns the value of the vector field (\ref{ch3:eqn:augflow})
evaluated at a given point $(x, T)$. Here we need the solution of
Eq.~(\ref{ch3:eqn:ode1}) which is obtained by applying a suitable
numerical integration scheme; see the next section and Appendix
\ref{ch:appendix5} for further details. The ODEPACK software package
makes use of the Jacobian matrix of the system being integrated and
provides the option of computing the Jacobian via finite differences
or via a user supplied routine. Note that in the case that the flow
is expected to be stiff much of the time, it is recommended that a
routine for the Jacobian is supplied and we do this. The derivative
of (\ref{ch3:eqn:augflow}) with respect to $(x, T)$ is given by
\begin{equation}\label{ch3:eqn:jacaug}
 \left[ \begin{array}{cc}
  C(J_T - I_n) & Cv_T \\
  -\alpha( v_T^{\mathsf{T}}(J_T - I_n) + g^{\mathsf{T}}DvJ_T) & -\alpha(g^{\mathsf{T}}Dvv_T
  + v_T^{\mathsf{T}}v_T) \\
 \end{array}\right],
\end{equation}
where $J_T = \partial\phi^T(x)/\partial x $, $v_T = v(\phi^T(x))$,
$Dv = dv/dx$ and $g = \phi^T(x) - x$ as usual.

Quite often one might wish to terminate simulation before the usual
stopping criteria of, for example, a maximum number of steps taken
or certain tolerances having been reached. A useful feature of the
{\tt dlsodar} algorithm is that it allows the optional user supplied
routine to do just this. To be more exact, it determines the roots
of any of a set of user supplied functions
$$h_i = h_i(t,x_1,\dots,x_n),\quad i = 1,\dots,m,$$
and returns the solution of (\ref{ch3:eqn:augflow}) at the root if
it occurs prior to the normal stopping criteria.

Note that, to increase the efficiency of the algorithm we wish to
avoid the following two instances: firstly, due to the local nature
of the STs, we should stop the search if we wander too far from the
initial condition, and secondly, since our search is governed by the
dynamics of Eq.~(\ref{ch3:eqn:assflow}) and not by those of
(\ref{ch3:eqn:ode1}), we might move off the attractor after some
time period so the convergence to a UPO becomes highly unlikely. In
our numerical experiments we supply the following function
\begin{equation}\label{ch3:eqn:normg}
    h_1  = a - |g|,
\end{equation}
where $a\in\mathbb{R}$ is a constant and $|\cdot|$ denotes the $L_2$
norm. In practice, we have found that there exists a threshold value
of $a$, such that convergence is highly unlikely once the norm of
$g$ surpasses it. Note that we also restrict the maximum number of
allowed integration steps since the convergence becomes less likely
if the associated flow is integrated for a long time.

We must also provide two tolerances, rtol and atol, which control
the local error of the ODE solver. In particular, the estimated
local error in $X = (x, T)$ will be controlled so as to be less than
$$\mathrm{rtol}\cdot||X||_{\infty} + \mathrm{atol}.$$
Thus the local error test passes if, in each component, either the
absolute error is less than atol or the relative error is less than
rtol. The accuracy with which we would like to solve the flow
(\ref{ch3:eqn:augflow}) is determined by the stability properties of
(\ref{ch3:eqn:ode1}). To understand this, we note that in evaluating
the RHS of (\ref{ch3:eqn:augflow}) it is the solution of
Eq.~(\ref{ch3:eqn:ode1}) at time $T$, i.e. $\phi^T(x)$, which is
critical for error considerations. Suppose that our initial point
lies within $\delta x$ of a true trajectory $x$. Then
$\phi^T(x+\delta x)$ lies approximately within $e^{\lambda T}\delta
x$ of the true trajectory, $\phi^T(x)$, where $\lambda$ is the
largest Lyapunov exponent of the system. Since $\lambda$ is positive
for chaotic systems, the error grows exponentially with the period
and we should take this into account when setting the tolerances
rtol and atol. This leads us to the following settings for the
tolerances
\begin{equation}\label{ch3:eqn:tolerances}
    \mathrm{rtol} = \mathrm{atol} = 10^{-5}/e^{\lambda T_0},
\end{equation}
where $T_0$ is the initial period and $\lambda$ is the largest
Lyapunov exponent of the flow $v$. We have computed the Lyapunov
exponent using the algorithm due to Benettin {\em et al}
\cite{Benettin80}; see Appendix \ref{ch:appendix6} for a description
of the routine as well as a brief review of Lyapunov exponents.

\subsection{Kuramoto-Sivashinsky equation}\label{sec:kse}
We have chosen the Kuramoto-Sivashinsky equation (KSE)\newabb{KSE}
for our numerical experiments. It is the simplest example of
spatiotemporal chaos and has been studied in a similar context
in~\cite{Christansen97,Lan04,Zoldi98}, where the detection of many
UPOs has been reported. We work with the KSE in the form
\begin{equation}\label{ch3:eqn:kse}
    u_t = -\frac{1}{2}(u^2)_x - u_{xx} - u_{xxxx},
\end{equation}
where $x\in [0, L]$ is the spatial coordinate, $t\in\mathbb{R}^{+}$
is the time and the subscripts $x$, $t$ denote differentiation with
respect to space and time.  For $L < 2\pi$, $u(x,t) = 0$ is the
global attractor for the system and the resulting long time dynamics
are trivial. However, for increasing $L$ the system undergoes a
sequence of bifurcations leading to complicated dynamics; see for
example~\cite{Kevrekidis90}.

Our setup will be close to that found in~\cite{Lan04}. In what
follows we assume periodic boundary conditions: $u(x,t) = u(x+L,t)$,
and restrict our search to the subspace of antisymmetric solutions,
i.e. $u(x,t) = -u(L-x,t)$. Due to the periodicity of the solution,
we can solve Eq.~(\ref{ch3:eqn:kse}) using the pseudo-spectral
method~\cite{HairerBook,TrefethenBook}. Representing the function
$u(x,t)$ in terms of its Fourier modes:
\begin{equation}
  \quad u(x,t) := {\cal F}^{-1}[\hat{u}] = \sum_{k\in{\mathbb Z}}
\hat{u}_k e^{-ikqx},
\end{equation}
where
\begin{equation}
\hat{u}:=(\dots,\hat{u}_{-1},\hat{u}_0,\hat{u}_1,\dots)^{\mathsf{T}}\,,\qquad
     \hat{u}_k := {\cal F}[u]_k = \frac{1}{L}\int_0^L u(x,t)
e^{ikqx}dx,
\end{equation}
we arrive at the following system of ODEs
\begin{equation}
  \frac{d\hat{u}_k}{dt} =[(kq)^2-(kq)^4]\hat{u}_k +
  \frac{ikq}{2}{\cal F}[({\cal F}^{-1}[\hat{u}])^2]_k\,.
\end{equation}
Here $q = 2\pi/L$ is the basic wave number. Since $u$ is real, the
Fourier modes are related by $\hat{u}_{-k} = \hat{u}^\ast_k$.
Furthermore, since we restrict our search to the subspace of odd
solutions, the Fourier modes are pure imaginary, i.e.
$\mathfrak{Re}(\hat{u}_{k}) = 0$.

The above system is truncated as follows: the Fourier transform
${\cal F}$ is replaced by its discrete equivalent
\begin{equation}
  a_k := {\cal F}_N[u]_k = \sum_{j = 0}^{N-1} u(x_j)
  e^{ikqx_j}\,,\qquad u(x_j) := {\cal F}_N^{-1}[a]_j
  = \frac{1}{N}\sum_{k = 0}^{N-1} a_j e^{-ikqx_j}\,,
\end{equation}
where $x_j = L/N$ and $a_{N-k} = a^\ast_k$. Since $a_0 = 0$ due to
Galilean invariance and setting $a_{N/2} = 0$ (assuming $N$ is
even), the number of independent variables in the truncated system
is $n = N/2-1$.  The truncated system looks as follows:
\begin{equation}\label{ch3:eqn:finite}
  \dot{a}_k =[(kq)^2-(kq)^4]a_k +
  \frac{ikq}{2}{\cal F}_N[({\cal F}_N^{-1}[a])^2]_k\,,
\end{equation}
with $k = 1,\ldots,n$, although in the Fourier transform we need to
use $a_k$ over the full range of $k$ values from 0 to $N-1$.

The discrete Fourier transform ${\cal F}_N$ can be computed using
fast Fourier transform (FFT).\newabb{FFT} In Fortran and C, the
routine {\tt REALFT} from Numerical Recipes~\cite{NumericalRecipes}
can be used. In Matlab, it is more convenient to use complex
variables for $a_k$. Note that Matlab function {\tt fft} is, in
fact, the inverse Fourier transform.

To derive the equation for the matrix of variations, we use the fact
that ${\cal F}_N$ is a linear operator to obtain
\begin{equation}\label{ch3:eqn:varform}
  \frac{\partial \dot{a}_k}{\partial a_j} =
  [(kq)^2-(kq)^4]\delta_{kj} +
  ikq{\cal F}_N[{\cal F}_N^{-1}[a]\otimes{\cal
  F}_N^{-1}[\delta_{kj}]]\,,\quad j = 1,\ldots,n,
\end{equation}
where $\otimes$ indicates componentwise product, and the inverse
Fourier transform is applied separately to each column of
$\delta_{kj}$. Here, $\delta_{kj}$ is not a standard Kronecker
delta, but the $N\times n$ matrix:
\begin{equation}
  \delta_{kj} = \left(
  \begin{array}{ccc}
  0 &  0 &\cdots\\
  1 &  0 &\cdots\\
  0 &  1 &\cdots\\
\multicolumn{3}{c}\dotfill \\
  0 & 0 &\cdots\\
\multicolumn{3}{c}\dotfill \\
  0  & -1 &\cdots\\
  -1 & 0 &\cdots\\
  \end{array}  \right),
\end{equation}
with index $k$ running from 0 to $N-1$.

In practice the number of degrees of freedom $n$ should be
sufficiently large so that no modes important to the dynamics are
truncated, whilst on the other hand, an increase in $n$ corresponds
to an increase in computation. To determine the order of the
truncation in our numerical experiments, we initially chose $n$ to
be large and integrated a random initial seed onto the attractor. By
monitoring the magnitude of the harmonics an integer $k$ was
determined such that $a_j<10^{-5}$ for $j>k$. The value of $n$ was
then chosen to be the smallest integer such that: (i) $n\geq k$, and
(ii) $N = 2n+2$ was an integer power of two. The second condition
ensures that the FFT is applied to vectors of size which is a power
of two resulting in optimal performance.

Note that in the numerical results to follow we work entirely in
Fourier space and use the ETDRK4 time-stepping to solve
(\ref{ch3:eqn:finite}) and (\ref{ch3:eqn:varform}); for further
details concerning the method of exponential-time-differencing see
Appendix \ref{ch:appendix5}. In particular the method uses a fixed
step-size ($h = 0.25$ in our calculations) thus it is necessary to
use an interpolation scheme in order to integrate up to arbitrary
times. In our work we have implemented cubic
interpolation~\cite{NumericalRecipes}. More precisely, to integrate
up to time $t\in[t_i, t_i+h]$, where the $t_i$ are integer multiples
of the step-size $h$. We construct the unique third order polynomial
passing through the two points $a(t_i)$ and $a(t_i+h)$, with
derivatives $a'(t_i)$ and $a'(t_i+h)$ at the respective points. In
this way we obtain the following cubic model:
\begin{align}\label{eqn:cubic}
    p(s) = [&\,2a(t_i)+a'(t_i)+a'(t_i+h)-2a(t_i+h)]s^3 +
           [3a(t_i+h)-3a(t_i)\notag\\&-2a'(t_i)-a'(t_i+h)]s^2 +
           a'(t_i)s + a(t_i),
\end{align}
where the parameter $s = (t-t_i)/h\in[0, 1]$.

\subsection{Numerical results}\label{sec:num}
We now present the results of our numerical experiments. We begin
with a comparison between our method and the nonlinear least squares
solver {\tt lmder} from the MINPACK software package \cite{Minpack}.
Note that {\tt lmder} is an implementation of the
Levenburg-Marquardt algorithm -- see \S\ref{sec:optim}. We have
chosen it because it has recently been applied successfully to
detect many UPOs of the closely related complex Ginzburg-Landau
equation.

In order to determine UPOs via the {\tt lmder} routine we introduce
the following augmented system
\begin{equation}\label{ch3:eqn:augsys}
      F(x,T) =  \left[
      \begin{array}{c}
        \phi^T(x) - x \\
        v(\phi^T(x))\cdot(\phi^T(x) - x)\\
      \end{array}
    \right] = 0,
\end{equation}
where the second equation is motivated in analogous fashion to the
auxiliary equation (\ref{ch3:eqn:Teqn}). To use the {\tt lmder}
solver we must provide a routine that returns the value of $F$
evaluated at a given point $(x, T)$. Also, we provide a routine to
compute the Jacobian analytically rather than use finite
differences, the formula of which differs from
Eq.~(\ref{ch3:eqn:jacaug}) only by multiplication by a constant
matrix. In addition to this the user must supply three tolerances:
ftol, xtol, and gtol. Here, ftol measures the relative error desired
in the sum of squares, xtol the relative error desired in the
approximate solution, and gtol measures the orthogonality desired
between the vector function $F$ and the columns of the Jacobian. In
the computations performed in the next section we set the tolerances
to the recommended values:
$$\mathrm{ftol} = 10^{-8}, \quad \mathrm{xtol} = 10^{-8}\quad \mathrm{and} \quad
\mathrm{gtol} = 0.$$ Finally, we specify a maximum number of
function evaluations allowed during each run of the {\tt lmder}
algorithm in order to increase its efficiency.

For our method we use the set of matrices proposed in
\S\ref{ch3:sec:stabtrans} with
\begin{equation}\label{ch3:eqn:subSD}
    \tilde{C} = \mathcal{C}_{\mathrm{SD}},
\end{equation}
since within the low-dimensional unstable subspace it is possible to
apply the full set of Schmelcher-Diakonos (SD) matrices. The UPOs
determined from our search will then be used as seeds to determine
new cycles. Here we proceed in analogous fashion to Chapter
\ref{ch:stabtrans} by constructing STs from the monodromy matrix,
$D\phi^{T^*}(a^*)$, of the cycle $(a^*, T^*)$. We then solve the
augmented flow (\ref{ch3:eqn:augflow}) from the new initial
condition $(a^*, \tilde{T})$, where the time $\tilde{T}$ is chosen
such that
\begin{equation}\label{ch3:eqn:uporet}
  a^*(0)\approx a^*(\tilde{T}),\quad \mathrm{and} \quad
  \tilde{T}~(\mathrm{mod}T^*) \neq 0.
\end{equation}
Note that any given cycle may exhibit many close returns,
particularly longer cycles, thus in general a periodic orbit may
produce many new initial seeds. This is an especially useful
feature, since we do not have to recompute the corresponding STs.

Initially, to determine that a newly detected cycle, $(x^*, T^*)$,
was different from those already found, we first checked whether the
periods differed, that is, we determined whether $|T^* - T'| >
T_{\mathrm{tol}}$ for all previously detected orbits. If two orbits
where found to have the same period, we then calculated the distance
between the first components of $x^*$, and all other detected orbits
$y^*$. However, in practice we found, that if two orbits have the
same period, then either they are the same or they are related via
symmetry, recall that if $u(x,t)$ is a solution then so is
$-u(L-x,t)$. Thus, in order to avoid the convergence of UPOs that
are trivially related by symmetry we will consider two orbits as
being equal if their periods differ by less than the tolerance
$T_{\mathrm{tol}}$. Of course, this criterion can, in theory, lead
us to discard cycles incorrectly, however, this is highly unlikely
in practice.

\subsubsection{Stabilising transformations vs Levenberg-Marquardt}
The search for UPOs is conducted within a rectangular region
containing the chaotic invariant set. Initial seeds are obtained by
integrating a random point within the region for some transient time
$\tau$. Once on the attractor, the search for close returns within
chaotic dynamics is implemented. That is, we integrate the system
from the initial point on the attractor until $a(t_0) \approx
a(t_1)$ for some times $t_0 < t_1$, and use the close return,
$(a(t_0), T_0)$, where $T_0 = t_1 - t_0$, as our initial guess to a
time-periodic solution. In order to build the STs we solve the
variational equations for each seed starting from the random initial
point, $a(0)$, for time $\tau + t_0$. In order to construct the
matrix $V_p$, we apply the SVD to the matrix
\begin{equation}\label{ch3:eqn:dfsvd}
    D\phi^{\tau+t_0}(a(0))= USW^{\mathsf{T}},
\end{equation}
the corresponding ST is given by
\begin{equation}\label{eqn:Cmatt}
    C = I_n + V_p(\tilde{C} - I_{n_u})V_p^{\mathsf{T}},
\end{equation}
where $V_p = U_{jk}$, $j=1,\dots,n$, $k = 1,\dots,n_u$, i.e. the
first $n_u$ columns of $U$ in (\ref{ch3:eqn:dfsvd}). Here $n_u$ is
the number of expanding directions which is determined by the number
of singular values with modulus greater than one.
\begin{table}
 \caption{The number of distinct periodic solutions for the
 Kuramoto-Sivashinsky equation detected by the method of STs.
 Here $L = 38.5$ and $\alpha = 0.25$.}
    \begin{center}
        \begin{tabular}{|c|c||c|c|c|c|c|}
            \hline
            Period & C & $N_{\mathrm{po}}$ & $N_{\mathrm{hit}}$ & $N_{\mathrm{fev}}$
            & $N_{\mathrm{jev}}$ & Work\\
            \hline
            \multirow{3}{*}{$10 - 100$} & $C_0$ & $28$ & $498$ & $252$ & $10$ & $412$ \\
            &  $C_1$ & $16$ & $296$ & $684$ & $32$ & $1196$
            \\ \cline{3-7}
            &  $\{C_i\}$ & $44$ & $794$ & $936$ & $42$ & $1608$ \\ \hline
            \hline
            \multirow{3}{*}{$100 - 250$} & $C_0$ & $235$ & $395$ & $903$ & $78$ & $2151$ \\
            &  $C_1$ & $64$ & $256$ & $1294$ & $112$ & $3086$
            \\ \cline{3-7}
            &  $\{C_i\}$ & $299$ & $651$ & $2197$ & $190$ & $5237$ \\ \hline
        \end{tabular}
    \end{center}\label{table1}
\end{table}
\begin{table}[t]
 \caption{The number of distinct periodic solutions for the
 Kuramoto-Sivashinsky equation detected by the Levenberg-Marquardt algorithm {\tt
 lmder} with $L = 38.5$.}
 \begin{center}
        \begin{tabular}{|c||c|c|c|c|c|}
            \hline
            Period & $N_{\mathrm{po}}$ & $N_{\mathrm{hit}}$ & $N_{\mathrm{fev}}$ & $N_{\mathrm{jev}}$ &
            Work\\\hline
            $10 - 100$ & $42$ & $497$ & $20$ & $15$ & $260$ \\
            $100 - 250$ & $208$ & $291$ & $500$ & $452$ & $7732$ \\
            \hline
        \end{tabular}
    \end{center}\label{table2}
\end{table}

We examine two different system sizes: $L = 38.5$, for which the
detected UPOs typically have one positive Lyapunov exponent, and $L
= 51.4$, for which the UPOs have either one or two positive Lyapunov
exponents. The corresponding systems sizes are $n = 15$ and $n =31$
respectively. Our experiments where conducted over two separate
ranges. We began by looking for shorter cycles with period $T\in[10,
100]$, the lower bound here was determined {\em a posteriori} so as
to be smaller than the shortest detected cycle. We then searched for
longer cycles, $T\in[100, 250]$ to be more precise, where the
maximum of $T=250$ was chosen in order to reduce the computational
effort.

In our calculations we set the positive constant $\alpha=0.25$ in
Eq.~(\ref{ch3:eqn:augflow}). Using the solver {\tt dlsodar} we
integrated $500$ random seeds over both ranges for time $s = 150$,
the seeds where chosen such that $|\phi^{T_0}(a(0))-a(0)| < 1.0$. If
the flow did not converge within $1000$ integration steps, we found
it more efficient to terminate the solver and to re-start with a
different ST or a new seed. As mentioned in the previous section we
choose a constant $a = 50$ experimentally so that integration is
terminated if the norm of $g$ grows to large, i.e. $|g(x)| =
|\phi^T(x) - x| > a$.

Typically on convergence of the associated flow the UPO is
determined with accuracy of about $10^{-7}$ at which point we
implement two or three iterations of the Newton-Armijo rule to the
following system
\begin{equation}\label{eqn:newtarmijo}
\left(\begin{array}{cc}
        M_T-I_n & v(\phi^T(x)) \\
        v(x) & 0 \\
      \end{array}\right)
\left(\begin{array}{c}
        \delta x \\
        \delta T \\
      \end{array}\right) =
-\left(\begin{array}{c}
        \phi^T(x) - x \\
        0 \\
\end{array}\right),
\end{equation}
in order to allow convergence to a UPO to within roundoff error. The
linear system (\ref{eqn:newtarmijo}) was introduced by Zoldi and
Greenside~\cite{Zoldi98} as a method for determining UPOs for
extended systems in its own right, and in particular, has been
applied successfully in the detection of UPOs for the KSE.
\begin{figure}[t]
\begin{center}
\includegraphics[scale=0.8]{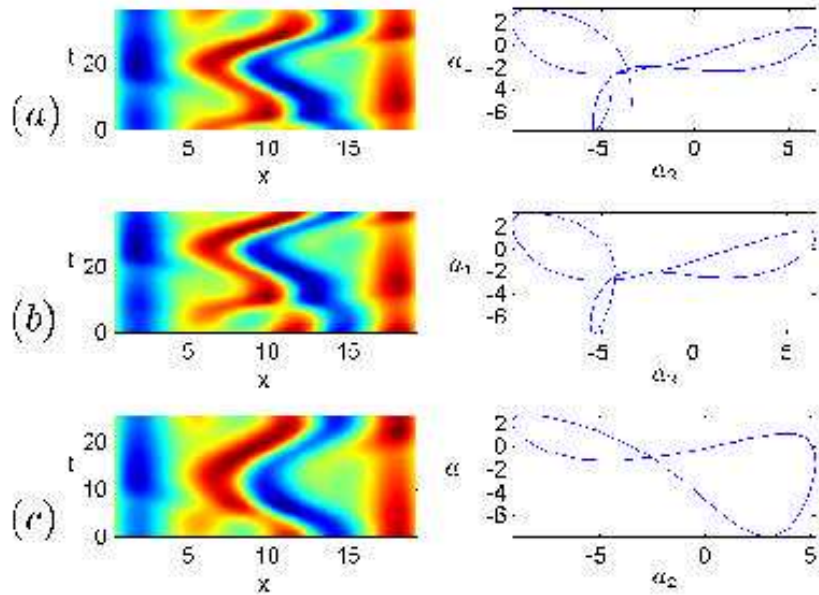}
\caption{Illustration of two UPOs of KSE detected from a single
seed. We show both a level plot for the solutions and a projection
onto the first two Fourier components. Since $u(x,t)$ is
antisymmetric on $[0,L]$, it is sufficient to display the space-time
evolution of $u(x,t)$ on the interval $[0, L/2]$: (a) Seed with time
$T = 37.0$, (b) a periodic solution of length $T = 36.9266$ detected
with stabilising transformation $\tilde{C} = +1$ and (c) a periodic
solution of length $T = 25.8489$ detected with stabilising
transformation $\tilde{C} = -1$.}\label{fig:kse}
\end{center}
\end{figure}
\begin{table}[t]
 \caption{The number of distinct periodic solutions for the
 Kuramoto-Sivashinsky equation detected by the method of stabilising
 transformations. Here $L = 51.4$ and $\alpha = 0.25$.}
    \begin{center}
        \begin{tabular}{|c|c||c|c|c|c|c|}
            \hline
            Period & C & $N_{\mathrm{po}}$ & $N_{\mathrm{hit}}$ & $N_{\mathrm{fev}}$
            & $N_{\mathrm{jev}}$ & Work\\
            \hline
            \multirow{9}{*}{$10 - 100$} & $C_0$ & $11$ & $366$ & $382$ & $16$ & $894$ \\
            &  $C_1$ & $1$ & $221$ & $410$ & $23$ & $1146$ \\
            &  $C_2$ & $7$ & $108$ & $456$ & $20$ & $1096$ \\
            &  $C_3$ & $2$ & $81$  & $357$ & $13$ & $773$ \\
            &  $C_4$ & $2$ & $157$ & $464$ & $26$ & $1296$ \\
            &  $C_5$ & $0$ & $171$ & $654$ & $35$ & $1774$ \\
            &  $C_6$ & $5$ & $174$ & $666$ & $36$ & $1818$ \\
            &  $C_7$ & $2$ & $139$ & $496$ & $36$ & $1648$ \\ \cline{3-7}
            &  $\{C_i\}$ & $30$ & $1417$ & $3885$ & $205$ & $10445$ \\ \hline
            \hline
            \multirow{9}{*}{$100 - 250$} & $C_0$ & $51$ & $330$ & $628$ & $17$ & $1172$ \\
            & $C_1$ & $7$  & $209$ & $807$ & $27$ & $1671$ \\
            & $C_2$ & $17$ & $138$ & $936$ & $31$ & $1928$ \\
            & $C_3$ & $21$ & $177$ & $917$ & $26$ & $1749$ \\
            & $C_4$ & $11$ & $161$ & $877$ & $36$ & $2029$ \\
            & $C_5$ & $0$  & $116$ & $960$ & $43$ & $2336$ \\
            & $C_6$ & $1$  & $117$ & $975$ & $42$ & $2319$ \\
            & $C_7$ & $6$  & $161$ & $879$ & $35$ & $1999$ \\ \cline{3-7}
            & $\{C_i\}$ & $114$ & $1409$ & $6979$ & $257$ & $15203$ \\
            \hline
        \end{tabular}
    \end{center}\label{table3}
\end{table}
\begin{table}[t]
 \caption{The number of distinct periodic solutions for the
 Kuramoto-Sivashinsky equation detected by the Levenberg-Marquardt algorithm {\tt
 lmder} with $L = 51.4$.}
    \begin{center}
        \begin{tabular}{|c||c|c|c|c|c|}
            \hline
            Period & $N_{\mathrm{po}}$ & $N_{\mathrm{hit}}$ & $N_{\mathrm{fev}}$ & $N_{\mathrm{jev}}$ &
            Work\\\hline
            $10 - 100$ & $20$ & $480$ & $64$ & $52$ & $1728$ \\
            $100 - 250$ & $86$ & $356$ & $374$ & $332$ & $10998$ \\
            \hline
        \end{tabular}
    \end{center}\label{table4}
\end{table}

Similarly, we run the {\tt lmder} routine from the same $500$ seeds
over the two different time ranges. The routine terminates if one of
the following three scenarios arise: (i) a predefined maximum number
of function evaluations is exceeded, we set the maximum number of
function evaluations equal to $1000$, (ii) the error between two
consecutive steps is less than xtol, but the sum of squares is
greater than ftol, indicating a local minimum has been detected, or
(iii) both xtol and ftol are satisfied indicating that convergence
to a UPO has been obtained.

In order to make a comparison between the efficiency of the two
methods we introduce the measure of the work done per seed:
\begin{equation}\label{ch3:eqn:cost}
    \mathrm{Work} = N_{\mathrm{fev}} + n\times N_{\mathrm{jev}}.
\end{equation}
Here $N_{\mathrm{fev}}$ is the average number of function
evaluations per seed, $N_{\mathrm{jev}}$ the average number of
Jacobian evaluations per seed and $n$ is the size of the system
being solved. The expression in (\ref{ch3:eqn:cost}) takes into
account the fact that evaluation of the Jacobian is $n$ times more
expensive than evaluation of the function itself.

We present the results of our experiments in Tables \ref{table1} --
\ref{table4}. Here $N_{\mathrm{po}}$ denotes the number of distinct
orbits found, $N_{\mathrm{hit}}$ gives the number of times we
converged to a UPO, and $N_{\mathrm{fev}}$, $N_{\mathrm{jev}}$, and
Work, are as defined above. In Tables \ref{table1} and \ref{table3}
the performance of the stabilising transformations are analysed both
collectively and on an individual basis; here the different $C_i$
denote the different SD matrices. In Table \ref{table1} $C_0 = +1$
and $C_1 = -1$, whilst in Table \ref{table3} we have
\begin{eqnarray*}
   \left\{
    C_{0}=\left[\begin{array}{cc}
    1&0\\0&1\\
    \end{array}\right],\right.
    C_{1}=\left[\begin{array}{cc}
    -1&0\\0&1\\
    \end{array}\right],
     C_{2}=\left[\begin{array}{cc}
    1&0\\0&-1\\
    \end{array}\right],
     C_{3}=\left[\begin{array}{cc}
    -1&0\\0&-1\\
    \end{array}\right],\\
    C_{4}=\left[\begin{array}{cc}
    0&1\\1&0\\
    \end{array}\right],
    C_{5}=\left[\begin{array}{cc}
     0&-1\\1&0\\
    \end{array}\right],
     C_{6}=\left[\begin{array}{cc}
    0&1\\-1&0\\
    \end{array}\right],
    \left.C_{7}=\left[\begin{array}{cc}
    0&-1\\-1&\\
    \end{array}\right]\right\}
\end{eqnarray*}
The total work done per seed is given by the sum over all the SD
matrices and is denoted by $\{C_i\}$.

For $L = 38.5$ we detected a total of $343$ UPOs using the ST method
as compared to $250$ UPOs using the {\tt lmder} algorithm, whilst
for $L = 51.4$ we found $144$ and $106$ UPOs with the respective
methods. Both methods found roughly the same number of orbits when
searching for shorter period cycles. However, as can be seen by
comparing the work done per seed in Tables \ref{table1} --
\ref{table4} {\tt lmder} performs better than the ST method in that
case. The situation changes, however, when we look at the detection
of longer cycles. For $L = 38.5$ in particular, we see that the ST
method computes many more orbits with a considerable reduction in
the cost. In fact, using the identity matrix alone, the ST method
detects more UPOs than {\tt lmder} at approximately a quarter of the
cost. For larger system size, $L = 51.4$, the ST method still
detects more orbits than {\tt lmder}. However, in doing so a
considerable amount of extra work is done. It is important to note
here, that the increase in work is not due to any deficiency in the
ST method. Rather, the rise in the computational cost is brought
about due to the fact that not all the SD matrices work well. For
example, the subset of matrices $\{C_0, C_2, C_3, C_4\} \subset
\mathcal{C}_{ \mathrm{SD}}$, detect approximately $90\%$ of the
longer period UPOs, yet they account for less than $50\%$ of the
overall work. Indeed, for this particular choice of matrices the ST
method still detects more orbits, but more importantly, its
performance now exceeds that of {\tt lmder}.

Note that it is not surprising that the SD matrices do not perform
equally well, this is, after all, what we would have expected based
upon our experience with maps. However, it is still important to try
all SD matrices -- when possible -- in order to compare their
efficiency, especially if one wishes to construct minimal sets of ST
matrices.

For detecting longer period orbits the performance of {\tt lmder}
suffers due to the large increase in the number of Jacobian
evaluations needed; see Tables \ref{table2} and \ref{table4}.
Recall, that {\tt lmder} is an implementation of the
Levenberg-Marquardt algorithm, and that in particular, it computes
the Jacobian once on each iterate. Thus, we may view the number of
Jacobian evaluations as the number of steps required in order to
converge. The increase in the number of iterations necessary to
obtain convergence can be understood by considering the dependance
of $|\phi^T(x)-x|$ on x. For increasing period the level curves of
$|g|^2$ are squeezed along the unstable manifold of $\phi^T(x)$,
resulting in a complicated surface with many minima, both local and
global, embedded within long, winding, narrow ``troughs''. Note that
this is a common problem with all methods that use a cost function
to obtain ``global'' convergence, since these methods only move in
the direction that decreases the cost function.

This can be explained by the following heuristics: for simplicity
let us assume that we are dealing with a map $x_{n+1} = f(x_n)$,
whose unstable manifold is a one-dimensional object. In that case,
we may define a one-dimensional map, locally, about a period-p
orbit, $x^*$, of the map $f$ as
\begin{equation}\label{eqn:localmap}
    \bar{h}(s) = |g(x^*+\delta x)|^2.
\end{equation}
Here $g = f^p(x)-x$ as usual, $\delta x = x(s) - x^*$ is small, and
we only allow $x(s)$ to vary along the unstable manifold. Expanding
$g$ in a Taylor series about the periodic orbit, $x^*$, we obtain
\begin{eqnarray*}\label{eqn:quad}
    \bar{h}(s) &=& |g(x^*) + Dg(x^*)\delta x + O(\delta x^2)|^2\\
               &=& |(Df^p(x^*)-I_n)\delta x + O(\delta x^2)|^2\\
               &=& \delta x^{\mathsf{T}}(\Lambda^p-I_n)^2\delta x + O(\delta x^3),
\end{eqnarray*}
where the third line follows since $\delta x$ is an eigenvector of
$Df(x^*)$ with corresponding eigenvalue $\lambda$, and $\Lambda =
\mathrm{diag}(\lambda,\dots\lambda)$. Note that in the above we
assume that $p$ is large but finite so that the term
$(\Lambda^p-I_n)^2$ remains bounded. Now, close to the periodic
orbit, $\bar{h}$ is approximately a quadratic form with slope of the
order $\lambda^p$, and since $|\lambda|>1$, it follows that if we
move along the unstable manifold $|g|^2$ will grow quicker for
larger periods, or, in other words, the level curves are compressed
along the unstable direction.

Since the {\tt lmder} algorithm reduces the norm of $f$ at each
step, it will typically follow the gradient to the bottom of the
nearest trough, where it will start to move along the narrow base
towards a minimum. Once at the base of a trough, however, {\tt
lmder} is forced into taking very small steps, this follows due to
the nature of the troughs, i.e. the base is extremely narrow and
winding, and since {\tt lmder} chooses its next step-size based on
the straight line search. One of the benefits of the ST method is
that it does not need to decrease the norm at each step and thus,
does not suffer from such considerations.

Another advantage of our method is that we can converge to several
different UPOs from just one seed, depending upon which ST is used.
Figure \ref{fig:kse} shows one of such cases, where
Eq.~(\ref{ch3:eqn:augflow}) is solved from the same seed for each of
the $2^{n_u}$ STs ($n_u = 1$ in this example), with two of them
converging to two different UPOs. Figure \ref{fig:kse}a shows the
level plot of the initial condition and a projection onto the first
two Fourier components. Figures \ref{fig:kse}b and \ref{fig:kse}c
show two unstable spatiotemporally periodic solutions which where
detected from this initial condition, the first of period $T =
36.9266$ was detected using $\tilde{C} = +1$, whilst the second of
period $T = 25.8489$ was detected using $\tilde{C} = -1$. The
ability to detect several orbits from one seed increases the
efficiency of the algorithm.

\subsubsection{Seeding with UPOs}
\begin{table}[t]
 \caption{Number of distinct orbits detected using the method of
 stabilising transformations with periodic orbits as seeds The number
 of seeds is $489$. $L = 38.5$, $\alpha = 0.25$.}
    \begin{center}
        \begin{tabular}{|c||c|c|c|c|c|}
            \hline
            C & $N_{\mathrm{po}}$ & $N_{\mathrm{hit}}$ & $N_{\mathrm{fev}}$ & $N_{\mathrm{jev}}$ &
            Work\\\hline
            $C_0$ & $46$ & $209$ & $2814$ & $136$ & $4854$ \\
            $C_1$ & $52$ & $198$ & $2819$ & $130$ & $4769$ \\\cline{2-6}
            $\{C_i\}$ & $98$ & $407$ & $5633$ & $266$ & $9623$ \\
            \hline
        \end{tabular}
    \end{center}\label{table5}
\end{table}
\begin{table}[t]
 \caption{Number of distinct orbits detected using the method of
 stabilising transformations with periodic orbits as seeds. The number
 of seeds is $123$. $L = 51.4$, $\alpha = 0.25$.}
    \begin{center}
        \begin{tabular}{|c||c|c|c|c|c|}
            \hline
            C & $N_{\mathrm{po}}$ & $N_{\mathrm{hit}}$ & $N_{\mathrm{fev}}$ & $N_{\mathrm{jev}}$ &
            Work\\\hline
            $C_0$ & $1$ & $21$ & $2797$ & $48$ & $4333$ \\
            $C_1$ & $2$ & $37$ & $2815$ & $40$ & $4095$ \\
            $C_2$ & $0$ & $0$ & $216$ & $3$ & $312$ \\
            $C_3$ & $0$ & $1$ & $187$ & $2$ & $251$ \\\cline{2-6}
            $\{C_i\}$ & $3$ & $59$ & $6015$ & $93$ & $8691$ \\
            \hline
        \end{tabular}
    \end{center}\label{table6}
\end{table}
In order to construct a seed from an already detected orbit, $(a^*,
T^*)$, we begin by searching for close returns. That is, starting
from a point on the orbit, $a^*(t_0)$, we search for a time $t_1$
such that $a(t_1)$ is close to $a(t_0)$. As long as
$\tilde{T}\mathrm{mod}T\neq 0$, where $\tilde{T} = t_1 - t_0$, we
take $(a^*(t_0), \tilde{T})$ as our initial guess to a time periodic
solution. For longer period cycles, we can find many close returns
by integrating just once over the period of the orbit. Shorter
cycles, however, produce fewer recurrences and, in general, must be
integrated over longer times to find good initial seeds. In our
experiments we searched for close returns $\tilde{T}\in(10, 2T^*)$
if $T^*<150.0$; otherwise, we chose $\tilde{T}\in(10, T^*)$.

Stabilising transformations are constructed by applying the polar
decomposition to the matrix $\tilde{G} = \tilde{Q}\tilde{B}$, where
$\tilde{G}$ is defined by
\begin{equation*}
   \tilde{G} := V(S\Lambda - \mathrm{I}_n)V^{-1}.
\end{equation*}
Here $V$ and $\Lambda$ are defined through the eigen decomposition
of $D\phi^{T^*}(a^*(t_0)) = V\Lambda V^{-1}$, and $S =
\mathrm{diag}(\pm 1, \pm 1,\ldots,\pm 1)$.  The different
transformations are then given by $C = -\tilde{Q}^{\mathsf{T}}$.
Recall from Chapter \ref{ch:stabtrans} that changing the signs of
the stable eigenvalues is not expected to result in a substantially
different stabilising transformation, thus we use the subset of $S$
such that $S_{ii} = 1$ for $i>n_u$. For $L = 38.5$, we have just two
such transformations, since all periodic orbits have only one
unstable eigenvalue. Whilst for $L=51.4$, we will have either two or
four transformations depending upon the number of unstable
eigenvalues of $D\phi^{T^*}(a^*(t_0))$.

In the calculations to follow, a close return was accepted if
$|a^*(t_1) - a^*(t_0)|<2.5$. Note that, this was the smallest value
to produce a sufficient number of recurrences to initiate the
search. A particular cycle may exhibit many close returns. In the
following, we set the maximum number of seeds per orbit equal to
$5$. To obtain convergence within tolerance $10^{-7}$ we had to
increase the integration time to $s = 250$ and the maximum number of
integration steps allowed to $2000$. If we converged to a UPO then
as in the preceding section, we apply two or three Newton-Armijo
steps to the linear system (\ref{eqn:newtarmijo}) in order to
converge to machine precision.

The results of our experiments are given in Tables \ref{table5} and
\ref{table6}. As in the previous section, $N_{\mathrm{po}}$ denotes
the number of distinct orbits found, $N_{\mathrm{hit}}$ the number
of times we converged to a UPO, $N_{\mathrm{fev}}$ the average
number of function evaluations per seed, and $N_{\mathrm{jac}}$ the
average number of jacobian evaluations per seed. The computational
cost per seed is measured in terms of the average number of function
evaluations per seed and is defined as in Eq.~(\ref{ch3:eqn:cost}).
In Tables \ref{table5} and \ref{table6} the different $C_i$ can be
uniquely identified by the signature of pluses and minuses defined
through the corresponding matrix $S$. For $L = 38.5$, the matrices
$C_0$ and $C_1$ correspond to the signatures $(+,\ldots,+)$ and
$(-,+,\ldots,+)$, respectively, whilst for $L = 51.4$, the matrices
$C_0$, $C_1$, $C_2$ and $C_3$ correspond to $(+,\ldots,+)$,
$(-,+,\ldots,+)$, $(-,-,+,\ldots,+)$, and $(+,-,\ldots,+)$
respectively. As before, the total work done per seed is defined as
the sum over all matrices and is denoted by $\{C_i\}$.

\begin{figure}[t]
\begin{center}
\includegraphics[scale=0.8]{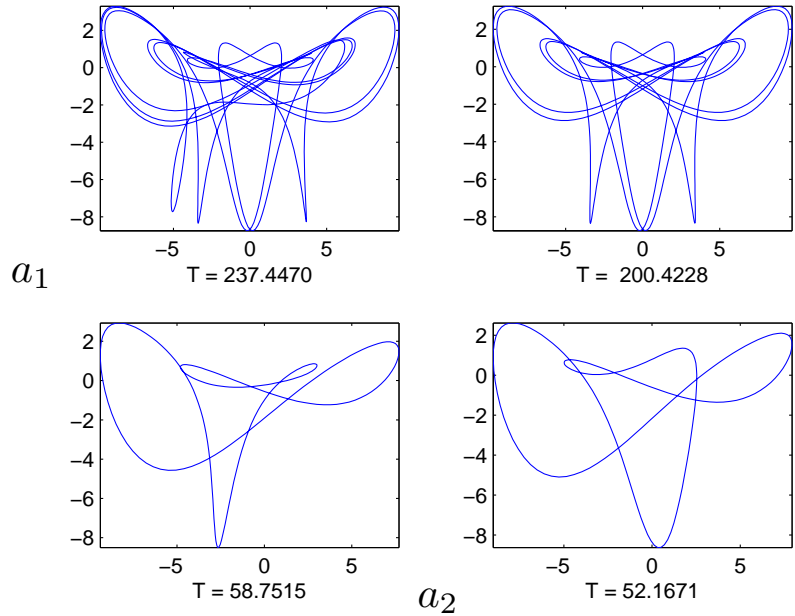}
\caption{Illustration of how a UPO can be used as a seed to detect
 new cycles. We show the projection onto the first two Fourier
 components of the initial seed ($T = 237.4470$) and the three detected
 orbits. Here $L = 38.5$ and $n = 15$.}\label{fig:seed}
\end{center}
\end{figure}

For $L = 38.5$ we where able to construct $489$ seeds from $343$
previously detected periodic orbits. From the new seeds we found a
further $98$ distinct UPOs, bringing the total number of distinct
orbits detected for $L = 38.5$ to $433$. An important observation to
make is that both matrices perform equally well, as can be seen from
Table \ref{table5}. This is in contrast to the SD matrices for which
the performance of the different matrices varies greatly; see Tables
\ref{table1} and \ref{table3}. The result of using periodic orbits
as seeds has, however, increased the cost per seed as compared with
the cost of seeding with close returns within a chaotic orbit, see
Tables \ref{table1} and \ref{table5}. The increased computation is
due mainly to the fact that the close returns obtained from periodic
orbits are not as good as those obtained from a chaotic orbit.
Recall that the stabilising transformations are based on the local
invariant directions of the orbit, and we would expect the
performance to suffer as we move further from the seed.

For larger system size, where the system becomes more chaotic, the
construction of seeds from close returns becomes increasingly
difficult. For $L = 51.4$, we constructed $123$ initial seeds from
the $144$ periodic orbits using the method of near recurrences. Here
we detected only $3$ new distinct UPOs. It is important to note,
however, that although we do not find many new UPOs for the case $L
= 51.4$, the method is still converging for approximately $20\%$ of
all initial seeds. Also, in Table \ref{table6}, the poor performance
of the matrices $C_2$, $C_3$ as compared to that of $C_0$, $C_1$, is
due to the fact that only a small percentage of seeds were
constructed from orbits with two positive Lyapunov exponents.

One advantage of using periodic orbits as seeds is that we can
construct many seeds from a single orbit. Figure \ref{fig:seed}
gives an illustration of three periodic orbits which were detected
from a long periodic orbit used as seed. Figure \ref{fig:seed} shows
the projection onto the first two Fourier components of the initial
seed (of period $T = 237.4470$) and the three orbits detected whose
periods in descending order are $T = 200.428$, $T = 58.7515$ and $T
= 52.1671$. For each seed determined from a particular periodic
orbit the stability transformations are the same. Hence, the ability
to construct several seeds from a single orbit, increases the
efficiency of the scheme.

\section{Summary}\label{sec:summary}
In this chapter we have presented a scheme for detecting unstable
periodic orbits (UPOs) in high-dimensional chaotic
systems~\cite{Crofts05,Crofts07} based upon the stabilising transformations
(STs) proposed in~\cite{Davidchack99c,Schmelcher97}. Due to the fact that
high-dimensional systems studied in dynamical systems typically
consist of low-dimensional dynamics embedded within a
high-dimensional phase space, it is possible to increase the
efficiency of the STs approach by restricting the construction of
such transformations only to the low-dimensional unstable subspace.
Following the approach often adopted in subspace iteration
methods~\cite{Lust98}, we construct a decomposition of the tangent
space into unstable and stable orthogonal subspaces.  We find that
the use of singular value decomposition to approximate the
appropriate subspaces is preferable to that of Schur decomposition,
which is usually employed within the subspace iteration approach. As
illustrated with the example of the Ikeda map, the decomposition
based on singular value decomposition is less susceptible to
variations in the properties of the tangent space away from a seed
and thus produce larger basins of attraction for stabilised periodic
orbits. Within the low-dimensional unstable subspace, the number of
useful STs is relatively small, so it is possible to apply the full
set of Schmelcher-Diakonos matrices. The detected orbits were then
used as seeds in order to search for new UPOs, thus enabling us to
apply the STs introduced in~\cite{Crofts06}.

\chapter{Summary and outlook}
\label{ch:summary}
\begin{quote}
In mathematics you don't understand things. You just get used to
them.\\ \emph{J. von~Neumann}
\end{quote}

\section{Summary}
In this thesis we have discussed in detail the application to
high-dimensional systems of the method of stabilising
transformations (ST). The main advantage of the ST approach as
compared to other methods is its excellent convergence properties.
In particular, the basins of attraction are much larger than the
basins produced by other iterative schemes, and are simply connected
regions in phase space. The convergence properties for methods of
Newton or secant type, can be improved upon by using a line search,
the step-size being determined by a suitable cost function. However,
such methods have no way of differentiating between true roots and
local minima of the cost function; a problem which increases
significantly with system dimension due to the topology of
multi-dimensional flows.

The application to high-dimensional systems however, is not
straightforward. The set of transformations conjectured by
Schmelcher and Diakonos (SD) have two major failings: (i) the
complete set of transformations has cardinality $2^nn!$ (here $n$ is
the system dimension), and (ii) this set contains a certain
redundancy, in that, not all matrices are useful. The main focus has
been on trying to understand the theoretical foundations of the
stabilising transformations in order to determine a new minimal set
of such transformations enabling efficient detection for systems
with $n\geq 4$.

Our approach is based on the understanding of the relationship
between the stabilising transformations and the properties of the
eigenvalues and eigenvectors of the stability matrices of the
periodic orbits. Of particular significance is the observation that
only the unstable eigenvalues are important for determining the
stabilisation matrices. Therefore, unlike the SD matrices, whose
numbers grow with system dimensionality as $2^nn!$, our set has at
most $2^k$ elements, where $k$ is the dimension of the unstable
manifold. The dependence of the number of transformations on the
dimensionality of the unstable manifold is especially important in
cases when we study low-dimensional chaotic dynamics embedded within
a high-dimensional phase space.

The ST method was originally developed for the detection of periodic
orbits in time-discrete maps. However, the periodic orbits can be
used similarly to infer the properties of time-continuous dynamical
systems. Here, the ST approach is typically applied directly to the
corresponding Poincar\'{e} map, where the Poincar\'{e} surface of
section (PSS) can be a stroboscopic map or defined in terms of phase
space intersections. For high-dimensional systems, however, the
correct choice of a PSS is a challenging problem in itself. Due to
the complex topology of a high-dimensional phase space, the
successful detection of unstable periodic orbits (UPOs) will depend
upon the choice of PSS. A major advantage of our method is that it
does not require a PSS in order to apply the ST method to determine
UPOs for flows.

One apparent drawback of the new scheme is that a small set of UPOs
needs to be available for the construction of the stabilising
transformation at the start of the detection process. However, we
have shown that it is possible to construct the stabilising
transformations without the knowledge of UPOs. Recalling that the
stabilising transformations depend mostly on the properties of the
unstable subspace, we use the fact that the decomposition into
stable and unstable subspaces can be defined at any, not just
periodic, point on the chaotic set. The decomposition is done in a
process similar to that used in the subspace iteration algorithm
\cite{Lust98}, the full set of $\mathcal{C}_{\mathrm{SD}}$ matrices
can then be applied in the low-dimensional unstable subspace.

The new transformations were tested on a kicked double rotor map,
three symmetrically coupled H\'{e}non maps and the
Kuramoto-Sivashinsky equation (KSE). For the time-discrete maps our
aim was to achieve the detection of plausibly complete sets of
periodic orbits of low periods up to as high a period as was
computationally feasible. In both cases our algorithm was able to
detect large numbers of UPOs with high degree of certainty that the
sets of UPOs for each period were complete. We have used the
symmetry of the systems in order to test the completeness of the
detected sets.  On the other hand, when the aim is to detect as many
UPOs as possible without verifying the completeness, the symmetry of
the system could be used to increase the efficiency of the detection
of UPOs: once an orbit is detected, additional orbits can be located
by applying the symmetry transformations.

In the case of the KSE our goal was somewhat different~\cite{Crofts05,Crofts07}.
Here the aim was to construct and implement stabilising transformations
determined from an arbitrary point in phase space. The proposed
method for detecting UPOs has been tested on both a $15$ and
$31$-dimensional system of ODEs representing odd solutions of the
KSE. Using the new set of stabilising transformations we have been
able to detect many periodic solutions of the KSE in an efficient
manner. The newly detected orbits were then used as seeds to detect
new orbits by extending the ideas in~\cite{Crofts06} to the
continuous case.

\section{Outlook} The method of stability transformations is a
powerful tool in the detection of unstable periodic orbits for a
large class of dynamical systems. Until now, its application has
been restricted to low-dimensional dynamical systems. In this thesis
we have used the insights gained from our two-dimensional analysis
in order to extend the ST method efficiently in higher dimensions.
However, the understanding of the theoretical foundations in more
than two dimensions remains a challenge. For example, the SD
Conjecture \ref{conj:sd} still remains an open problem for $n>2$.
One possible approach to this question would be to show that for any
orthogonal matrix $Q$ there exists at least one matrix, $C\in
\mathcal{C}_{ \mathrm{SD}}$, which is close to $Q$ for all $n$. Here
the measure of closeness is that defined in~\cite{Crofts06}.

The detection of spatial and temporal patterns in the time evolution
of nonlinear PDEs is an area of research currently receiving a lot
of attention; see~\cite{Kida01,Kerswell04,Viswanath06} and
references therein. Here, an important question is whether or not
the results of the periodic orbit theory generalise, and more
importantly, remain useful for such high-dimensional dynamical
systems. The numerical treatment of such systems, however, is
limited due to computational constraints. Thus, the construction of
efficient tools for detecting spatiotemporal patterns is extremely
important. Our numerical results for the KSE are a first step in
this direction. In the future, an exploration of the full solution
space should be conducted; the restriction to odd subspace is a
computational convenience, rather than a physically meaningful
constraint. Importantly, for the full system, solutions are
invariant under translations along the $x$-axis, i.e. if $u(x,t)$ is
a solution then so is $u(x+\Delta,t)$, $\forall \Delta \in
\mathbb{R}$. Now UPOs only tell part of the story and it is expected
that relative periodic solutions, i.e. $u(x+\Delta,t+T) = u(x,T)$,
will become important for a proper description of the dynamics. One
way to eliminate this arbitrariness is to supply an extra equation,
this may be done, for example, in analogy to the auxiliary equation
(\ref{ch3:eqn:Teqn}) of Chapter \ref{ch:kse}:
\begin{equation}\label{eqn:deltaaux}
    \frac{d\Delta}{ds} \varpropto -\frac{\partial |g|^2}{\partial\Delta},
\end{equation}
where $g(x,\Delta,T) = \phi^T(x+\Delta) - x$.

Also, it is clear that as we start to study more and more complex
systems, that it becomes increasingly difficult to find UPOs from
which to initialise the detection process. Thus, the idea of
constructing matrices from an arbitrary point in phase space becomes
increasingly important, and although numerical results indicate that
the approximation of unstable subspace via singular value
decomposition is preferable to that of Schur decomposition, the
theory at present holds only for Schur decomposition. Future work
should concentrate on the mathematical analysis needed to formulate
similar ideas for singular value decomposition; a clearer
mathematical understanding should enable the development of an
efficient strategy for systematic detection of periodic orbits in
high-dimensional systems.

We construct STs by approximating the local stretching rates of the
system at the initial seed $x$. As we evolve the associated flow
however, we can wander far away from the initial seed, in that case
the ST matrix is no longer valid due to its local nature. One might
try to improve the applicability of the method by constructing
transformations adaptively, which amounts to continually updating
our approximation of the unstable subspace. However, since we
approximate the unstable subspace at $x$ by evaluating the singular
value decomposition at its preimage, i.e. $\phi^{-t}(x)$, for some
time $t$, we must be able to integrate backwards in time.
Unfortunately, for dissipative systems such as the KSE, evolving the
solution for negative times is not possible. However, this is not the 
case for reversible systems, for example Hamiltonian systems, in that 
case we can integrate either forward or backwards in time in a
straightforward manner.

Another possible avenue of exploration would be the application of
ST approach to Hamiltonian systems, which are most relevant for the
study of quantum dynamics of classically chaotic systems and
celestial mechanics. This should allow for further improvements --
particularly theoretically -- to be achieved through the
consideration of the symmetry properties of the Hamiltonian systems.
Indeed, since the choice of the transformations is directly related
to the eigenvalues of the orbit monodromy matrix, one of the obvious
properties that can be exploited in this case is the time-reversible
nature of the Hamiltonian systems, which manifests itself in the
symmetry of the eigenvalues of the corresponding monodromy matrix.
Other symmetry properties related to the symplectic structure could
also be explored.

An important consideration when applying our method to Hamiltonian
problems is that for a given energy the search for UPOs take place
on the corresponding energy surface. However, in the determination
of UPOs we solve an associated flow which, in general, will cause us
to move away from the energy surface. One possible approach to avoid
this, would be to project out any part of the vector field
transverse to the surface and to solve the resulting system of ODEs,
thus remaining on the energy surface. Alternatively, we could add an
extra equation, which would act so as to force the associated flow
onto the energy surface in an asymptotic manner. This could work,
for example, in a similar fashion to penalty merit functions in
nonlinear programming problems, where one typically has to make a
tradeoff between minimising the function of interest and remaining
on the manifold of feasible points.

\appendix

\achapter{Frequently used notation}\label{ch:appendix1} \addnotation
p:{Period of a discrete system}{p}
\\
\addnotation f:{Discrete map}{f}
\\
\addnotation g:{$f^p(x)-x$}{g}
\\
\addnotation n:{System dimension}{n}
\\
\addnotation U:{A discrete dynamical system}{U}
\\
\addnotation \Sigma:{Associated flow}{sigma}
\\
\addnotation C:{Stabilising transformation}{C}
\\
\addnotation \mathcal{C}_{\mathrm{SD}}:{Stabilising transformations
proposed by Schmelcher and Diakonos}{CSD}
\\
\addnotation G:{The Jacobian of $g(x)$}{G}
\\
\addnotation \mathrm{I}_n:{$n\times n$ identity matrix}{I}
\\
\addnotation Df^p:{The jacobian of $f^p(x)$}{Df}
\\
\addnotation Dg:{Alternative notation for the Jacobian of
$g(x)$}{Dg}
\\
\addnotation \mathsf{T}:{Matrix transpose}{T}
\\
\addnotation \phi^t:{The flow map of a differential equation}{phi}
\\
\addnotation T: {Period of a continuous system}{TT}
\\
\addnotation J:{The Jacobian for the flow $J\equiv D\phi^t(x)$}{J}

\achapter{Abbreviations}\label{ch:appendix2} \addabbreviation
\mathrm{UPO}:{Unstable periodic orbit}{UPO}
\\
\addabbreviation \mathrm{SD}:{Schmelcher-Diakonos}{SD}
\\
\addabbreviation \mathrm{ODE}:{Ordinary differential equation}{ODE}
\\
\addabbreviation \mathrm{PDE}:{Partial differential equation}{PDE}
\\
\addabbreviation \mathrm{ST}:{Stabilising transformation}{ST}
\\
\addabbreviation \mathrm{PSS}:{Poincar\'{e} surface of section}{PSS}
\\
\addabbreviation \mathrm{BW}:{Biham-Wenzel}{BW}
\\
\addabbreviation \mathrm{CB}:{Characteristic bisection}{CB}
\\
\addabbreviation \mathrm{NR}:{Newton-Raphson}{NR}
\\
\addabbreviation \mathrm{GN}:{Gauss-Newton}{GN}
\\
\addabbreviation \mathrm{LM}:{Levenberg-Marquardt}{LM}
\\
\addabbreviation \mathrm{CHM}:{Coupled H\'{e}non map}{CHM}
\\
\addabbreviation \mathrm{SchD}:{Schur decomposition}{SchD}
\\
\addabbreviation \mathrm{SVD}:{Singular value decomposition}{SVD}
\\
\addabbreviation \mathrm{KSE}:{Kuramoto-Sivashinsky equation}{KSE}
\\
\addabbreviation \mathrm{FFT}:{Fast Fourier transform}{FFT}

\achapter{Derivation of the kicked double rotor map}
\label{ch:appendix3} The double rotor map is a four-dimensional
discrete system, physically it describes the dynamics of a double
rotor under the influence of a periodic kick; see Figure
\ref{fig:drmp}. Note that our discussion closely follows that given
in~\cite{Grebogi86}, we start by giving a description of the
mechanical device that is the double rotor, before deriving the
equations of motion and the corresponding map.

The double rotor is made up of two thin, massless rods, connected as
in figure \ref{fig:drmp}. The first rod, of length $l_1$, pivots
about the fixed point $P_1$, whilst the second rod, of length $l_2$,
pivots about the moving point $P_2$. The position at time $t$ of the
two rods is given by the angular variables $\theta_1(t)$ and
$\theta_2(t)$ respectively. A mass $m_1$ is attached to the first
rod at $P_2$, and masses $m_2/2$ are attached at either end of the
second rod. Friction at $P_1$ is assumed to slow the first rod at a
rate proportional to $\dot{\theta}_1$, and friction at $P_2$ slows
the second rod at a rate proportional to $\dot{\theta}_2 -
\dot{\theta}_1$. The end of the second rod, marked $K$, receives a
periodic kick at times $t = T, 2T, 3T\dots$, the force of which is
constant.

We may construct the equations of motion for the system by noting
that the kick represents a potential energy of the system which is
given by
\begin{equation}\label{eqn:pot}
    V(t) = (l_1\cos{\theta_1} + l_2\cos{\theta_2})f(t),
\end{equation}
here $f(t)$ consists of periodic delta function kicks, i.e. $f(t) =
\sum_kf_0\delta(t-k\mathsf{T})$ . The kinetic energy of the system
is easily seen to be
\begin{equation}\label{eqn:kin}
    K(t) = \frac{1}{2}(\mathrm{I}_1\dot{\theta}_1^2 + \mathrm{I}_2\dot{\theta}_2^2)
\end{equation}
where $\mathrm{I}_1 = (m_1 + m_2)l_1^2$, $\mathrm{I}_2 = m_2l_2^2$
and $\dot{\theta}_i = \mathrm{d}\theta_i/\mathrm{d}t$ for $i = 1,2$.
Using the Lagrangian formulation, $L = K - V$, with
\begin{equation}\label{eqn:lang}
    \frac{d}{dt}\left(\frac{\partial L}{\partial \dot{q}_i}\right) - \frac{\partial L}{\partial
    q_i} = - \frac{\partial F}{\partial\dot{\theta}_i},
\end{equation}
where Rayleighs dissipation function is given by 
$F = \frac{1}{2}\nu_1I_1\dot{\theta_1}^2 +\frac{1}{2}\nu_2I_2(\dot{\theta_2}-\dot{\theta_1})^2$, 
yields the following system of ODEs
\begin{eqnarray}  \label{eqn:ode1}
  \mathrm{I}_1\frac{\mathrm{d}\dot{\theta}_1}{dt} &=& \mathrm{I}_1f(t)\sin{\theta_1} -\nu_1\mathrm{I}_1\dot{\theta}_1
  + \nu_2\mathrm{I}_2(\dot{\theta}_2 - \dot{\theta}_1)\,,
  \label{eqn:ode1} \\
  \mathrm{I}_2\frac{\mathrm{d}\dot{\theta}_2}{dt} &=&
  \mathrm{I}_2f(t)\sin{\theta_2} -\nu_2\mathrm{I}_2(\dot{\theta}_2 -
  \dot{\theta}_1)\,,
  \label{eqn:ode2}
\end{eqnarray}
here $\nu_1, \nu_2$ are the coefficients of friction.

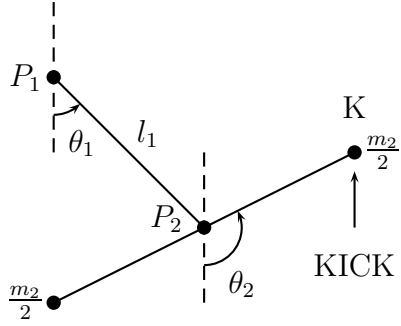
\begin{figure}[t]
 \begin{center}
  \psset{xunit=1cm,yunit=1cm}
  \begin{pspicture}(0,3)(8,8)
  \psline[linestyle=dashed](2,8)(2,6)\psline[linestyle=dashed](4,6)(4,4)
  \psdot[dotscale=1.5](2,7)
  \psline(2,7)(4,5)
  \psdot[dotscale=1.5](4,5)
  \psline(2,4)(6,6)\psdot[dotscale=1.5](2,4)\psdot[dotscale=1.5](6,6)
  \psline{->}(6,5)(6,5.75)
  \psarc{->}(2,7){0.5}{270}{315} \psarc{->}(4,5){0.5}{270}{26.5651}

  \rput(2.375,6.125){$\theta_1$} \rput(4.5,4.25){$\theta_2$}
  \rput(1.625,7){$P_1$} \rput(3.5,5.125){$P_2$}
  \rput(3.25,6.25){$l_1$} 
  \rput(1.625,4){$\frac{m_2}{2}$} \rput(6.375,6){$\frac{m_2}{2}$}
  \rput(6,6.6){K} \rput(6,4.5){KICK}
  \end{pspicture}
 \end{center}
 \caption{The double rotor under the influence of a periodic kick.}\label{fig:drmp}
\end{figure}

Since the kicks are treated as being instantaneous, i.e. $f(t) = 0$,
except for times, $t$, which are integer multiples of the period,
the equations of motion (\ref{eqn:ode1}) -- (\ref{eqn:ode2}) reduce
to the following linear system for $t\in((n-1)T, nT)$,
$n\in\mathbb{Z}^+$,
\begin{equation}\label{eqn:lindrmp}
    \left(\begin{array}{c}
        \mathrm{d}\dot{\theta}_1/\mathrm{d}t \\
        \mathrm{d}\dot{\theta}_2/\mathrm{d}t \\
    \end{array}\right) =
    \left(\begin{array}{cc}
        -(\nu_1 + \nu_2) & \nu_2 \\
        \nu_2 & -\nu_2 \\
    \end{array}\right)
    \left(\begin{array}{c}
        \dot{\theta}_1 \\
        \dot{\theta}_2 \\
      \end{array}\right).
\end{equation}
Given the initial conditions $(\dot{\theta}_1(0),
\dot{\theta}_2(0))^\mathsf{T}$, Eq.~(\ref{eqn:lindrmp}) admits the
following solution
\begin{equation}\label{eqn:sol}
    \left(\begin{array}{c}
        \dot{\theta}_1(t) \\
        \dot{\theta}_2(t) \\
    \end{array}\right) = L(t)
    \left(\begin{array}{c}
        \dot{\theta}_1(0) \\
        \dot{\theta}_2(0)\\
    \end{array}\right),
\end{equation}
where
\begin{equation}\label{eqn:L}
    L(t) = \left(\begin{array}{cc}
      u_{11}^2e^{-s_1t} + u_{21}^2e^{-s_2t} & u_{11}u_{12}e^{-s_1t} + u_{21}u_{22}e^{-s_2t} \\
      u_{11}u_{12}e^{-s_1t} + u_{21}u_{22}e^{-s_2t} & u_{12}^2e^{-s_1t} + u_{22}^2e^{-s_2t} \\
    \end{array}\right),
\end{equation}
$s_1$, $s_2$ are eigenvalues of the matrix in
Eq.~(\ref{eqn:lindrmp}) with corresponding orthonormal eigenvectors
$\textbf{u}_1 = (u_{11}, u_{12})^{\mathsf{T}}$, $\textbf{u}_2 =
(u_{21}, u_{22})^{\mathsf{T}}$. For the initial condition
$(\theta_1(0), \theta_2(0))$, we may integrate Eq.~(\ref{eqn:sol})
to determine the position of the rods
\begin{equation}\label{eqn:int}
    \left(\begin{array}{c}
        \theta_1(t) \\
        \theta_2(t) \\
    \end{array}\right) = M(t)
    \left(\begin{array}{c}
        \dot{\theta}_1(0) \\
        \dot{\theta}_2(0)\\
    \end{array}\right) +
    \left(\begin{array}{c}
        \theta_1(0) \\
        \theta_2(0) \\
    \end{array}\right),
\end{equation}
here $M(t) = \int_0^t L(\xi)\mathrm{d}\xi$. Note that
Eqs.~(\ref{eqn:L}) -- (\ref{eqn:int}) completely describe the
dynamics of the double rotor for $t\in((n-1)T, nT)$,
$n\in\mathbb{Z}^+$.

At $t=T$ the angular velocity of each rod changes instantaneously,
thus the angular velocity $\dot{\theta}_i$ of each rod is
discontinuous. Denote the limits from the left and right as
$\dot{\theta_i}(T^-)$, $\dot{\theta_i}(T^+)$ respectively, then the
discontinuity is given by
\begin{equation}\label{eqn:disc}
    \dot{\theta_i}(T^+) - \dot{\theta_i}(T^-) = \frac{l_if_0}{\mathrm{I}_i}\sin{\theta_i}(T),\quad i=1, 2.
\end{equation}
The position however, will vary continuously for all time, i.e.
$\theta_i(T^+) = \theta_i(T^-)$, $i = 1, 2$. In this way the
solution of (\ref{eqn:ode1}) -- (\ref{eqn:ode2}) is a composition of
the solution of the linear system (\ref{eqn:lindrmp}) and the
periodic kicks at $t = T, 2T,\dots$. Thus to understand the dynamics
of the double rotor it suffices to study the following map obtained
from Eqs.~(\ref{eqn:sol}) -- (\ref{eqn:disc}),
\begin{equation}\label{eqn:drmp}
 \left(\begin{array}{c}
        \theta_1^{(n+1)} \\
        \theta_2^{(n+1)} \\
        \dot{\theta}_1^{(n+1)} \\
        \dot{\theta}_2^{(n+1)} \\
 \end{array}\right) =
 \left(\begin{array}{cc}
     M(t) & \textbf{0} \\
     \textbf{0} & L(t) \\
 \end{array}\right)
 \left(\begin{array}{c}
     \dot{\theta}_1^{(n)} \\
     \dot{\theta}_2^{(n)} \\
     \dot{\theta}_1^{(n)} \\
     \dot{\theta}_2^{(n)} \\
 \end{array}\right) +
 \left(\begin{array}{c}
    \theta_1^{(n)} \\
    \theta_2^{(n)} \\
    \frac{l_1f_0}{\mathrm{I}_1}\sin{\theta_1^{(n+1)}}\\
    \frac{l_2f_0}{\mathrm{I}_2}\sin{\theta_2^{(n+1)}}\\
  \end{array}\right)
\end{equation}
which give the position and the angular velocity of the rods
immediately after each kick.

\achapter{A modicum of linear algebra} \label{ch:appendix4} Below we
define some of the theory from linear algebra used throughout the
current piece of work. Proofs are omitted for brevity but can be
found in the books on matrix analysis by Horn and Johnson from which
the results have been taken~\cite{HornBook1,HornBook2}. In what
follows we denote by $M_n$ the set of $n\times n$, possibly complex
valued matrices. Further we denote by $A^*$ the conjugate transpose
of $A$, this of course equals the usual transpose for real matrices.

We begin with two theorems concerning the matrix decompositions
which have been used so frequently in our discussions
\begin{theorem}
If $G\in M_n$, then it may be written in the form
\begin{equation}
G = QB\label{eqn:polar}
\end{equation}
where $B$ is positive semi-definite and $Q$ is unitary. The matrix
$B$ is always uniquely defined as $B\equiv (GG^*)^{1/2}$; if $G$ is
nonsingular, then $Q$ is uniquely determined as $Q\equiv B^{-1}G$.
If $G$ is real, then $B$ and $Q$ may be taken to be real.
\end{theorem}\noindent
The factorisation \ref{eqn:polar} is known as the \emph{polar
decomposition}; an important point to note here is the uniqueness of
the two factors in the case that $G$ is nonsingular. A perhaps even
more useful decomposition is that of \emph{singular value
decomposition} (SVD). Its applications are far reaching and include
for example, the study of linear inverse problems, uses in signal
processing and many areas of statistics.
\begin{theorem} If $G\in M_{p,n}$ has rank $k$, then it may be
written in the form
\begin{equation}
G = VSW^*
\end{equation}
where $V$ and $W$ are unitary. The matrix $S = [s_{ij}]\in M_{p,n}$
has $s_{ij} = 0$ for all $i\neq j$, and $s_{11}\geq
s_{22}\geq\cdots\geq s_{kk}> s_{k+1,k+1} = \cdots = s_{qq} = 0$,
where $q = \min(p,n)$. The numbers $\{s_{ii}\} = \{s_i\}$ are the
nonnegative square roots of the eigenvalues of $GG^*$, and hence are
uniquely determined. The columns of $V$ are eigenvectors of $GG^*$
and the columns of $W$ are eigenvectors of $G^*G$. If $p\leq n$ and
if $GG^*$ has distinct eigenvalues, then $V$ is determined up to a
right diagonal factor $D =
\mathrm{diag}(e^{i\theta_1},\dots,e^{i\theta_n})$ with all
$\theta_i\in\mathbb{R}$; that is, if $G = V_1SW_1^* = V_2SW_2^*$,
then $V_2 = V_1D$. If $p< n$, then $W$ is never uniquely determined;
if $p=n=k$ and $V$ is given, then $W$ is uniquely determined. If
$n\leq p$, the uniqueness of $V$ and $W$ is determined by
considering $G^*$. if $G$ is real, then $V$, $S$ and $W$ may all be
taken to be real.
\end{theorem}\noindent
The two decompositions admit the following relation
\begin{eqnarray}\label{eqn:pdsvd}
  G = QB &=& VSW^* \\
         &=& VW^*WSW^*
\end{eqnarray}
thus we have that $Q = VW^*$ and $B = WSW^*$. In constructing the
stabilising transformations of Chapter \ref{ch:stabtrans}, we use
the SVD of the matrix $G$, rather than its polar decomposition, in
our numerical calculations.

Below we will give the details of the theorems of Lyapunov and
Sylvester which where used in the proof of
Corollary~\ref{corr:lyapunov}. However, before doing so we find it
instructive to present a couple of useful definitions of what we
believe to be nonstandard linear algebra. The first notion is that
of the \emph{inertia} of a general matrix in $M_n$
\begin{definition}
If $A\in M_n$, define:
\begin{description}
 \item[$i_+(A) \equiv$] the number of eigenvalues of $A$, counting
 multiplicities, with positive real part;
 \item[$i_-(A) \equiv$] the number of eigenvalues of $A$, counting
 multiplicities, with negative real part; and
 \item[$i_0(A) \equiv$] the number of eigenvalues of $A$, counting
 multiplicities, with zero real part.
\end{description}
Then, $i_+(A) + i_-(A) + i_0(A) = n$ and the row vector
$$i(A) \equiv [i_+(A), i_-(A), i)(A)]$$
is called the inertia of A.
\end{definition}\noindent
This leads us nicely onto our second definition, which simply
translates our idea of stability in terms of a dynamical system,
into the language of linear algebra
\begin{definition}
A matrix $A\in M_n$ is said to be {\em positive stable} if $i(A) =
[n, 0, 0]$, that is, $i_+(G) = n$.
\end{definition}\noindent
The matrices which we study in this work are typically derived from
some dynamical system of interest. In this case, stability is
defined in terms of matrices who possess only eigenvalues with
negative real parts, however, it is convention in linear algebra to
discuss positive matrices, and we shall stick to convention. It
should be clear that to determine that a matrix $A$ is stable, it
suffices to show that the matrix $-A$ is positive stable. Our last
definition of this section concerns the idea of \emph{congruence} of
two matrices
\begin{definition}
Let $A, B\in M_n$ be given. If there exists a nonsingular matrix $S$
such that
\begin{equation}
A = SBS^*
\end{equation}
then $A$ is said to be $^*$\emph{congruent} (``star-congruent'') to
B.
\end{definition}
The defining property for the equivalence class of $^*$congruent
matrices introduced above is that their inertia remains constant.
That this is so is given by the following theorem which is known as
{\em Sylvester's law of inertia}
\begin{theorem}
Let $A, B \in M_n$ be Hermitian matrices. There is a nonsingular
matrix $S\in M_n$ such that $A = SBS^*$ if and only if $A$ and $B$
have the same inertia, that is, the same number of positive,
negative, and zero eigenvalues.
\end{theorem}\noindent
We finish the section with the following result due to Lyapunov
\begin{theorem}
Let $A\in M_n$ be given. Then $A$ is positive stable if and only if
there exists a positive definite $G\in M_n$ such that
\begin{equation}
GA + A^*G = H \label{eqn:lyap}
\end{equation}
is positive definite. Furthermore, suppose there are Hermitian
matrices $G, H\in M_n$ that satisfy (\ref{eqn:lyap}), and suppose
$H$ is positive definite; then $A$ is positive stable if and only if
$G$ is positive definite.
\end{theorem}\noindent
The theorem gives a nice relation between the class of positive
definite matrices and the class of positive stable matrices.

\achapter{Exponential time differencing} \label{ch:appendix5} Stiff
systems of ODEs arise naturally when solving PDEs by spectral
methods, and their numerical solutions require special treatment if
accurate solutions are to be obtained efficiently. In this appendix
we describe one such class of solvers, the \emph{Exponential time
differencing} (ETD) schemes, they where first used in the field of
computational electrodynamics, where the problem of computing
electric and magnetic fields in a box typically result in a stiff
system of ODEs; see~\cite{Cox01} and references therein for further
details. In this discussion we concentrate on the Runge-Kutta
version of these schemes, in particular, we look at a modification
of the ETD fourth-order Runge-Kutta method presented by Kassam and
Trefethen~\cite{Kassam05}.

Let us represent our PDE in the following form
\begin{equation}\label{eqn:pde}
    u_t = \mathcal{L}u + \mathcal{N}(u,t),
\end{equation}
here $\mathcal{L}$ and $\mathcal{N}$ are linear and nonlinear
operators respectively. Applying a spatial discretisation yields the
following system of ODEs
\begin{equation}\label{eqn:fmodes}
    u_t = \textbf{L}u + \textbf{N}(u,t).
\end{equation}
To derive the ETD methods, we begin by multiplying
Eq.~(\ref{eqn:fmodes}) by the integrating factor $e^{-\textbf{L}t}$,
and integrating over one time step from $t = t_n$ to $t = t_{n+1} =
t_n + h$ to obtain
\begin{equation}\label{eqn:exact}
    u(t_{n+1}) = u(t_n)e^{\textbf{L}h}
    +e^{\textbf{L}h}\int_{0}^{h}e^{-\textbf{L}\tau}N(u(t_n+\tau),t_n+\tau)d\tau,
\end{equation}
this formula is exact, and the different ETD methods result from the
particular choice of approximation for the integral in
Eq.~(\ref{eqn:exact}). For example, if we assume that $\textbf{N}$
is constant, i.e. $\textbf{N} = \textbf{N}(u_n) + O(h)$, over a
single time step then we obtain the ETD1 method
\begin{equation}\label{eqn:etd1}
    u_{n+1} = u_ne^{\textbf{L}h} + \textbf{L}^{-1}(e^{\textbf{L}h}-1)\textbf{N}(u_n),
\end{equation}
which has a local truncation error $h^2\dot{\textbf{N}}/2$.
In~\cite{Cox01} Cox and Matthews present a host of recurrence
formulae that provide higher-order methods, as well as introducing
the set of methods based on Runge-Kutta time-stepping which they
name ETDRK schemes.

In our work we have used the fourth order scheme, known as ETDRK4,
the derivation of which is nonstandard according to Cox and
Matthews. For this reason we simply present the formulae and refer
the interested reader to~\cite{Cox01}
\begin{eqnarray}
  a_n &=& e^{\textbf{L}h/2}u_n + \textbf{L}^{-1}(e^{\textbf{L}h/2} - \textbf{I})\textbf{N}(u_n,t_n), \\
  b_n &=& e^{\textbf{L}h/2}u_n + \textbf{L}^{-1}(e^{\textbf{L}h/2} - \textbf{I})\textbf{N}(a_n,t_n+h/2), \\
  c_n &=& e^{\textbf{L}h/2}a_n + \textbf{L}^{-1}(e^{\textbf{L}h/2} - \textbf{I})(2\textbf{N}(b_n,t_n+h/2)
  -
  \textbf{N}(u_n,t_n)), \\
  u_n &=& e^{\textbf{L}h/2} + h^{-2}\textbf{L}^{-3}\{[-4 - \textbf{L}h
  + e^{\textbf{L}h}(4 - 3\textbf{L}h +
  (\textbf{L}h)^2)]\textbf{N}(u_n,t_n)\\ \nonumber
  && +2[2 + \textbf{L}h + e^{\textbf{L}h}(-2 + \textbf{L}h)](\textbf{N}(a_n,t_n+h/2) +
  \textbf{N}(b_n,t_n+h/2))\\
  \nonumber
  && +[-4 - 3\textbf{L}h - (\textbf{L}h)^2 +
  e^{\textbf{L}h}(4-\textbf{L}h)]\textbf{N}(c_n,t_n+h)\}.
\end{eqnarray}
Unfortunately, the stated method is prone to numerical instability,
the source of which can be understood by examining the following
expression
\begin{equation}\label{eqn:instab}
    f(z) = \frac{e^z-1}{z}.
\end{equation}
The accurate computation of this function is a well known problem in
numerical analysis, and is further discussed and referenced in the
paper by Kassam and Trefethen; the difficulties are born from
cancelation errors which arise for small $z$.

To understand why Eq.~(\ref{eqn:instab}) relates to the ETDRK4
method it is useful to examine the coefficients in the square
brackets of the update formula
\begin{eqnarray}
  \alpha &=& h^{-2}\textbf{L}^{-3}[-4 - \textbf{L}h + e^{\textbf{L}h}(4 - 3\textbf{L}h + (\textbf{L}h)^2)],\nonumber \\
  \beta  &=& h^{-2}\textbf{L}^{-3}[2 + \textbf{L}h + e^{\textbf{L}h}(-2 + \textbf{L}h)], \\
  \gamma &=& h^{-2}\textbf{L}^{-3}[-4 - 3\textbf{L}h - (\textbf{L}h)^2 +
  e^{\textbf{L}h}(4-\textbf{L}h)].\nonumber
  \label{eqn:coeff}
\end{eqnarray}
These coefficients are actually higher-order variants of
Eq.~(\ref{eqn:instab}), and thus, are susceptible to cancelation
errors. Indeed, all three terms suffer disastrous cancelation errors
when the matrix $\textbf{L}$ has eigenvalues close to zero. Cox and
Matthews knew of this problem in their work, and proposed a cutoff
point for small eigenvalues, whereafter the coefficients would be
computed by Taylor series. However, two problems arise: firstly,
there may exist a region where neither the formula nor the Taylor
series are accurate, and secondly, it is not obvious how to extend
these ideas in the case that $\textbf{L}$ is not diagonal.

In order to obtain numerical stability the accurate computation of
the coefficients (\ref{eqn:coeff}) is of paramount importance. With
this in mind Kassam and Trefethen suggest evaluating the
coefficients via an integral over a contour in the complex plane.
Recall from Cauchy's integral formula that a function $f(z)$ may be
computed as follows
\begin{equation}\label{eqn:cauchy}
    f(z) = \frac{1}{2\pi}\int_{\Gamma}\frac{f(t)}{t-z}dt,
\end{equation}
where $\Gamma$ is a contour that encloses $z$ and is separated from
$0$. The application of Cauchy's formula in the case that $z$ is a
matrix is straightforward
\begin{equation}\label{eqn:cauchymat}
    f(\textbf{L}) = \frac{1}{2\pi}\int_\Gamma f(t)(t\textbf{I} -
    \textbf{L})^{-1}dt,
\end{equation}
notice that the term $(t-z)^{-1}$ has been replaced by the resolvent
matrix, and that we now choose $\Gamma$ so that it encloses all the
eigenvalues of $\textbf{L}$.

When the above stabilisation procedure is applied to the ETDRK4
method the result is a stiff solver which works equally well,
whether the linear part of Eq.~(\ref{eqn:fmodes}) is diagonal or
not, it is extremely fast and accurate, and allows one to take large
time-steps. For example, in the computation of the finite-time
Lyapunov exponents of the KSE in Chapter \ref{ch:kse}, we where able
to use step-size $h = 0.25$, the same computation using the
eighth-order RK method due to Dormand and Prince \cite{HairerBook2}
needed to use $h = 2\times10^{-4}$.

\achapter{Numerical calculation of Lyapunov
spectra}\label{ch:appendix6} In this section we briefly remind
ourselves of the idea of Lyapunov exponents, before discussing an
algorithm first introduced by Benettin {\em et al} \cite{Benettin80}
in order to efficiently determine them; in Chapter \ref{ch:kse} we
use this algorithm to determine the largest Lyapunov exponents for
the KSE.

The concept of Lyapunov exponents is an important notion for
dynamical systems, particularly in applications, where it is often
used as a criteria for determining the existence of chaos. In the
following we give a description in the case of a discrete dynamical
system, $x_{k+1} = f(x_k)$. In that case the Lyapunov exponents give
a description of the average behaviour of the derivative, $Df$,
along the orbit of some initial point $x_0$, i.e. $\{x_0, f(x_0),
f^2(x), \dots\}$.

Write $E_0$ for the unit ball in $n$-dimensional phase space centred
at $x_0$, and define successive iterates by
\begin{equation}\label{eqn:ellipse}
    E_{k+1} = Df(x_k)E_k.
\end{equation}
Note that each $E_k$ is an ellipsoid. Now, let $\alpha_{j,k}$ denote
the length of the jth largest axis of $E_k$, then we define the jth
Lyapunov number of $f$ at $x_0$ to be
\begin{equation}\label{eqn:lyapnum}
    L_j = \lim_{k\to\infty}(\alpha_{j,k})^{1/k},
\end{equation}
when the limit exists. The {\em Lyapunov exponents} are then given
by the natural logarithms
\begin{equation}\label{eqn:lyapexp}
    \lambda_j = \ln{L_j}.
\end{equation}
The trajectory of the point $x_0$ is called a {\em chaotic
trajectory} if (1) the trajectory is bounded and is not asymptotic
to an equilibrium position of $f$, and (2) $f$ has at least one
positive Lyapunov exponent.

\subsubsection{Computation of Lyapunov exponents}
As noted an efficient algorithm for determining the Lyapunov
exponents of chaotic orbits has been put forward by Benettin {\em et
al} and is as follows: starting with an orthogonal set of unit
vectors $\{u_1, \dots, u_n\}$. Define
$$v_i = Df(x)u_i$$
where $i = 1,\dots,n$. Applying the Gram-Schmidt algorithm we
compute a set of orthogonal vectors, $w_i$ $(i = 1,\dots, n)$, such
that $w_1 = v_1$, and for $i>1$, the $w_i$ are defined inductively
as follows
\begin{equation}\label{eqn:orthog}
    w_i = v_i - \sum_{j=1}^{i-1}<v_i,\bar{u}_j>\bar{u}_j,
\end{equation}
here $\bar{u}_j = w_j/||w_j||$, where $||\cdot||$ denotes the
Euclidean norm, and $<\cdot,\cdot>$ denotes the corresponding inner
product.

The new set of vectors, $\bar{u}_j$, are approximations of the
directions of the $i$th axis of $E_{k+1}$, further, the approximate
ratio of the length of the $i$th axis of $E_{k+1}$ to that of $E_k$
is given by $r_i = ||w_i||$. To determine our approximation this
operation is repeated on each iterate of the map $f$. Writing
$r_{i,j}$ for the ratio $r_i$ at the $j$th iterate of $f$ we obtain
the following approximation to the $i$th Lyapunov exponent
\begin{equation}\label{eqn:lyapapprox}
    \lambda_i^k = \frac{1}{k}\sum_{j=1}^{k}\ln{r_{i,j}}.
\end{equation}
The approximation improves with increasing $k$.


\addcontentsline{toc}{chapter}{Bibliography}
\bibliographystyle{plain}
\bibliography{ThesisPhD}

\begin{thebibliography}{100}

\bibitem{Anosov67}
D.~V. Anosov.
\newblock Geodesic flows and closed {R}iemannian manifolds.
\newblock {\em Proc. Steklov Inst. Math.}, 90, 1967.

\bibitem{Benettin80}
G.~Benettin, L.~Galgani, A.~Giorgilli, and J.~M. Strelcyn.
\newblock Lyapunov characteristic exponents for smooth dynamical systems and
  for {H}amiltonian systems; a method for computing all of them. {P}art 2:
  Numerical application.
\newblock {\em Meccanica}, 15:21--29, 1980.

\bibitem{Biham89}
O.~Biham and W.~Wenzel.
\newblock Characterisation of unstable periodic orbits in chaotic attractors
  and repellers.
\newblock {\em Phys. Rev. Lett}, 63:819--822, 1989.

\bibitem{Biham90}
O.~Biham and W.~Wenzel.
\newblock Unstable periodic orbits and the symbolic dynamics of the complex
  {H}\'{e}non map.
\newblock {\em Phys.~Rev.~E}, 42(8):4639, 1990.

\bibitem{Bollt97}
E.~M. Bollt, Y.~C. Lai, and C.~Grebogi.
\newblock Coding, channel capacity and noise resistance in communicating with
  chaos.
\newblock {\em Phys. Rev. Lett}, 63:3787--3790, 1997.

\bibitem{Tassos02}
T.~C. Bountis, C.~Skokos, and M.~N. Vrahatis.
\newblock Computation and stability of periodic orbits of nonlinear mappings.
\newblock {\em $4^{th}$ GRACM Congress on Computational Mechanics}, 27-29 June
  2002.

\bibitem{Bowen75}
R.~Bowen.
\newblock $\omega$-limit sets for axiom-{A} diffeomorphisms.
\newblock {\em J. of Differential Equations}, 18, 1975.

\bibitem{Brouwer1910}
L.~E.~J. Brouwer.
\newblock Ueber eineindeutige, stetige {T}ransformationen von {F}l\"{a}chen in
  sich.
\newblock {\em Math. Ann.}, 69:176--180, 1910.

\bibitem{Christansen97}
F.~Christiansen, P.~Cvitanovi\'{c}, and V.~Putkaradze.
\newblock Spatiotemporal chaos in terms of unstable recurrent patterns.
\newblock {\em Nonlinearity}, 10:55, 1997.

\bibitem{Corless91}
R.~M. Corless, C.~Essex, and M.~A.~H. Nerenberg.
\newblock Numerical methods can suppress chaos.
\newblock {\em Phys.~Lett.~A}, 157(1):27--36, July 1991.

\bibitem{Coven88}
E.~M. Coven, I.~Kan, and J.~A. Yorke.
\newblock Pseudo-orbit shadowing in the family of tent maps.
\newblock {\em Trans.~Amer.~Math.~Soc.}, 308(1):227--241, July 1988.

\bibitem{Cox01}
S.~M. Cox and P.~C. Matthews.
\newblock Exponential time differencing for stiff systems.
\newblock {\em J.~Comput.~Phys.}, 176:430--455, 2002.

\bibitem{Crofts05}
J.~J. Crofts and R.~L. Davidchack.
\newblock {\em Proceedings of the International Conference on Numerical
  Analysis and Applied Mathematics 2005}, chapter Efficient detection of
  periodic orbits in high dimensional systems, pages 247--251.
\newblock WILEY-VCH, 2005.

\bibitem{Crofts06}
J.~J. Crofts and R.~L. Davidchack.
\newblock Efficient detection of periodic orbits in chaotic systems by
  stabilising transformations.
\newblock {\em SIAM J.~Sci.~Comput.}, 28(4):1275--1288, 2006.

\bibitem{Crofts07}
J.~J. Crofts and R.~L. Davidchack.
\newblock On the use of stabilising transformations for detecting unstable
  periodic orbits for the {K}uramoto-{S}ivashinsky equation.
\newblock {\em University of Leicester Technical Report}, (MA-07-007), 2007.

\bibitem{Cvitanovic91}
P.~Cvitanovi\'{c}.
\newblock Periodic orbits as the skeleton of classical and quantum chaos.
\newblock {\em Physica~D}, 51:138--151, 1991.

\bibitem{CvitanovicBook}
P.~Cvitanovi\'{c}, R.~Artuso, R.~Mainieri, G.~Tanner, and G.~Vattay.
\newblock {\em {C}haos: {C}lassical and {Q}uantum, {\tt
  http://www.nbi.dk/CHAOSBOOK/}}.
\newblock 2003.

\bibitem{cvitanovic88}
P.~Cvitanovi\'{c}, G.~H. Gunaratne, and I.~Procaccia.
\newblock Topological and metric properties of {H}\'{e}non-type strange
  attractors.
\newblock {\em Physical Review A}, 38(3):1503--1520, 1988.

\bibitem{Davidchack99c}
R.~L. Davidchack and Y.~C. Lai.
\newblock Efficient algorithm for detecting unstable periodic orbits in chaotic
  systems.
\newblock {\em Phys.~Rev.~E}, 60:6172--6175, 1999.

\bibitem{Davidchack00a}
R.~L. Davidchack, Y.~C. Lai, E.~M. Bollt, and M.~Dhamala.
\newblock Estimating generating partitions of chaotic systems by unstable
  periodic orbits.
\newblock {\em Phys.~Rev.~E}, 61:1353--1356, 2000.

\bibitem{Davidchack01b}
R.~L. Davidchack, Y.~C. Lai, A.~Klebanoff, and E.~M. Bollt.
\newblock Towards a complete detection of unstable periodic orbits.
\newblock {\em Phys.~Lett.~A}, 287:99--104, 2001.

\bibitem{Dawson94}
S.~Dawson, C.~Grebogi, T.~Sauer, and J.~A. Yorke.
\newblock Obstructions to shadowing when a {L}yapunov exponent fluctuates about
  zero.
\newblock {\em Phys.~Rev.~Lett}, 73(14):1927--1930, 1994.

\bibitem{DevaneyBook}
R.~L. Devaney.
\newblock {\em {A}n {I}ntroduction to {C}haotic {D}ynamical {S}ystems}.
\newblock Addison-Wesley, 2nd edition, 1989.

\bibitem{Do04}
Y.~Do and Y.~Lai.
\newblock Statistics of shadowing times in nonhyperbolic chaotic systems with
  unstable dimension variability.
\newblock {\em Phys.~Rev.~E}, 69:0162213, 2004.

\bibitem{Doyon05}
B.~Doyon and L.~J. Dube.
\newblock On {J}acobain matrices for flows.
\newblock {\em Chaos}, 15(1):13108, 2005.

\bibitem{Farantos98}
S.~C. Farantos.
\newblock {POMULT}: {A} program for computing periodic orbits in {H}amiltonian
  systems based on multiple shooting algorithms.
\newblock {\em Comp.~Phys.~Comm.}, 108:240--258, 1998.

\bibitem{Foias85}
C.~Foias, B.~Nicolaenko, G.~R. Sell, and R.~Temam.
\newblock Inertial manifold for the {K}uramoto-{S}ivashinsky equation.
\newblock {\em Acad. Sci. I-Math}, 301(6), 1985.

\bibitem{Galias01}
Z.~Gallias.
\newblock Interval methods for rigorous investigations of periodic orbits.
\newblock {\em Int. J. of Bifurcation and Chaos}, 11:2427--2450, 2001.

\bibitem{Galias02}
Z.~Gallias.
\newblock Rigorous investigation of the {I}keda map by means of interval
  arithmetic.
\newblock {\em Nonlinearity}, 15:1759--1779, 2002.

\bibitem{Galias}
Z.~Gallias and P.~Zgliczy\'{n}ski.
\newblock An interval method for finding periodic orbits of infinite
  dimensional discrete dynamical systems.
\newblock submitted to SIAM J.~Num.~Anal.

\bibitem{Grassberger89}
P.~Grassberger, H.~Kantz, and U.~Moenig.
\newblock On the symbolic dynamics of the {H}\'{e}non map.
\newblock {\em J.~Phys.~A}, 22:5217--5230, 1989.

\bibitem{Grebogi86}
C.~Grebogi, E.~Kostelich, E.~Ott, and J.~A. Yorke.
\newblock Multi-dimensioned intertwined basin boundaries: basin structures of
  the kicked double rotor map.
\newblock {\em Physica~D}, 25:347--360, 1986.

\bibitem{Grebogi87}
C.~Grebogi, E.~Ott, and J.~A. Yorke.
\newblock Chaos, strange attractors and fractal basin boundaries in nonlinear
  dynamics.
\newblock {\em Science}, 238:585, 1987.

\bibitem{Grebogi88}
C.~Grebogi, E.~Ott, and J.~A. Yorke.
\newblock Unstable periodic orbits and the dimensions of multifractal chaotic
  attractors.
\newblock {\em Phys.~Rev.~A}, 37(5):1711--1724, 1988.

\bibitem{GutzwillerBook}
M.~C. Gutzwiller.
\newblock {\em {C}haos in {C}lassical and {Q}uantum {M}echanics}.
\newblock Springer, 1st edition edition, 1990.

\bibitem{HairerBook}
E.~Hairer and G.~Wanner.
\newblock {\em {S}olving {O}rdinary {D}ifferential {E}quations {I}: {N}onStiff
  {P}roblems}.
\newblock Springer--Verlag, 2nd edition, 1993.

\bibitem{HairerBook2}
E.~Hairer and G.~Wanner.
\newblock {\em {S}olving {O}rdinary {D}ifferential {E}quations {II}: {S}tiff
  and {D}ifferential-{A}lgebraic {P}roblems}.
\newblock Springer--Verlag, 2nd edition, 1996.

\bibitem{Hammel85}
S.~M. Hammel, C.~K.~R.~T Jones, and J.~Maloney.
\newblock Global dynamical behaviour of the optical field in a cavity ring.
\newblock {\em J.~Opt.~Soc.~Am.~B}, 2:552, 1985.

\bibitem{Hammel88}
S.~M. Hammel, J.~A. Yorke, and C.~Grebogi.
\newblock Do numerical orbits of chaotic dynamical processes represent true
  orbits?
\newblock {\em J.~Complexity}, 3:136--145, 1987.

\bibitem{Hansen92}
K.~T. Hansen.
\newblock Remarks on the symbolic dynamics of the {H}\'{e}non map.
\newblock {\em Phys.~Lett.~A}, 165(2):100--104, 1992.

\bibitem{Hansen95}
K.~T. Hansen.
\newblock Alternative method to find orbits in chaotic systems.
\newblock {\em Phys.~Rev.~E}, 52:2388--2391, 1995.

\bibitem{Helleman80}
R.~H.~G. Helleman.
\newblock Self-generated chaotic behavior in nonlinear dynamics.
\newblock {\em Fundamental problems in statistical mechanics}, 5, 1980.

\bibitem{Henon76}
M.~H\'{e}non.
\newblock A two dimensional mapping with a strange attractor.
\newblock {\em Comm.~Math.~Phys.}, 50:69, 1976.

\bibitem{ODEPACK}
A.~C. Hindmarsh.
\newblock {ODEPACK}, {A} systemized collection of {ODE} solvers.
\newblock {\em Sci.~Comput.}, 1983.

\bibitem{HornBook1}
R.~A. Horn and C.~R. Johnson.
\newblock {\em {M}atrix {A}nalysis}.
\newblock Cambridge University Press, 1st edition, 1985.

\bibitem{HornBook2}
R.~A. Horn and C.~R. Johnson.
\newblock {\em {T}opics in {M}atrix {A}nalysis}.
\newblock Cambridge University Press, 1st edition, 1994.

\bibitem{HouseholderBook}
A.~S. Householder.
\newblock {\em {T}he {N}umerical {T}reatment of a {S}ingle {N}onlinear
  {E}quation}.
\newblock McGraw-Hill, New York, 1970.

\bibitem{HumphreysBook}
J.~E. Humphreys.
\newblock {\em {R}eflection {G}roups and {C}oexetor {G}roups}.
\newblock Cambridge University Press, 1st edition, 1990.

\bibitem{Hyman86}
J.~M. Hyman and B.~Nicolaenko.
\newblock The {K}uramoto-{S}ivashinsky equation: a bridge between {PDE}s and
  dynamical systems.
\newblock {\em Physica~D}, 18:113--126, 1986.

\bibitem{Ikeda79}
K.~Ikeda.
\newblock Multiple-valued stationary state and its instability of the
  transmitted light by a ring cavity system.
\newblock {\em Opt.~Commun.}, 30(2), 1979.

\bibitem{Kassam05}
A.~Kassam and L.~N. Trefethen.
\newblock Fourth order time stepping for stiff {PDE}s.
\newblock {\em SIAM J.~Sci.~Comput.}, 26(4):1214--1233, 2005.

\bibitem{Kida01}
G.~Kawahara and S.~Kida.
\newblock Periodic motion embedded in plane {C}ouette turbulence: regeneration
  cycle and burst.
\newblock {\em J.~Fluid~Mech.}, 449:291--300, 2001.

\bibitem{KelleyBook}
C.~T. Kelley.
\newblock {\em Solving Nonlinear Equations with Newton's Method}.
\newblock Fundamentals of Algorithms. SIAM, 1st edition, 2003.

\bibitem{Kevrekidis90}
I.~G. Kevrekidis, B.~Nicolaenko, and J.~Scovel.
\newblock Back in the saddle again: a computer assisted study of the
  {K}uramoto-{S}ivashinsky equation.
\newblock {\em SIAM J.~Appl.~Math}, 50:760, 1990.

\bibitem{Klebanoff01}
A.~Klebanoff and E.~M. Bollt.
\newblock Convergence analysis of {Davidchack} and {Lai's} algorithm for
  finding periodic orbits.
\newblock {\em Chaos, Solitons and Fractals}, 12:1305--1322, 2001.

\bibitem{Kuramoto76}
Y.~Kuramoto and T.~Tsuzuki.
\newblock Persistent propagation of concentration waves in dissipative media
  far from equilibrium.
\newblock {\em Prog.~Theor.~Phys.}, 55:365, 1976.

\bibitem{Lan04}
Y.~Lan and P.~Cvitanovi\'{c}.
\newblock Variational method for finding periodic orbits in a general flow.
\newblock {\em Phys.~Rev.~E}, 69:016217, 2004.

\bibitem{Lathrop89}
D.~P. Lathrop and E.~J. Kostelich.
\newblock Characterisation of an experimental attractor by periodic orbits.
\newblock {\em Phys.~Rev.~A}, 40, 1989.

\bibitem{Levenberg44}
K.~Levenberg.
\newblock A method for the solution of certain problems in least squares.
\newblock {\em Quart.~Appl.~Math}, 2:164--1687, 1944.

\bibitem{Lopez05}
V.~Lopez, P.~Boyland, M.~T. Heath, and R.~D. Moser.
\newblock Relative periodic solutions of the complex {G}inzburg-{L}andau
  equation.
\newblock {\em SIAM J.~Appl.~Dynamical~Systems}, 4(4):1042--1075, 2005.

\bibitem{Lorenz63}
E.~N. Lorenz.
\newblock Deterministic nonperiodic flow.
\newblock {\em J.~Atmos.~Sci}, 20:130--141, 1963.

\bibitem{Lust98}
K.~Lust, D.~Roose, A.~Spence, and A.~R. Champneys.
\newblock An adaptive {N}ewton-{P}icard algorithm with subspace iteration for
  computing periodic solutions.
\newblock {\em SIAM J.~Sci.~Comput.}, 19:1188--1209, 1998.

\bibitem{Marquardt63}
D.~Marquardt.
\newblock An algorithm for least-squares estimation of nonlinear parameters.
\newblock {\em SIAM J.~Appl.~Math}, 11:431--441, 1963.

\bibitem{Miller00}
J.~R. Miller and J.~A. Yorke.
\newblock Finding all periodic orbits of maps using {N}ewton methods: sizes of
  basins.
\newblock {\em Physica~D}, 135:195--211, 2000.

\bibitem{Moore79}
R.~E. Moore.
\newblock {\em {M}ethods and {A}pplications of {I}nterval {A}nalysis}.
\newblock Philadelphia:SIAM, 1979.

\bibitem{Minpack}
J.~J. Mor\'{e}, B.~S. Gorbow, and K.~E. Hillstrom.
\newblock User guide for {\tt minpack-1}.
\newblock {\em Technical report ANL-80-74, Argonne national laboratory, {\tt
  http://www.netlib.org/minpack/}}, 1980.

\bibitem{Newton}
I.~Newton.
\newblock De methodis fluxionum et serierum in finitorum, {L}ondon.
\newblock English translation J.~Colson, 1736.

\bibitem{Ott81}
E.~Ott.
\newblock Strange attractors and chaotic motions of dynamical systems.
\newblock {\em Rev.~Mod.~Phys.}, 53:655, 1981.

\bibitem{OttBook}
E.~Ott.
\newblock {\em {C}haos in {D}ynamical {S}ystems}.
\newblock Cambridge University Press, Cambridge, UK, 1st edition, 1993.

\bibitem{Ott90}
E.~Ott, C.~Grebogi, and J.~A. Yorke.
\newblock Controlling chaos.
\newblock {\em Phys.~Rev.~Lett}, 64:1196--1199, 1990.

\bibitem{Parsopoulos02}
K.~E. Parsopoulos and M.~N. Vrahatis.
\newblock Computing periodic orbits of nonlinear mappings through particle
  swarm optimization.
\newblock {\em 4th GRACM Congress on Computational Mechanics}, 27--29 June
  2002.

\bibitem{Pingel04}
D.~Pingel, P.~Schmelcher, and F.~K. Diakonos.
\newblock Stability transformation: a tool to solve nonlinear problems.
\newblock {\em Phys.~Rep.}, 400:67--148, 2004.

\bibitem{Plumecoq00a}
J.~Plumecoq and M.~Lefranc.
\newblock From template analysis to generating partitions {I}: Periodic orbits,
  knots and symbolic encodings.
\newblock {\em Physica~D}, 144:231--258, 2000.

\bibitem{poincare}
H.~Poincar\'{e}.
\newblock {\em New {M}ethods in {C}elestial {M}echanics}, volume~13 of {\em
  History of modern physics and astronomy}.
\newblock Springer-Verlag, New York, 1992.

\bibitem{Politi92}
A.~Politi and A.~Toricini.
\newblock Towards a statistical mechanics of spatiotemporal chaos.
\newblock {\em Phys.~Rev.~Lett.}, 69:3421--3424, 1992.

\bibitem{Powell70}
M.~J.~D. Powell.
\newblock {\em Numerical methods for nonlinear algebraic equations}, chapter A
  hybrid method for nonlinear equations, pages 87--114.
\newblock London, 1970.

\bibitem{NumericalRecipes}
W.~H. Press, S.~A. Teukolsky, W.~T.Vetterling, and B.~P. Flannery.
\newblock {\em Numerical Recipes in C, The Art of Scientific Computing}.
\newblock Cambridge University Press, 1992.

\bibitem{Robinson95}
J.~C. Robinson.
\newblock Finite dimensional behaviour in dissipative partial differential
  equations.
\newblock {\em Chaos}, 5:330--345, 1995.

\bibitem{RobinsonBook2}
J.~C. Robinson.
\newblock {\em {F}rom {F}inite to {I}nfinite {D}imensional {D}ynamical
  {S}ystems: {P}roceedings of the {NATO} {A}dvanced {S}tudy {I}nstitute,
  {C}ambridge, {UK}, 21 {A}ugust 1995}.
\newblock NATO Science Series \textrm{2}: Mathematics, Physics and Chemistry.
  Kluwer Academic Publishers, 1st edition, May 2001.

\bibitem{Romeiras92}
F.~J. Romeiras, C.~Grebogi, E.~Ott, and W.~P. Dayawansa.
\newblock Controlling chaotic dynamical systems.
\newblock {\em Physica~D}, 58:165--192, 1992.

\bibitem{Sauer97}
T.~Sauer, C.~Grebogi, and J.~A. Yorke.
\newblock How long do chaotic solutions remain valid?
\newblock {\em Phys.~Rev.~Lett.}, 79(1):59--62, July 1997.

\bibitem{Sauer02}
T.~D. Sauer.
\newblock Shadowing breakdown and large errors in dynamical simulations of
  physical systems.
\newblock {\em Phys.~Rev.~E}, 65:036220, 2002.

\bibitem{Schmelcher97}
P.~Schmelcher and F.~K. Diakonos.
\newblock Detecting unstable periodic orbits.
\newblock {\em Phys.~Rev.~Lett}, 78:4733--4736, 1997.

\bibitem{Schmelcher98}
P.~Schmelcher and F.~K. Diakonos.
\newblock General approach to the localisation of unstable periodic orbits in
  chaotic dynamical systems.
\newblock {\em Phys.~Rev.~E}, 57:2739--2746, 1998.

\bibitem{Schmelcher00}
P.~Schmelcher, F.~K. Diakonos, and O.~Biham.
\newblock Theory and application of the systematic detection of unstable
  periodic orbits in dynamical systems.
\newblock {\em Physical Review E}, 62(2), 2000.

\bibitem{Controlbook99}
H.~G. Schuster, editor.
\newblock {\em Handbook of {C}haos {C}ontrol}.
\newblock WILEY-VCH, 1st edition, 1999.

\bibitem{Gear79}
L.~F. Shampine and C.~W. Gear.
\newblock A user's view of solving stiff ordinary differential equations.
\newblock {\em SIAM Review}, 21(1):1--17, Jan. 1979.

\bibitem{Shroff93}
G.~M. Shroff and H.~B. Keller.
\newblock Stabilisation of unstable procedures: the recursive projection
  method.
\newblock {\em SIAM J.~Numer.~Anal.}, 30:1099--1120, 1993.

\bibitem{Simo}
C.~Sim\'{o}.
\newblock New families of solutions in {$N$}-body problems.
\newblock {\em In Proc. ECM 2000, Barcolona}, July 10--14 2000.

\bibitem{Sivashinsky77}
G.~I. Sivashinsky.
\newblock Nonlinear analysis of hydrodynamic instability in laminar flames.
\newblock {\em Acta Astron.}, 4:1177, 1977.

\bibitem{Smale67}
S.~Smale.
\newblock Differentiable dynamical systems.
\newblock {\em Bull.~Amer.~Math.~Soc.}, 73:747, 1967.

\bibitem{StoerBook}
J.~Stoer and R.~Bulirsch.
\newblock {\em Introduction to {N}umerical {A}nalysis}.
\newblock Number~12 in Texts in applied mathematics. Springer-Verlag, 2nd
  edition, 1993.

\bibitem{TemamBook}
R.~Temam.
\newblock {\em {I}nfinite-{D}imensional {D}ynamical {S}ystems in {M}echanics
  and {P}hysics}.
\newblock Number~68 in Appl.~Math.~Sci. NewYork:Springer-Verlag, 1st edition,
  1988.

\bibitem{TrefethenBook}
L.~N. Trefethen.
\newblock {\em Spectral {M}ethods in {M}atlab}.
\newblock SIAM, 1st edition, 2000.

\bibitem{Viswanath06}
D.~Viswanath.
\newblock Recurrent motions within plane {C}ouette turbulence.
\newblock {\em \tt arXiv:physics/0604062 v1}.

\bibitem{Vrahatis95}
M.~N. Vrahatis.
\newblock An efficient method for locating and computing periodic orbits for
  nonlinear mappings.
\newblock {\em J.~Comp.~Phys.}, 119(1):105--119, 1995.

\bibitem{Kerswell04}
H.~Wedin and R.~R. Kerswell.
\newblock Exact coherent structures in pipe flow: traveling wave solutions.
\newblock {\em Journal of Fluid Mechanics}, 508:333--371, 2004.

\bibitem{Yamaguti81}
M.~Yamaguti and S.~Ushiki.
\newblock Chaos in numerical analysis of ordinary differential equations.
\newblock {\em Physica~D}, 3(3):618--626, 1981.

\bibitem{Zgliczynski01}
P.~Zgliczy\'{n}ski and K.~Mischaikow.
\newblock Rigorous numerics for partial differential equations: the
  {K}uramoto-{S}ivashinsky equation.
\newblock {\em Foundations of Comput. Math.}, 1:255--288, 2001.

\bibitem{Zoldi98}
S.~M. Zoldi and H.~S. Greenside.
\newblock Spatially localised unstable periodic orbits of a high dimensional
  chaotic system.
\newblock {\em Phys.~Rev.~E}, 57:2511--2514, 1998.

\end{thebibliography}

\end{document}